\documentclass[11pt,letterpaper]{article}
\usepackage[utf8x]{inputenc}
\usepackage{float}
\usepackage{quoting}
\usepackage{mdframed}
\usepackage{dsfont}
\usepackage{thm-restate}
\quotingsetup{vskip=0pt,leftmargin=20pt,rightmargin=20pt}

\def\showauthornotes{0}
 
\def\showkeys{0}
\def\showdraftbox{0}

\def\usemicrotype{1}
\def\showfixme{0}

\usepackage{etex}

\usepackage[l2tabu, orthodox]{nag}

\usepackage{xspace,enumerate}

\usepackage[dvipsnames]{xcolor}

\usepackage[T1]{fontenc}
\usepackage[full]{textcomp}

\usepackage[american]{babel}

\usepackage{mathtools}

\usepackage{amsthm}

\newtheorem{theorem}{Theorem}[section]
\newtheorem*{theorem*}{Theorem}

\newtheorem*{proposition*}{Proposition}
\newtheorem{lemma}[theorem]{Lemma}
\newtheorem*{lemma*}{Lemma}
\newtheorem{corollary}[theorem]{Corollary}
\newtheorem*{conjecture*}{Conjecture}
\newtheorem{fact}[theorem]{Fact}
\newtheorem*{fact*}{Fact}

\newtheorem*{hypothesis*}{Hypothesis}

\newtheorem{claim}[theorem]{Claim}
\newtheorem*{claim*}{Claim}

\theoremstyle{definition}
\newtheorem{definition}[theorem]{Definition}

\newtheorem{notation}[theorem]{Notation}

\theoremstyle{remark}

\newtheorem*{remark*}{Remark}

\newtheorem*{observation*}{Observation}

\usepackage[
letterpaper,
top=0.7in,
bottom=0.9in,
left=1in,
right=1in]{geometry}

\usepackage[varg]{pxfonts}
\usepackage{textcomp}
\usepackage[scr=rsfso]{mathalfa}%
\usepackage{bm} %
\let\mathbb\varmathbb

\ifnum\showkeys=1
\usepackage[color]{showkeys}
\fi

\usepackage{hyperref}
\hypersetup{
pagebackref,
colorlinks=true,
urlcolor=blue,
linkcolor=blue,
citecolor=OliveGreen,
}
\usepackage{prettyref}

\newcommand{\savehyperref}[2]{\texorpdfstring{\hyperref[#1]{#2}}{#2}}

\newrefformat{eq}{\savehyperref{#1}{\textup{(\ref*{#1})}}}
\newrefformat{lem}{\savehyperref{#1}{Lemma~\ref*{#1}}}
\newrefformat{def}{\savehyperref{#1}{Definition~\ref*{#1}}}
\newrefformat{thm}{\savehyperref{#1}{Theorem~\ref*{#1}}}
\newrefformat{cor}{\savehyperref{#1}{Corollary~\ref*{#1}}}
\newrefformat{cha}{\savehyperref{#1}{Chapter~\ref*{#1}}}
\newrefformat{sec}{\savehyperref{#1}{Section~\ref*{#1}}}
\newrefformat{app}{\savehyperref{#1}{Appendix~\ref*{#1}}}
\newrefformat{tab}{\savehyperref{#1}{Table~\ref*{#1}}}
\newrefformat{fig}{\savehyperref{#1}{Figure~\ref*{#1}}}
\newrefformat{hyp}{\savehyperref{#1}{Hypothesis~\ref*{#1}}}
\newrefformat{alg}{\savehyperref{#1}{Algorithm~\ref*{#1}}}
\newrefformat{rem}{\savehyperref{#1}{Remark~\ref*{#1}}}
\newrefformat{item}{\savehyperref{#1}{Item~\ref*{#1}}}
\newrefformat{step}{\savehyperref{#1}{step~\ref*{#1}}}
\newrefformat{conj}{\savehyperref{#1}{Conjecture~\ref*{#1}}}
\newrefformat{fact}{\savehyperref{#1}{Fact~\ref*{#1}}}
\newrefformat{prop}{\savehyperref{#1}{Proposition~\ref*{#1}}}
\newrefformat{prob}{\savehyperref{#1}{Problem~\ref*{#1}}}
\newrefformat{claim}{\savehyperref{#1}{Claim~\ref*{#1}}}
\newrefformat{relax}{\savehyperref{#1}{Relaxation~\ref*{#1}}}
\newrefformat{red}{\savehyperref{#1}{Reduction~\ref*{#1}}}
\newrefformat{part}{\savehyperref{#1}{Part~\ref*{#1}}}

\newcommand{\Sref}[1]{\hyperref[#1]{\S\ref*{#1}}}

\usepackage{nicefrac}

\ifnum\usemicrotype=1
\usepackage{microtype}
\fi

\ifnum\showauthornotes=1
\newcommand{\Authornote}[2]{{\sffamily\small\color{red}{[#1: #2]}}}
\newcommand{\Authornotecolored}[3]{{\sffamily\small\color{#1}{[#2: #3]}}}
\newcommand{\Authorcomment}[2]{{\sffamily\small\color{gray}{[#1: #2]}}}
\newcommand{\Authorstartcomment}[1]{\sffamily\small\color{gray}[#1: }

\newcommand{\Authorfnote}[2]{\footnote{\color{red}{#1: #2}}}
\newcommand{\Authorfixme}[1]{\Authornote{#1}{\textbf{??}}}
\newcommand{\Authormarginmark}[1]{\marginpar{\textcolor{red}{\fbox{\Large #1:!}}}}
\else
\newcommand{\Authornote}[2]{}
\newcommand{\Authornotecolored}[3]{}
\newcommand{\Authorcomment}[2]{}
\newcommand{\Authorstartcomment}[1]{}

\newcommand{\Authorfnote}[2]{}
\newcommand{\Authorfixme}[1]{}
\newcommand{\Authormarginmark}[1]{}
\fi

\ifnum\showfixme=0

\fi

\usepackage{boxedminipage}

\newcommand{\Paren}[1]{\left(#1\right)}

\newcommand{\abs}[1]{\lvert#1\rvert}

\newcommand{\norm}[1]{\lVert#1\rVert}

\newcommand{\defeq}{\stackrel{\mbox{\normalfont\tiny def}}=}

\newcommand\bdot\bullet

\ifx\mathds\undefined %

\else

\fi

\DeclareMathOperator{\Tr}{Tr}

\DeclareMathOperator{\poly}{poly}

\DeclareMathOperator{\polylog}{polylog}

\newcommand{\C}{\mathbb C}

\newcommand{\cA}{\mathcal A}
\newcommand{\cB}{\mathcal B}
\newcommand{\cC}{\mathcal C}
\newcommand{\cD}{\mathcal D}
\newcommand{\cE}{\mathcal E}

\newcommand{\cL}{\mathcal L}

\newcommand{\cO}{\mathcal O}
\newcommand{\cP}{\mathcal P}

\newcommand{\cS}{\mathcal S}
\newcommand{\cT}{\mathcal T}

\newcommand{\cV}{\mathcal V}

\renewcommand{\leq}{\leqslant}

\renewcommand{\geq}{\geqslant}

\ifnum\showdraftbox=1

\else

\fi

\let\epsilon=\varepsilon

\numberwithin{equation}{section}

\newcommand\MYcurrentlabel{xxx}

\newcommand{\MYstore}[2]{%
  \global\expandafter \def \csname MYMEMORY #1 \endcsname{#2}%
}

\newcommand{\MYload}[1]{%
  \csname MYMEMORY #1 \endcsname%
}

\newcommand{\MYnewlabel}[1]{%
  \renewcommand\MYcurrentlabel{#1}%
  \MYoldlabel{#1}%
}

\newcommand{\MYdummylabel}[1]{}

\newcommand{\torestate}[1]{%
  \let\MYoldlabel\label%
  \let\label\MYnewlabel%
  #1%
  \MYstore{\MYcurrentlabel}{#1}%
  \let\label\MYoldlabel%
}

\newcommand{\restatetheorem}[1]{%
  \let\MYoldlabel\label
  \let\label\MYdummylabel
  \begin{theorem*}[Restatement of \prettyref{#1}]
    \MYload{#1}
  \end{theorem*}
  \let\label\MYoldlabel
}

\newcommand{\restatelemma}[1]{%
  \let\MYoldlabel\label
  \let\label\MYdummylabel
  \begin{lemma*}[Restatement of \prettyref{#1}]
    \MYload{#1}
  \end{lemma*}
  \let\label\MYoldlabel
}

\newcommand{\restateprop}[1]{%
  \let\MYoldlabel\label
  \let\label\MYdummylabel
  \begin{proposition*}[Restatement of \prettyref{#1}]
    \MYload{#1}
  \end{proposition*}
  \let\label\MYoldlabel
}

\newcommand{\restatefact}[1]{%
  \let\MYoldlabel\label
  \let\label\MYdummylabel
  \begin{fact*}[Restatement of \prettyref{#1}]
    \MYload{#1}
  \end{fact*}
  \let\label\MYoldlabel
}

\newcommand{\restate}[1]{%
  \let\MYoldlabel\label
  \let\label\MYdummylabel
  \MYload{#1}
  \let\label\MYoldlabel
}

\newcommand{\e}{\epsilon}
\newcommand{\eps}{\epsilon}

\let\origparagraph\paragraph
\renewcommand{\paragraph}[1]{\origparagraph{#1.}}
\allowdisplaybreaks
\usepackage{paralist}

\usepackage{comment}

\let\citet\cite
\theoremstyle{definition}

\usepackage{relsize}
\usepackage[font=footnotesize]{caption}
\usepackage{yfonts}

\usepackage{appendix}
\usepackage{amssymb}
\usepackage{amsfonts}
\usepackage{latexsym}
\usepackage{amsmath}
\usepackage{cleveref}

\usepackage{subcaption}

\usepackage{qcircuit}			

\usepackage{tikz}
\usetikzlibrary{shapes.geometric}

\hypersetup{pagebackref, colorlinks=true, urlcolor=blue,
  linkcolor=blue, citecolor=magenta}

\newcommand{\Id}{\mathop{\mathds{1}}\!\mathinner{}}

\newcommand{\wt}[1]{\widetilde{#1}}
\newcommand{\comp}[1]{\overline{#1}}

\newcommand{\ket}[1]{|#1\rangle}
\newcommand{\bra}[1]{\langle#1|}

\newcommand{\ip}[2]{\langle #1 | #2 \rangle}
\newcommand{\ketbra}[2]{|#1\rangle\! \langle #2|}

\newcommand{\sC}{{\mathsf{C}}}

\newcommand{\sF}{{\mathsf{F}}}

\newcommand{\sR}{{\mathsf{R}}}
\newcommand{\sS}{{\mathsf{S}}}
\newcommand{\sT}{{\mathsf{T}}}

\newcommand{\timeconfig}{{\bm{\tau}}}

\begin{document}

\title{Good approximate quantum LDPC codes\\ from spacetime circuit Hamiltonians}%
\makeatletter
\renewcommand\@date{{%
  \vspace{-\baselineskip}%
  \large\centering
  \begin{tabular}{@{}c@{}}
     Thomas C. Bohdanowicz \\
    \footnotesize Caltech \\ \footnotesize \href{mailto:thom@caltech.edu}{thom@caltech.edu}
  \end{tabular}%
  \hspace{1em}
  \begin{tabular}{@{}c@{}}
      Elizabeth Crosson\\
    \footnotesize University of New Mexico \\ \footnotesize \href{mailto:crosson@unm.edu}{crosson@unm.edu}
  \end{tabular}
  \hspace{1em}
  \begin{tabular}{@{}c@{}}
     Chinmay Nirkhe \\
     \footnotesize UC Berkeley \\
    \footnotesize \href{nirkhe@cs.berkeley.edu}{nirkhe@cs.berkeley.edu}
  \end{tabular}
    \hspace{1em}
  \begin{tabular}{@{}c@{}}
     Henry Yuen \\
     \footnotesize University of Toronto \\
    \footnotesize \href{mailto:hyuen@cs.toronto.edu}{hyuen@cs.toronto.edu}
  \end{tabular}

}}
\makeatother
\maketitle \pagenumbering{gobble}
\vspace{-10pt}
\begin{abstract}

We study \emph{approximate} quantum low-density parity-check (QLDPC) codes, which are approximate quantum error-correcting codes specified as the ground space of a frustration-free local Hamiltonian, whose terms do not necessarily commute. 
Such codes generalize stabilizer QLDPC codes, which are exact quantum error-correcting codes with sparse, low-weight stabilizer generators (i.e. each stabilizer generator acts on a few qubits, and each qubit participates in a few stabilizer generators). Our investigation is motivated by an important question in Hamiltonian complexity and quantum coding theory: do stabilizer QLDPC codes with constant rate, linear distance, and constant-weight stabilizers exist?

We show that obtaining such optimal scaling of parameters (modulo polylogarithmic corrections) is possible if we go beyond stabilizer codes: we prove the existence of a family of $[[N,k,d,\varepsilon]]$ approximate QLDPC codes that encode $k = \wt{\Omega}(N)$ logical qubits into $N$ physical qubits with distance $d = \wt{\Omega}(N)$ and approximation infidelity $\varepsilon = \cO (1/\polylog(N))$. The code space is stabilized by a set of $10$-local \emph{noncommuting} projectors, with each physical qubit only participating in $\mathcal{O}(\polylog N)$ projectors.  We prove the existence of an efficient encoding map and show that the spectral gap of the code Hamiltonian scales as $\wt{\Omega}(N^{-3.09})$.  We also show that arbitrary Pauli errors can be locally detected by circuits of polylogarithmic depth.

Our family of approximate QLDPC codes is based on applying a recent connection between circuit Hamiltonians and approximate quantum codes (Nirkhe, et al., ICALP 2018) to a result showing that random Clifford circuits of polylogarithmic depth yield asymptotically good quantum codes (Brown and Fawzi, ISIT 2013). Then, in order to obtain a code with sparse checks and strong detection of local errors, we use a \emph{spacetime} circuit Hamiltonian construction in order to take advantage of the parallelism of the Brown-Fawzi circuits. 

The analysis of the spectral gap of the code Hamiltonian is the main technical contribution of this work.  We show that for any depth $D$ quantum circuit on $n$ qubits there is an associated spacetime circuit-to-Hamiltonian construction with spectral gap $\Omega(n^{-3.09} D^{-2} \log^{-6}(n))$.  To lower bound this gap we use a Markov chain decomposition method to divide the state space of partially completed circuit configurations into overlapping subsets corresponding to uniform circuit segments of depth $\log n$, which are based on bitonic sorting circuits.  We use the combinatorial properties of these circuit configurations to show rapid mixing between the subsets, and within the subsets we develop a novel isomorphism between the local update Markov chain on bitonic circuit configurations and the edge-flip Markov chain on equal-area dyadic tilings, whose mixing time was recently shown to be polynomial (Cannon, Levin, and Stauffer, RANDOM 2017).  Previous lower bounds on the spectral gap of spacetime circuit Hamiltonians have all been based on a connection to exactly solvable quantum spin chains and applied only to 1+1 dimensional nearest-neighbor quantum circuits with at least linear depth.  

\end{abstract}

\tableofcontents

\clearpage
\pagebreak \pagenumbering{arabic} \setcounter{page}{1}

\section{Introduction}

A central result in the theory of classical error correcting codes is that there exist families of \emph{good} linear $[N,k,d]$ codes, which have linear dimension $k = \Omega(N)$, linear distance $d = \Omega(N)$, constant sparsity parity checks, and linear time encoding and decoding algorithms.  These \emph{low-density parity check} (LDPC) codes~\cite{Gallager63low-densityparity-check} have many theoretical as well as practical applications.

A grand challenge in quantum information theory is to construct a \emph{quantum} counterpart to classical LDPC codes with similarly optimal parameters. Traditionally this effort has focused on CSS stabilizer codes\footnote{The CSS construction~\cite{calderbank1996good,steane1996multiple} combines two classical codes, $\mathcal{C}_1 = [N,k_1,d_1]$ and $\mathcal{C}_2 = [N,k_2,d_2]$ to form an $[[N,k_1 + k_2 - N, \min(d_1,d_2)]]$ QECC with commuting check terms that generate a stabilizer subgroup of the Pauli group.}, where the notion of sparse parity checks corresponds to stabilizer generators that each act on $\mathcal{O}(1)$ physical qubits, with each qubit participating in only $\mathcal{O}(1)$ of such checks. The existence of QLDPC codes with good parameters and fast encoding/decoding algorithms would have significant practical impact; for example, Gottesman has shown these would imply schemes for fault tolerant quantum computation with constant overhead~\cite{gottesman2013fault}. 

Despite many years of investigation, we do not yet know of QLDPC codes that simultaneously achieve constant rate and relative distance while maintaining constant locality and sparsity. The QLDPC codes of~\cite{tillich2014quantum, lloyd2017polylog} have a constant rate, but the minimum distance does not exceed $\mathcal{O}(\sqrt{N})$ where $N$ is the number of physical qubits. So far the QLDPC code with the best distance scaling is the construction of Freedman, Meyers and Luo~\cite{freedman2002z2} which achieves minimum distance distance $\mathcal{O}(\sqrt{N \log N})$, but only encodes a single qubit.  Bravyi and Hastings gave a probabilistic construction of a code with constant rate and linear distance, but the stabilizer generators each act on $\sqrt{N}$ physical qubits~\cite{bravyi2014homological}. Hastings proved that, assuming a conjecture about high dimensional geometry, there exist QLDPC codes encoding a constant number of qubits (i.e. have vanishing rate) with distance scaling as $\Omega(N^{1 - \xi})$ for any $\xi > 0$~\cite{hastings2017quantum,hastings2017weight}.

The question of whether good QLDPC codes exist also has importance for Hamiltonian complexity and the construction of exotic models in physics.  This connection arises because any QECC code space that can be enforced by a set of constant-weight check operators can also be identified as the ground space of a local Hamiltonian.   A central goal in these areas is to identify classes of local Hamiltonians with robust entanglement properties, and QLDPC codes provide a fruitful source of candidates.  However, if the local terms are stabilizers then $H$ is always a commuting Hamiltonian, and despite the richness of these systems they only capture a subset of local Hamiltonians and the properties they can exhibit.%

Here we explore the QLDPC Conjecture (which posits that there exist asymptotically good QLDPC codes) through the correspondence between QLDPC codes and local Hamiltonians.  This leads us to relax the requirement of being a CSS stabilizer code in two ways:
\begin{enumerate}
  \item The code satisfies an \emph{approximate} error-correction property: after an error channel is applied the decoding procedure recovers encoded states up to some $1 - \varepsilon$ fidelity, where $\varepsilon = o(1)$.
  
  \item The codespace is specified as the groundspace of a frustration-free local Hamiltonian $H = \Pi_1 + \cdots + \Pi_m$, where the local projectors $\Pi_i$ don't necessarily commute. 
\end{enumerate}
Codes satisfying the approximate reovery condition are known as \emph{approximate quantum error correcting codes} (AQECC), and codes with noncommuting frustration-free local check terms have been considered as a generalization of QLDPC in Hamiltonian complexity, therefore we call codes satisfying satisfying these conditions \emph{approximate QLDPC codes.}%

\subsection{Our results}

Our main result is a construction of approximate QLDPC codes with nearly-optimal parameters. 

\begin{theorem}%
\label{thm:informalMain}
For infinitely many $N$ there exists $N$-qubit subspaces $\{\cC_N\}$ with the following properties:
\begin{enumerate}
  \item $\cC_N$ is an AQECC that encodes $k = \tilde{\Omega}(N)$ logical qubits in $N$ physical qubits, has distance $d = \tilde{\Omega}(N)$, approximation error $\varepsilon = \cO(1/\polylog N)$, and a $\poly(N)$ time encoding algorithm.
  \item $\cC_N$ is the ground space of a frustration-free local Hamiltonian $H^{(N)} = \sum H_i^{(N)}$ such that each term $H_i^{(N)}$ acts on $O(1)$ qubits, and each physical qubit participates in at most $\polylog N$ terms. 
  \item The Hamiltonian $H^{(N)}$ has spectral gap $\tilde{\Omega}(N^{-3.09})$ and it is spatially local in $\polylog(N)$ dimensions (i.e. it can be embedded in $\mathbb{R}^{\polylog N}$ with finite qubit density and geometrically local interactions). 
\end{enumerate}
Here, the notation $\tilde{\Omega}(\cdot)$ suppresses factors of $\polylog N$.
\end{theorem} 
The fact that the local check terms do not commute means that it is impossible to measure them all simultaneously.  However, in Section \ref{sec:localDetect} we show that any Pauli error will increase the energy of at least one local check term by at least $1/\polylog(N)$, and we use this to show that this family of codes is capable of \emph{locally} detecting arbitrary Pauli errors with $\polylog(N)$ depth circuits.
\begin{theorem}
\label{thm:local_detection}
  For the family of codes described above, there exists with high probabilty a collection $\cD$ of $\polylog(N)$-local projectors satisfying the following properties:
  \begin{enumerate}
    \item Each projector $\Pi \in \cD$ acts on $10$ physical qubits in the code and $s = \polylog(N)$ ancilla qubits initialized in the $\ket{0}$ state, and $\Pi \ket{\psi} \ket{0^s} = 0$ for all $\Pi \in \cD$ if and only if $\ket{\psi} \in \mathcal{C}^N$.
    \item For all Pauli channels $\cE$, for all codewords $\ket{\psi} \in \cC$, there exists a projector $\Pi \in \cD$ such that
    \begin{equation}
      \Tr \Paren{\Pi \, \Paren{\cE(\psi) \otimes \ketbra{0^s}{0^s}} } \geq (1 - \alpha)(1 - 2^{-\polylog(N)})
    \end{equation}
    where $\psi = \ketbra{\psi}{\psi}$ and $\alpha$ is the total weight of the channel $\cE$ on the (nonlocal) Pauli stabilizers in $\cS$. %
    \end{enumerate}
    Furthermore, there exists a measurement $M$, implementable by a circuit of $\polylog(N)$ depth acting on $\mathcal{O}(N \polylog(N))$ qubits, such that for all Pauli channels $\cE$ and for all codewords $\ket{\psi} \in \cC$
    \begin{equation}
      \Tr \Paren{ M \, \Paren{ \cE(\psi) \otimes \ketbra{0^{Ns}}{0^{Ns}} }} \geq (1 - \alpha) (1 - 2^{-\polylog(N)}).
    \end{equation}
\end{theorem}
Our construction of this family of codes is based on a recently discovered connection between AQECC and Feynman-Kitaev (FK) Hamiltonians~\cite{nirkhe_et_al:LIPIcs:2018:9095}.  FK Hamiltonians have ground states of the form $\frac{1}{\sqrt{T+1}} \sum_{t = 0}^T |t\rangle |\psi_t\rangle$, where $|\psi_t\rangle = U_t ... U_1 |0^n\rangle$\footnote{We use $n$ for the number of input qubits in a circuit Hamiltonian, and $N$ for the number of physical qubits in our code construction. $N = n\polylog(n)$ in our construction because of the overhead used to represent the clock.} is the state of a quantum circuit at time $t$, and are used to prove the quantum version of the Cook-Levin theorem.   The connection to AQECC is based on mapping the encoding circuit of a QECC to the ground space of a local Hamiltonian.  To construct the family of codes in Theorem \ref{thm:informalMain} we apply the connection formed in~\cite{nirkhe_et_al:LIPIcs:2018:9095} to a randomized construction of good quantum codes with polylogarithmic depth encoding circuits~\cite{brown2013short}.   The polylogarithmic factors in our construction arises from the additional ``clock'' qubits that are used in this mapping from circuits to  ground states.  However, the standard FK construction uses a single global clock variable and does not allow for gates to be applied in parallel; to take full advantage of these parallel encoding circuits we present a substantial new technical analysis of the many-clock ``spacetime''~\cite{mizel2007simple, breuckmann2014space} version of the FK construction that assigns an independent clock variable $t_i$ to each qubit $i$ in the circuit\footnote{The term ``spacetime'' comes from relativistic physics, in which time is necessarily measured by local clocks.}.  

The spacetime circuit Hamiltonian enforces a ground state that is a uniform superposition over all valid configurations of these clocks (where validity is determined by the pattern of gates in the circuit), and it is unitarily equivalent to the normalized Laplacian of a random walk on the high-dimensional space of partially completed circuit configurations.  Spacetime circuit Hamiltonians have been used previously for universal adiabatic computation and QMA-completeness constructions that are spatially local on a square lattice and do not require perturbative gadgets~\cite{breuckmann2014space,PhysRevLett.114.140501,lloyd2016adiabatic}.  The analysis of the spectral gap in these previous works has always relied on the exact solutions to certain 1 + 1 dimensional quantum spin chains~\cite{koma1997spectral}.  Here we develop a nearly tight lower bound on the spectral gap of the spacetime circuit Hamiltonian for a particular uniform class of circuits based on bitonic sorting networks.   These sorting networks are used to transform a $D$ depth circuit with arbitrary connectivity and $n$ qubits into a depth $D \log(n)^2$ circuit\footnote{All logarithms in this work are base 2.} with spatially local connectivity in $\log(n)$ dimensions.    By analyzing these sorting networks we prove the following general theorem in Section \ref{sec:analysis}.
\begin{theorem}\label{thm:spacetimeGeneralResult}
For any depth $D$ quantum circuit of 2-local gates on $n$ qubits, where $n$ is a power of 2, there is an associated spacetime circuit-to-Hamiltonian construction which is spatially local in $\textrm{polylog}(n)$ dimensions and has a spectral gap that is $\Omega(n^{-3.09} D^{-2} \log^{-6}(n))$.
\end{theorem}

The spectral gap of a code Hamiltonian lower bounds the soundness of the code, since it determines the minimum energy of states outside of the code space.  In our code construction we take $D = \polylog(n)$, and since the circuit Hamiltonian acts on a total of $N = n \polylog(n)$ qubits this accounts for the bound on the spectral gap in Theorem \ref{thm:informalMain}.  Since our proof holds for any circuit with arbitrary connectivity we state the general result here for future potential applications to QMA and universal adiabatic computation.  
\subsection{Discussion}

We believe that our approximate QLDPC codes, beyond being an attempt to address the QLDPC Conjecture via a different perspective, also illustrate a compelling synthesis of various intriguing concepts of quantum information theory, and furthermore, highlight several connections that deserve closer investigation.

\paragraph{Approximate quantum error correction} AQECCs generalize QECCs by only requiring that the quantum information stored in the code, after the action of an error channel, be recoverable with fidelity at least $1 -\eps$.  AQECCs have long been known to be capable of achieving better parameters than standard QECCs~\cite{leung1997approximate, crepeau2005approximate}, though the necessary and sufficient conditions for approximate recovery were only established within the last decade~\cite{beny2010general}.   AQECC have found applications to fault-tolerant quantum computation~\cite{bravyi2010topological,lee2017topological} through the analysis of realistic perturbations to exact QECC, and have recently experienced a resurgence in popularity in physics due to connections made with the holographic correspondence in quantum gravity~\cite{almheiri2015bulk}.  Recently~\cite{flammia2017limits} have considered a version of local AQECC which also includes the possibility of locally approximate correction of errors in order to investigate the ultimate limits of the storage of quantum information in space.   One can interpret our approximate QLDPC codes as providing another demonstration that the AQECC condition is a useful relaxation that facilitates the construction of codes with superior parameters than what is (known to be) achievable in the standard QECC framework.
\paragraph{Codes from local Hamiltonians} As previously mentioned, QLDPC codes have been a fruitful source of local Hamiltonians with robust entanglement properties, which are central objects of study in quantum Hamiltonian complexity and condensed matter theory.  The first example of a QLDPC code was Kitaev's toric code, which is also a canonical example of a topologically ordered phase of matter~\cite{kitaev2003fault}.  Most research on QECC has been focused on stabilizer codes, like the toric code, for which the associated code Hamiltonians are commuting and frustration-free. In this paper we proceed in the opposite direction by asking: what kinds of quantum codes can we construct from local Hamiltonians whose terms don't necessarily commute? With this perspective, the extensive toolbox of techniques for constructing and analyzing Hamiltonians in quantum computing and quantum physics becomes immediately useful. This approach is inspired by several recent papers:
\begin{enumerate}
  \item In~\cite{eldar2016need}, Eldar, et al. defined general QLDPC codes to be subspaces $S$ that are stabilized by a collection of local projectors $\{ \Pi_i \}$; in other words, $\Pi_i \ket{\psi} = 0$ for all $i$ if and only if $\ket{\psi} \in S$. They call the $\Pi_i$ projectors ``parity checks'' in analogy to the parity check terms of CSS codes; however, the projectors $\{\Pi_i \}$ need not be parity checks in the traditional sense.
  
  \item In~\cite{flammia2017limits}, Flammia, et al. formalized a notion of local AQECCs that includes an additional condition of approximate local correctability.  This notion was applied to derive bounds on the ultimate limits of the storage of quantum information in spatially local codes. %

\item In~\cite{brandao2017quantum}, Brandao et al. show that qutrit systems on a line with nearest-neighbor interactions can form approximate QLDPC that encode $\log(N)$ qubits with distance $\log(N)$, and also show that AQECC can appear generically in energy subspaces of local Hamiltonians.

  \item In~\cite{nirkhe_et_al:LIPIcs:2018:9095}, Nirkhe, et al. shows that by using the Feynman-Kitaev circuit-to-Hamiltonian construction and a \emph{non-local} CSS code, one can obtain a \emph{local} approximate QECC where the corresponding Hamiltonian's ground space is approximately the original CSS code. 
\end{enumerate}

Although there are still many hurdles to climb before codes with noncommuting checks can be realistically applied to fault-tolerance protocols, these recent developments form an exciting frontier 
in the study of local Hamiltonians.  Another example of this connection is that the approximate codes developed in~\cite{nirkhe_et_al:LIPIcs:2018:9095} and extended here can be seen as an instance of the recently formalized notion of Hamiltonian sparsification~\cite{aharonov2018gap}.

\paragraph{Comparison with the sparse subsystem codes of~\cite{bacon2017sparse}}

In~\cite{bacon2017sparse} Bacon et al. construct subsystem codes with distance $\Omega(N^{1-\xi})$ for $\xi = \mathcal{O}(1/\sqrt{\log N})$ and constant weight gauge generators, and these were termed ``sparse subsystem codes.''  These are the best parameters achieved to date for any exact QECC in the ground space of a local Hamiltonian.  Even more remarkable, in relation to the present work, is the fact that the codes of Bacon et al. have local checks that arise in a completely different way from quantum circuits.  The difference is that~\cite{bacon2017sparse} considers fault-tolerant circuit gadgets (instead of encoding circuits as in~\cite{nirkhe_et_al:LIPIcs:2018:9095} and this work) and enforces the correct operation of these Clifford circuits according to the Gottesman-Knill theorem (rather than FK circuit Hamiltonians).  

Another difference between these code constructions is that the code Hamiltonians of Bacon et al. are necessarily frustrated due to the fact that the noncommuting gauge generators are all Pauli operators, which therefore anticommute and share no simultaneous eigenstates.  Although frustration does not always preclude the possibility of local error correction~\cite{flammia2017limits}, there is no lower bound established on the spectral gap of the codes in~\cite{bacon2017sparse} (and so there may be states outside the codespace with exponentially small energy), and detecting an error on a single qubit requires measuring $\poly(N)$ gauge generators in order to ascertain the syndromes of nonlocal stabilizers.   With this understanding we summarize past results on QECC with strong parameters: %

\begin{figure}[H]
\begin{center}
    \begin{tabular}{ | l | l | l | l | p{5cm} | }
    \hline
    \textbf{Reference} & \textbf{\# of logical qubits} & \textbf{Distance} & \textbf{Locality} & \textbf{Notes} \\ \hline
    \cite{tillich2014quantum} & $\Theta(N)$ & $\Theta(\sqrt{N})$ & $O(1)$ & CSS Stabilizer code \\ \hline
    \cite{freedman2002z2} & $O(1)$ & $O(\sqrt{N \log N})$ & $O(1)$ & CSS Stabilizer code \\ \hline
    \cite{bravyi2014homological} & $\Theta(N)$ & $\Theta(N)$ & $\Omega(\sqrt{N})$ & CSS Stabilizer code \\ \hline
  \cite{hastings2017quantum,hastings2017weight} & $O(1)$ & $\Omega(N^{1 - \xi})$ for all $\xi > 0$ & $O(1)$ & CSS code, assumes conjecture in high dimensional geometry \\ \hline
  \cite{bacon2017sparse} & $O(N)$ & $\Omega(N^{1 - \xi})$ for all $\xi > 0$ & $O(1)$ & Subsystem Stabilizer code, frustrated Hamiltonian \\ \hline
  This paper & $\Omega(N/\polylog N)$ & $\Omega(N/\polylog N)$ & $O(1)$ & approximate QLDPC code \\ \hline
    \end{tabular}
\end{center}
\end{figure}

\paragraph{Connections with QPCP} 

On of the most significant open problems in Hamiltonian complexity is to resolve the quantum PCP conjecture~\cite{aharonov2013guest}, which posits that quantum proofs can be made probabilistically checkable.  Since local Hamiltonians and the complexity class QMA are the respective quantum generalizations of constraint satisfaction problems and NP, the QPCP conjecture is equivalent to the statement that it is QMA-complete to decide whether the ground state energy of a Hamiltonian $H = \sum_{i =1}^m H_i$ is less than $a$ or greater than $b$ (under the promise that one of these is the case), where $b - a > \frac{c}{(m\cdot \max_i \|H_i\|)}$ for some $c = \Omega(1)$ corresponds to constant relative precision.  One reason this question is difficult is any trivial state which is output by a constant-depth quantum circuit acting on a product state can be given as an NP witness, and many of the commonly studied classes of local Hamiltonians necessarily have low-energy trivial states.  Therefore in order for QPCP to hold there must be some Hamiltonian with no low-energy trivial states, and even this weaker NLTS conjecture~\cite{hastings2013trivial} remains an open problem.   

One approach to resolving the NLTS and QPCP conjectures is to develop the quantum analogue of locally testable codes, which are defined in~\cite{aharonov2015quantum} as codes with frustration-free but not necessarily commuting local checks, good parameters, and a \emph{soundness} property which states that the energy of a state with respect to the constraints grows linearly with its distance from the code space.  Therefore constructing good QLDPC is necessary for constructing QLTC, but it is not sufficient since in general QLDPC may have low energy states outside the code space.  This collection of open challenges that are stimulating innovations in Hamiltonian complexity is known as the robust entanglement zoo~\cite{eldar2017local}, since they all involve generalizing known properties of quantum ground states to states with constant relative distance above the code space.

Just as the classical PCP Theorem indirectly transforms a Cook-Levin computational tableau into a probabilistically checkable CSP, a QPCP construction could be seen as transforming the FK circuit-to-Hamiltonian construction into a local Hamiltonian with robust entanglement.   While known limitations on generalized FK constructions make such a direct approach unlikey~\cite{ganti2013gap,Bausch2018analysislimitations,gonzalez2018history}, our Theorem \ref{thm:local_detection} on local error detection in $\polylog(N)$ depth is the first result to quantitatively substantiate the belief that the spacetime Hamiltonian construction is more robust than the standard global-clock FK Hamiltonian.  Specifically, we show that the energy of a state after the application of a Pauli error channel is inversely proportional to the \emph{depth} of the circuit in the spacetime construction, whereas it is proportional to the \emph{size} of the circuit in the standard FK construction.  In fact in Section \ref{sec:globalFK} we describe an alternate version of our approximate QLDPC construction that is based on global-clock FK and a modified distribution over time steps of the quantum circuit, and this version can achieve any scaling of the approximation error $\varepsilon(N) > 0$ at the expense of decreasing the spectral gap to $\tilde{\Omega}(\varepsilon N^{-3})$, but this substantially weakens the corresponding version of Theorem \ref{thm:local_detection} and forces the local error detection circuits to have superlinear depth.  This results suggest that continued investigation into alternative circuit-to-Hamiltonian constructions might be a fruitful direction of research, and might possibly make headway towards the mystery of the QPCP conjecture.

\paragraph{Overview of the remaining sections}  Section \ref{sec:codeHamSketch} overviews the spacetime circuit Hamiltonian used in our construction, and Section \ref{sec:sketch} sketches the proof techniques we use to lower bound the spectral gap of the code Hamiltonian, which is the main technical contribution of this work.   Section \ref{sec:prelim} formally defines approximate QLDPC codes and develops the machinery needed to describe our construction, including the good codes with polylogarithmic depth encoding circuits due to Brown and Fawzi~\cite{brown2013short} in Section \ref{subsec:parallel}, spacetime circuit Hamiltonians in Section \ref{subsec:spacetimeconstruction}, and bitonic sorting networks in Section \ref{subsection:bitonicblock}.  Our code construction and the efficient encoding circuit are given in Section \ref{sec:construction}, and the analysis of the spectral gap result in Theorems \ref{thm:informalMain} and \ref{thm:spacetimeGeneralResult} is given in Section \ref{sec:analysis}.   The local error detection analysis underlying Theorem \ref{thm:local_detection} is given in Section \ref{sec:localDetect}, and finally we discuss alternate versions of the construction in Section \ref{sec:globalFK} and a spatially local embedding in Section \ref{sec:spatialLocality}.  Appendix \label{appendix:bitonic} contains many detailed results on combinatorial properties of partially completed circuit configurations of bitonic sorting networks, as well as the connection between these circuit configurations and dyadic tilings.

\subsection{Description of the code Hamiltonian}\label{sec:codeHamSketch}
In~\cite{nirkhe_et_al:LIPIcs:2018:9095} it was recognized that the FK Hamiltonian which maps circuits to ground states could be used to develop a set of local checks for AQECC for which only an efficient encoding circuit was previously been found.  For a circuit with local gates $U_1,...,U_T$ the FK ground states are
\begin{equation}
  \ket{\Psi} = \frac{1}{\sqrt{T+1}} \sum_{t = 0}^T \ket{t}_{\sC} \otimes (U_t U_{t-1} \cdots U_1) \ket{\psi, 0 \ldots 0 }_{\sS}.\label{eq:historystate}
\end{equation}
Such states are called \emph{history states}.  The register $\sC$, called the \emph{clock register}, indicates how many gates have been applied to the all zeroes state, which is stored in register $\sS$ (called the \emph{state register}) containing an initial state $\ket{\psi}$ and ancillas. 

Although this state has only a $1/(T+1)$ fidelity with the output of the circuit, the standard technique for increasing the overlap to be inverse polynomially close to 1 is to pad the end of the circuit with identity gates (for recent work on more efficient methods for biasing the history state towards its endpoints, see~\cite{Bausch2018analysislimitations, caha2018clocks}).  This technique allows history states to capture approximate versions of QECC that have efficient encoding circuits.  The approximation error of the code is directly related to history state overlap with the output of the encoding circuit.  

\begin{figure}[h!]
\begin{center}
\includegraphics[scale=.3]{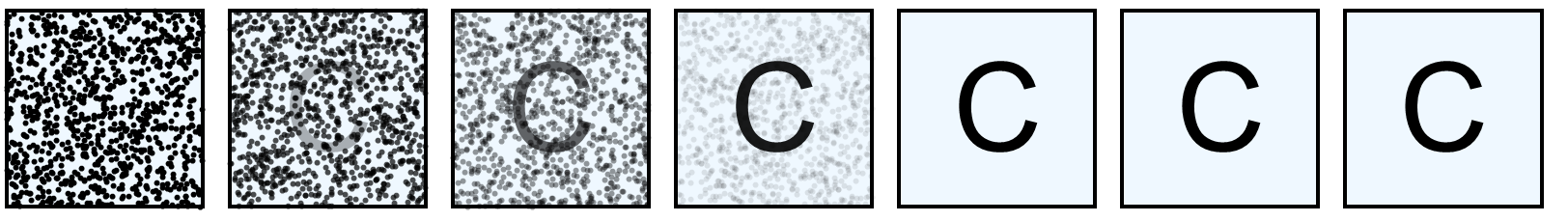}
\end{center}
\caption{The approximate nature of the codes introduced in~\cite{nirkhe_et_al:LIPIcs:2018:9095} arises from the fact that part of the history state superposition corresponding to early time steps, which do not match the output of the encoding circuit and are treated as noise in our analysis.  Once a sufficient depth to form a codeword is reached, the computation can be padded with identity gates in order to increase the overlap of this approximate codeword with the original codeword it is approximating.}
\end{figure}

The Hamiltonian which enforces the ground space spanned by states of the form \eqref{eq:historystate} is formed by projectors that check the input state of the computation, as well as \emph{propagation terms} that check that the branch of the superposition corresponding to time $t$ and the branch corresponding to time $t+1$ differ by the application of the gate $U_{t+1}$ to the state register. The linear ordering of the computation $U_1,\ldots,U_T$ is enforced via the sum of these propagation terms.  The propagation Hamiltonian is unitarily equivalent to a normalized Laplacian on the path graph with vertices $\{0,...,T\}$ and therefore has a spectral gap that is $\Theta(T^{-2})$.  For the purpose of lower bounding the energy of excitations that leave the code space, it is important to check the spectral gap of the full Hamiltonian including the input check terms, see Section \ref{sec:sketch} for further discussion.  

In this work we use the spacetime version of the FK circuit Hamiltonian~\cite{breuckmann2014space}, which assigns a clock register to each computational qubit, and has a ground space spanned by uniform superposition over all valid time configurations $\boldsymbol{\tau} = (t_1,...,t_n)$ of the state of the computation after the gates prior to $\boldsymbol{\tau}$ have been performed, 
\begin{equation}
  \ket{\psi} = \frac{1}{|\mathcal{T}|^{1/2}} \sum_{\boldsymbol{\tau} \in \mathcal{C}} \ket{\boldsymbol{\tau}}_{\sC} \otimes U(\boldsymbol{\tau} \leftarrow 0)\ket{0 \cdots 0}_{\sS}.
\end{equation}
Here $\mathcal{T}$ is the set of all valid time configurations $\boldsymbol{\tau}$, which is any vector $(t_1, \ldots, t_n)$ that the clock registers could hold if a subset of gates that respected causal dependence (see Definition \ref{def:partial-config}) were applied.  To avoid boundary effects at the beginning and end of the computation we use circular (periodic) time, which involves reversing the gates in the second half of the circuit so that the computation returns to its initial state.  In Section \ref{subsec:spacetimeconstruction} implement these periodic clocks using qubits.  

The necessity of including these causal constraints is one of the complications introduced by the use of spacetime circuit Hamiltonians, but a far more significant challenge is lower bounding the spectral gap of the spacetime propagation Hamiltonian.  In contrast with single-clock circuit Hamiltonians, the geometric arrangement of the gates in the circuit now has a significant effect on the spectrum of the spacetime circuit Hamiltonian due to the causal constraints.  All lower bounds in previous works apply to spacetime Hamiltonians in 2 spatial dimensions, which represent 1 (space) + 1 (time) dimensional quantum circuits.  This is not only due to the importance of planar connectivity for practical applications, but it is also a symptom of the general fact that exactly solvable models in mathematical physics are hardly known beyond 1 + 1 dimensions.   The 1 + 1 dimensional circuit propagation Hamiltonian is unitarily equivalent to a stochastic model describing the evolution of a string in the plane.  For higher dimensional circuits it corresponds to the dynamics of membranes or crystal surface growth, where no known solutions are available.  To overcome this in the present work we use sorting networks to turn arbitrary random circuits into circuits with uniform connectivity, and then we apply powerful techniques and past results from the theory of Markov chains to analyze the resulting high-dimensional spacetime circuit Hamiltonians.  
 
\subsection{Proof sketch for the spectral gap analysis}\label{sec:sketch}
Our analysis of the spectral gap $\Delta_\textrm{prop}$ of the spacetime circuit propagation Hamiltonian begins with the standard mapping from $H_\textrm{prop}$ to a a Markov chain transition matrix $P$. \footnote{The re-scaled Hamiltonian $H_{\textrm{prop}}/\|H_{\textrm{prop}}\|$ is unitarily equivalent to a normalized graph Laplacian $\mathcal{L}$ for the graph with vertices corresponding to valid time configurations and edges corresponding to local gate updates on those time configurations.  $P$ is the transition matrix for the random walk on this graph, which is obtained from $I-\mathcal{L}$ by a similarity transformation.  The point is that these mappings provide an algebraic relation between $\Delta_{\textrm{prop}}$ and $\Delta_P$.}  To analyze the latter, we apply a Markov chain decomposition method due to Madras and Randall~\cite{madras2002}, which is used to split the Markov chain and its state space into pieces that are easier to analyze individually.  For our decomposition of choice these pieces come in several closely related variants, which all essentially correspond to the set of time configurations contained within the final phase of a bitonic sorting circuit (as shown in Figure \ref{fig:bitonicBlockChain} for 8 lanes) which we call a bitonic block.  As described in Appendix \ref{appendix:bitonic}, an arbitrary circuit consisting of 2-local gates can be transformed into a sequence of consecutive bitonic blocks, with at most a polylogarithmic factor of blow up in the depth. 

\begin{figure}[h!]
\begin{center}
\includegraphics[scale=.5]{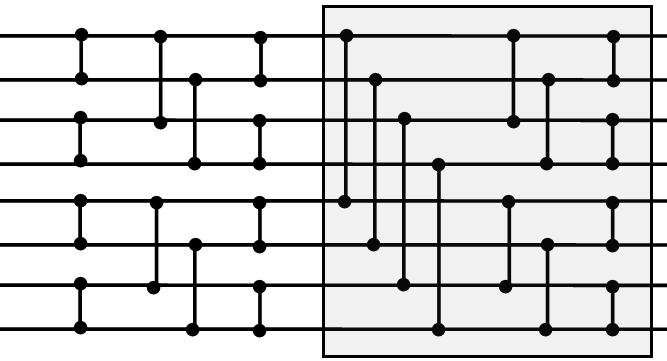}
\end{center}
\caption{A bitonic sorting architecture on $n = 8$ bits. We refer to the final phase of the architecture, corresponding to the last $\log(n) = 3$ layers enclosed in a gray box, as a bitonic block.  Note that the gates in each layer are executed simultaneously, but are drawn as non-overlapping for visual clarity.  An arbitrary circuit consisting of 2-local gates can be transformed to have the architecture of consecutive repetitions of bitonic blocks \iffalse (Lemma \ref{lem:permforspaciallylocal})\fi at the cost of increasing the depth by a factor of $\log(N)^2$.  \label{fig:bb3}}
\end{figure}
\begin{figure}[h!]

\begin{center}
\includegraphics[width=0.7\textwidth]{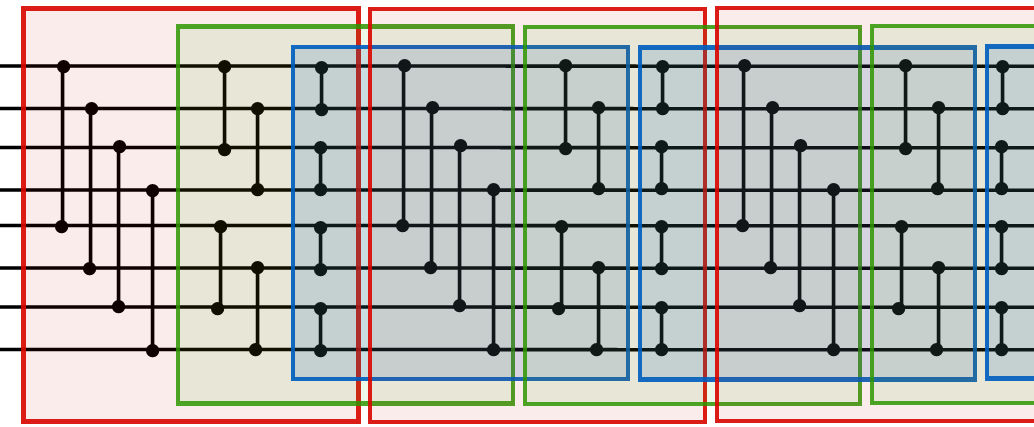}
\end{center}
\caption{The Markov chain block decomposition for a sequence of padded bitonic sorting architecture on 8 bits. The set of valid time configurations contained entirely within the $i$-th colored rectangle constitutes the block $\Omega_i$.  The set of time configurations in two rectangles of different colors are related by a permutation of the qubit wires.  The aggregate chain $\overline{P}$ has a nonzero transition probability $\overline{P}(i,j)$ iff the rectangles corresponding to the blocks $\Omega_i$ and $\Omega_j$ are overlapping.  Each block $\Omega_i$ has a nonzero transition probability to $\log N$ other blocks $\Omega_j$.   Every valid time configuration is contained in at least one of the blocks, and no time configuration is contained in more than $\log N$ blocks. \label{fig:bitonicBlockChain}}
\end{figure}
After dividing the set of valid time configurations $\Omega$ (the state space of the Markov chain) into subsets $\Omega_i$ of configurations confined to bitonic blocks of the form illustrated in Figure \ref{fig:bitonicBlockChain}, the subsets will form a quasi-linear chain in the sense that $\Omega_i$ and $\Omega_j$ have nonempty intersections when $|i - j| \leq \log n$.  To apply the decomposition method we need to analyze (1) the spectral gap of the restricted Markov chains $P_i$ that are confined to stay within each of the subsets $\Omega_i$, and (2) the spectral gap of an aggregate Markov chain $\overline{P}$ that moves between the blocks based on transition probabilities related to the size of the intersections of the blocks.  

As suggested by its quasi-linear connectivity, the spectral gap of the aggregate chain can be lower bounded using Cheeger's inequality in similar manner as is done for the path graph Laplacian.  The main technical challenge is to accurately compute the transition probabilities $\overline{P}(i,j) = \pi(\Omega_i \cap \Omega_j) / \left(\Theta \pi(\Omega_i)\right)$, which involve the ratio of the number of configurations within each of the blocks to the number within the pairwise intersections, $|\Omega_i \cap \Omega_j| / |\Omega_i|$, as well as the maximum number of blocks $\Theta$ that can contain any particular time configuration.  In Appendix \ref{appendix:bitonic}, we develop a recurrence relation to exactly count these configurations and show that the former is constant for consecutive blocks (and decays doubly exponentially with $|i - j|$ for longer distance transitions), and the latter is logarithmic in $n$.  Using asymptotic properties of the recurrence relation we show that the transition probabilities between $i,i+1$ are equal to $(\phi\log n)^{-1}$, where $\phi = (1 + \sqrt{5})/2$ is the golden ratio.  If there are $m$ blocks in total so that the length of the path is $m$, we use Cheeger's inequality to show that the spectral gap $\Delta_{\overline{P}}$ of the aggregate chain satisfies
\begin{equation}
\Delta_{\overline{P}} \geq \left(\phi m \log n\right)^{-2} . \label{eq:gapOverlineP}
\end{equation}

\begin{figure}[h!]
\begin{center}
\includegraphics[width=0.8\textwidth]{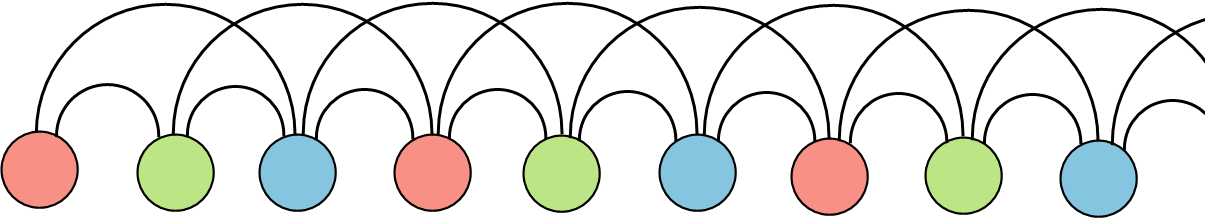}
\end{center}
\caption{An illustration of the states and transitions in the aggregate chain corresponding to the subsets of time configurations contained with the blocks in fig \ref{fig:bitonicBlockChain}. \label{fig:aggregateChain}}
\end{figure}

Turning to the analysis of the restricted chains $P_i$, we present the discovery of a surprising and beautiful connection between valid time configurations of architectures of the form shown in Figure \ref{fig:bb3} with combinatorial structures known as dyadic tilings~\cite{randomdyadictilingsoftheunitsquare}.  Dyadic tilings are tilings of the unit square by equal-area dyadic rectangles, which are rectangles of the form $[a 2^{-s}, (a+1)2^{-s}]\times [b 2^{-t},(b+1) 2^{-t}]$, where $a,b,s,t$ are nonnegative integers.  These tilings have a natural recursive characterization: beginning from the unit square, draw a line that is either a horizontal or vertical bisector.  This divides the square into two rectangles, and in each of these one chooses a horizontal or vertical bisector, and so on.  After $\ell = \log(n)$ such recursive steps one obtains a dyadic tiling of rank $\ell$ with a total of $n$ dyadic rectangles, each with area $1/n$. Some examples are given in Figure \ref{fig:dyadicex_intro}.
\begin{figure}[h!]
\begin{center}
\includegraphics[scale =.3]{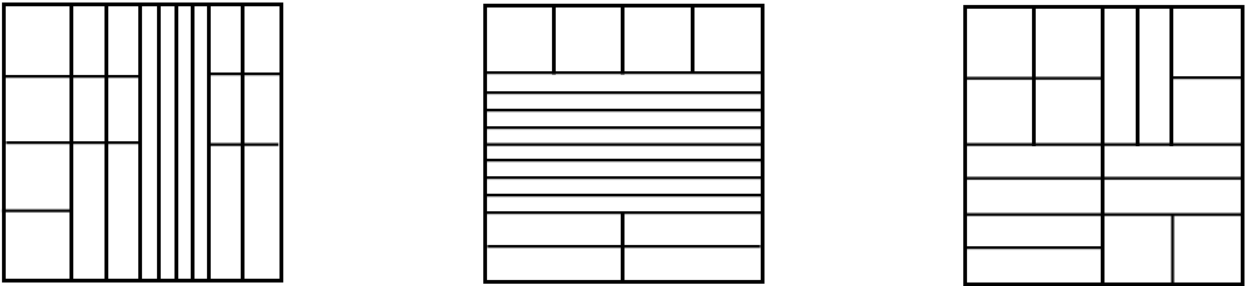}
\caption{Examples of dyadic tilings of rank 4. \label{fig:dyadicex_intro}}
\end{center}
\end{figure}

For a spacetime circuit with $n$ qubits, we choose the blocks $\Omega_i$ in the decomposition so that for each block there is an exact bijection between the time configurations within the block and the set of equal-area dyadic tilings of rank $\ell = \log n$.  Moreover, it turns out that the natural Markov chain on time configurations can also be mapped onto a previously defined Markov chain for dyadic tilings called the edge-flip chain.  This Markov chain selects a rectangle of area $1/n$ in the current dyadic tiling and one of its four edges at random, and flips this edge if the result would be another dyadic tiling. The correspondence is described in Figure \ref{fig:vdi2}.
\begin{figure}[h!]
\begin{center}
\includegraphics[scale=.3]{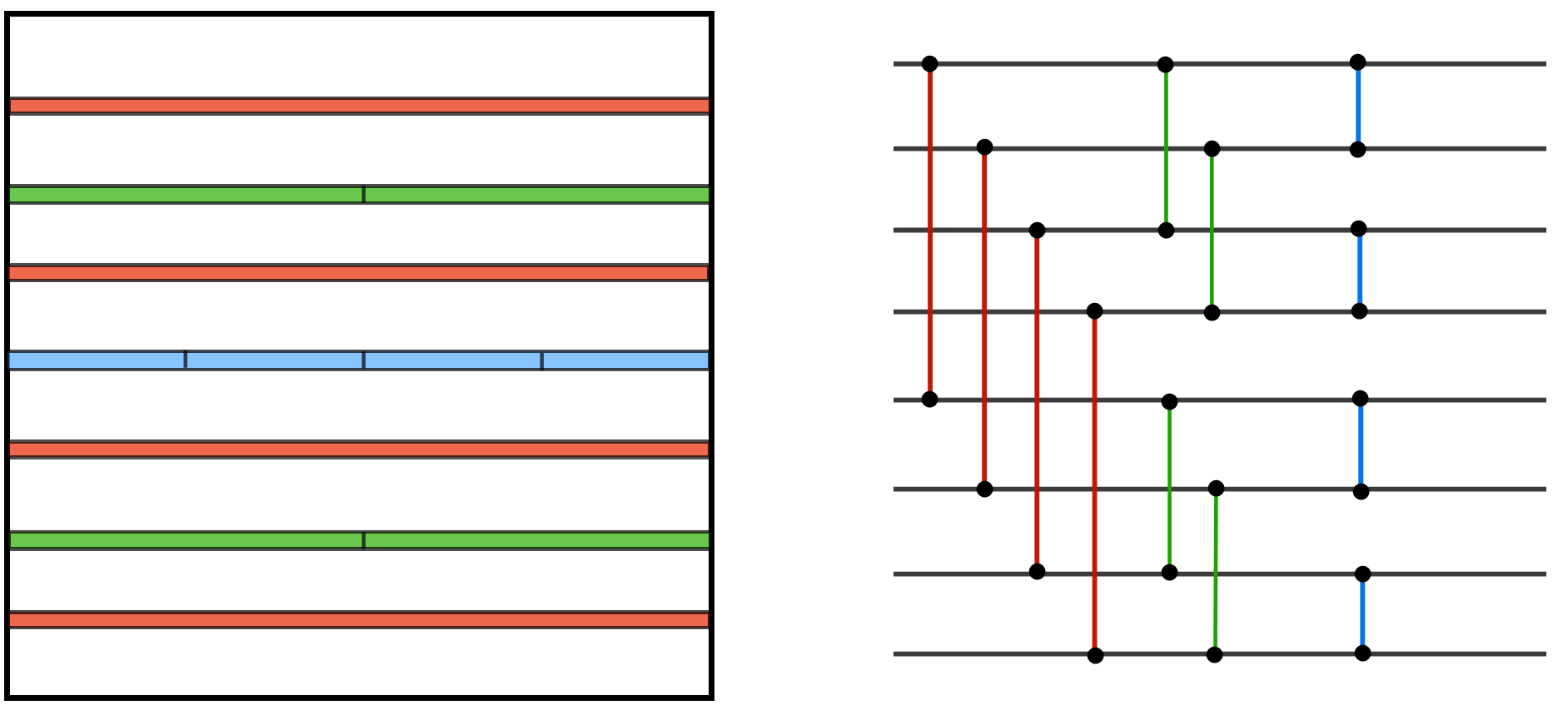}
\end{center}
\caption{A color-coding of the correspondence between dyadic tilings and valid time configurations of a bitonic sorting circuit.  The colored line segments in (a) correspond to sub-edges which when rotated by $\pi/2$ about their midpoint will be sub-edges of a vertical edge in some dyadic tiling.  These edges are placed in correspondence with the gates of the bitonic sorting circuit in (b), with the convention that colored line segments in (a) are ordered from left to right and from top to bottom, and the gates in a given commuting layer in (b) are enumerated from top to bottom.  Given an arbitrary dyadic tiling, one checks which of the colored line segments in (a) correspond to vertical sub-edges in the tiling, and these correspond to gates that are in the past causal cone of the bitonic time configuration associated with that tiling. \label{fig:vdi2}}
\end{figure}

The mixing time of this edge flip chain was an open problem for over a decade, but has recently been the subject of a tour de force analysis that establishes an upper bound on the mixing time that is polynomial in $n$.   Adapting these results using our bijection between these Markov chains yields
\begin{equation}
\Delta_{P_i} = \Omega\left(n^{-4.09} \right) \quad ,\quad \textrm{for all } \; \; i=1,...,m, \label{eq:gapPi}
\end{equation}
where the value of the exponent can be taken to be $\log (17) = 4.087\ldots$.
Once \eqref{eq:gapOverlineP} and \eqref{eq:gapPi} are established, we combine them according to the decomposition result,
$$
\Delta_P \geq \frac{1}{2} \Delta_{\overline{P}} \min_{i = 1,...,m} \Delta_{P_i} = \Omega\left(n^{-4.09} m^{-2} \textrm{polylog}(n)^{-1}\right),
$$
which is an inverse polynomial lower bound on the gap.  The circuit propagation Hamiltonian is equivalent to the Markov chain $P$ scaled by a factor of $n$, and so we obtain $\Delta_{\textrm{prop}} = \wt{\Omega}(n^{-3.09})$.  Finally, using the version of the spacetime Hamiltonian with circular time we show that every state in the code space has overlap $1/\polylog(n)$ with the input terms and so the geometrical lemma yields a gap of $\wt{\Omega}(n^{-3.09})$ for the full code Hamiltonian.

\section{Preliminaries}
\label{sec:prelim}

In what follows, we present the definitions of the main ingredients of our code construction and analysis.

\subsection{Approximate QLDPC codes}

\newcommand{\Enc}{\mathrm{Enc}}
\newcommand{\Rec}{\mathrm{Rec}}

Here we present the formal definition of an approximate QLDPC code.

\begin{definition}[Approximate QLDPC code]
  A $2^k$-dimensional subspace $C$ of $(\C^2)^{\otimes N}$ is a $[[N,k,d,\eps,\ell,s]]$ \emph{approximate QLDPC code} iff there exists a (not necessarily commuting) set of projectors $\{H_1,\ldots,H_m\}$ acting on $N$ qubits such that
  \begin{enumerate}
    \item Each term $H_i$ acts on at most $\ell$ qubits (i.e. \emph{locality}) and each qubit participates in at most $s$ terms (i.e. \emph{sparsity}).
    \item For all $\ket{\psi}$, we have that $\ket{\psi} \in C$ if and only if $\bra{\psi} H \ket{\psi} = 0$, where $H = H_1 + \cdots + H_m$. 
    \item There exist encoding and recovery maps $\Enc,\Rec$ such that for all $\ket{\phi} \in (\C^2)^{\otimes k} \otimes \mathcal{R}$ where $\mathcal{R}$ is some purifying register, for all completely positive trace preserving maps $\mathcal{E}$ acting on at most $(d-1)/2$ qubits, we have that the image of $\Enc$ is exactly the code $C$ and 
  \begin{equation}
    F \Paren{ \Rec \circ \cE \circ \Enc(\ketbra{\phi}{\phi}), \ketbra{\phi}{\phi} } \geq 1 - \eps
  \end{equation}
  where $F(\cdot,\cdot)$ denotes the fidelity function. Here, the maps $\Enc$, $\mathcal{E}$, and $\Rec$ do not act on register $\mathcal{R}$.
  \end{enumerate}
  \end{definition}
The first condition of the above definition enforces the locality and sparsity conditions of the approximate QLDPC code. The second condition enforces that the code is the ground space of a frustration-free local Hamiltonian. The third condition corresponds to the approximate error-correcting condition, where we only require that the decoded state is \emph{close} to the original state (i.e., we no longer insist that $\Rec \circ \cE \circ \Enc$ is exactly the identity channel). Although there are few results on approximate quantum error-correcting codes, we do know that relaxing the exact decoding condition yields codes with properties that cannot be achieved using exact codes~\cite{leung1997approximate,beny2010general}.

\subsection{Parallel quantum circuits}
\label{subsec:parallel}

We establish some notational conventions for parallel quantum circuits. 

Consider the following model of depth $D$ circuits on $n$ qubits. The circuit $C$ consists of $D$ layers $L_1,\ldots,L_D$. In each layer $L_t$ for $1 \leq t \leq D$, the $n$ qubits are partitioned into $n/2$ disjoint pairs $\{ (p,q) \}$, and a two-qubit gate $U_t(p,q)$ acts on the qubit pair $(p,q)$. Layer $L_1$ is applied first, then layer $L_2$, and so on. The unitary corresponding to circuit $C$ is
\begin{equation}
  \prod_{t=1}^D \bigotimes_{(p,q) \in L_t} U_t[p,q]
\end{equation}
where the product is written from right to left. In other words, the unitary $\bigotimes_{(p_1,q_1) \in L_1} U_1[p_1,q_1]$ is the rightmost factor, followed by $\bigotimes_{(p_2,q_2) \in L_2} U_2[p_2,q_2]$, and so on.

\paragraph{Model for random low-depth Clifford circuits} Our model for random depth $D$ Clifford circuits is to choose, for each layer $L_t$, a random partition $\{ (p,q) \}$ of the $n$ qubits, and then for each pair $(p,q)$, and let $U_t[p,q]$ be a uniformly chosen from the two-qubit Clifford group (i.e., the set of all unitaries that preserve the Pauli group under conjugation). 

Brown and Fawzi showed that for $D = \mathcal{O}(\log^3 n)$, the circuit $C$ is an encoding circuit for a good error-correcting code with high probability~\cite{brown2013short}:

\begin{theorem}[\cite{brown2013short}]
\label{thm:bf}
  For all $\delta > 0$, for all integers $n,k,d > 0$ satisfying
  \begin{equation}
    \frac{k}{m} \leq 1 - h(d/n) - \log(3)d/n - 4\delta,
  \end{equation}
  with $h(\cdot)$ as the binary entropy function, the circuit $C$ described in the paragraph above is an encoding circuit for a $[[m,k,d]]$ stabilizer code with probability at least $1 - \Omega(n^{-8})$. In other words, with high probability the subspace $\mathcal{C} = \{ C \ket{\psi} \ket{0}^{\otimes (n-k)} : \text{$\ket{\psi}$ is a $k$-qubit state} \}$ is a $[[n,k,d]]$ stabilizer code.
\end{theorem}

\begin{notation}
  To avoid confusion with the blocklength of our approximate QLDPC code that we construct in our paper (which is denoted by $N$), we will use $n$ to denote the blocklength of the Brown-Fawzi random circuit code.
\end{notation}

Since the circuits are Clifford circuits, the resulting code is a stabilizer code.

\subsection{The spacetime circuit Hamiltonian construction}
\label{subsec:spacetimeconstruction}

As mentioned in the introduction, we use a small variant of the spacetime circuit Hamiltonian of Brueckmann and Terhal~\cite{breuckmann2014space} to create our code Hamiltonian. In this section, we present the spacetime construction for general depth $D$ circuits. In Section~\ref{sec:construction}, we will describe the specific circuit that we will use for our code Hamiltonian.

Let $D$ be an even integer and let $C$ be an $n$-qubit circuit of depth $D$ where $L_1,\ldots,L_{D}$ be the $D$ layers of $C$, where each $L_t$ is a set of $n/2$ two-qubit gates\footnote{By padding with identity gates, we can assume without loss of generality that every layer has exactly $n/2$ two-qubit gates.} $U_t[p,q]$ acting on disjoint pairs of qubits $\{ (p,q) \}$. We assume that $C$ is a ``circular'' circuit; in other words, that it is equivalent to the identity circuit. %

We let $H_{circuit}[C]$ denote the circular spacetime circuit Hamiltonian corresponding to the circular circuit $C$. Let $X \defeq \frac{D-2}{2}$. The Hamiltonian is defined on $n (1+ X+1)=n(X+2)$ qubits, which is divided into three classes of registers: (1) data registers $\sS_1,\ldots,\sS_n$, (2) clock registers $\sC_1,\ldots,\sC_n$, and (3) flag registers $\sF_1,\ldots,\sF_n$. 

The data register $\sS_i$ is a qubit register that corresponds to the $i$-th qubit that the circuit $C$ acts on. The flag register $\sF_i$ is a qubit register that indicates whether the $i$-th qubit's local clock is in the ``forward phase'' or the ``backward phase''; this denotes which half (first or second) of clock states the clock is in. The clock register $\sC_i$ consists of $X$ qubits and indicates the local time of the $i$-th data qubit (within the forward phase or the backward phase). The valid clock states for register $\sC_i$ are $\{\ket{1^j 0^{X - j}}\}$ for $0 \leq j \leq X$ (i.e. a domain wall clock).

Following Brueckmann and Terhal~\cite{breuckmann2014space}, the flag register combined with the clock register allows us to put our qubit clocks ``on a circle'': %
we index time from $0$ to $2X+1 = D-1$, and we identify time $t = D$ with $t = 0$. %
We encode time steps $t$ according to the following convention. For notational convenience, we let the register $\sT_i$ (for ``time register'') denote the union of $\sF_i$ and $\sC_i$. 

\begin{equation}
\label{eq:time_register}
  \ket{t}_{\sT_i} \defeq \begin{cases}
  \ket{0}_{\sF_i} \otimes \ket{1^t \, 0^{X - t}}_{\sC_i} & \qquad \text{ if } t \in \{0,1,\ldots,X\} \\
  \ket{1}_{\sF_i} \otimes \ket{1^X}_{\sC_i} & \qquad \text{ if } t = X+1 \\
  \ket{1}_{\sF_i} \otimes \ket{1^{2X+1-t}\, 0^{t - X-1}}_{\sC_i} & \qquad \text{ if } t \in \{X+2,\ldots,2X+1 \}. 
  \end{cases}
\end{equation}

In other words, the time register evolves in the following way:
\begin{align}
  &\ket{0}_{\sF_i} \otimes \ket{00 \cdots 0}_{\sC_i} \to \cdots \to \ket{0}_{\sF_i} \otimes \ket{11 \cdots 1}_{\sC_i} & \qquad (0 \leq t \leq X) \\
  \to \quad &\ket{1}_{\sF_i} \otimes \ket{11 \cdots 1}_{\sC_i} & \qquad (t = X+1)\\
  \to \quad &\ket{1}_{\sF_i} \otimes \ket{1 \cdots 10}_{\sC_i} \to \cdots \to \ket{1}_{\sF_i} \otimes \ket{00 \cdots 0}_{\sC_i} & \qquad (X+2 \leq t \leq 2X+1).
\end{align}
Notice that in any transition from $\ket{t}_{\sT_i}$ to $\ket{t+1}_{\sT_i}$, there is at most one qubit being flipped. 

For the remainder of this section we fix a circuit $C$ and assume it fixed. The spacetime Hamiltonian $H_{circuit}[C]$ is defined as
\begin{equation}
\label{eq:code_hamiltonian}
  H_{circuit} = H_{clock} + H_{init} + H_{prop} + H_{causal}.
\end{equation}
\begin{notation}
In what follows, subscripts of operators such as ``$\sF_i$'' in ``$\ketbra{0}{0}_{\sF_i}$'' indicates which registers the operators act on. Let $\Pi_{\sR}^{(\alpha)}$ be the projector $\ketbra{\alpha}{\alpha}_\sR$ for any register $\sR$.
\end{notation}

The terms $H_{clock}$, $H_{init}$, $H_{prop}$, and $H_{causal}$ are defined as follows:
\begin{description}
  \item[(1) $H_{clock}$: ] The term $H_{clock}$ enforces that all the clock registers are encoded as described above. We write $H_{clock} = \sum_{i = 1}^n H_{clock}[i]$ where
  \begin{equation}
  \label{eq:clock}
      H_{clock}[i] = \sum_{j = 1}^{D-1} \Pi_{\sC_{i,j} \sC_{i,j+1}}^{(01)}.
  \end{equation}
  This enforces that the register $\sC_i$ encodes a domain wall.
  \item[(2) $H_{init}$: ] The initialization term is defined as $H_{init} = \sum_{i = k+1}^n H_{init}[i]$ for some integer $1\leq k \leq n$, \footnote{In our case, $k$ will eventually be the number of logical qubits.} where
\begin{equation}
\label{eq:init_term}
  H_{init}[i] = \Pi_{\sC_{i,0}}^{(1)} \otimes \Pi_{\sS_i}^{(1)}.
\end{equation}
This term checks that the last $n-k$ qubits are in the state $\ket{0}$ when their corresponding time registers are in state $\ket{0}_{\sT_i}$ or $\ket{2X+1}_{\sT_i}$. We only need to check one bit of the time register $\sT_i$ because of the previous set of terms enforcing that the clock is a domain wall.%

\item[(3) $H_{prop}$: ] The propagation term $H_{prop}$ is defined to be $H_{prop} =\sum_{t = 0}^{D-1} \sum_{(p,q) \in L_t} H_t[p,q]$, where 
\begin{equation}
\begin{aligned}
\label{eq:prop_term2}
H_t[p,q] &= \frac{1}{2} \Big[ \left( A_{t,t}[p,q] + A_{t+1,t+1}[p,q] \right) \otimes \Id \\
& \qquad  \qquad - A_{t+1,t}[p,q] \otimes U_t[p,q] - A_{t,t+1}[p,q] \otimes (U_t[p,q])^\dagger \Big ]
\end{aligned}
\end{equation}
\begin{align}
\text{and } A_{t,t'}[p,q] &= \ketbra{u_t[p]}{u_{t'}[p]} \otimes \ketbra{u_t[q]}{u_{t'}[q]}, \\
\ket{u_t[p]} &= 
\begin{cases}
\Id \otimes \ket{0}_{\sF_p} \ket{1}_{\sC_{p,t}} \ket{0}_{\sC_{p,t+1}} & \text{if } 0 \leq t < X \\
\Id \otimes \ket{0}_{\sF_p} \ket{1}_{\sC_{p,X}} & \text{if } t = X \\
\Id \otimes \ket{1}_{\sF_p} \ket{1}_{\sC_{p,X}} & \text{if } t = X + 1 \\
\Id \otimes \ket{1}_{\sF_p} \ket{1}_{\sC_{p,2X+1-t}} \ket{0}_{\sC_{p,2X+2-t}} & \text{if } X + 1 \leq t \leq 2X+1.
\end{cases}
\label{eq:cases_for_h_prop}
\end{align}
Here, $\ket{u_t[p]}$ is the tensor product of a state on the specified qubits and the identity operator on all unspecified qubits. This term enforces the agreement of slices of the superposition corresponding to two time configurations differing by a gate with respect to the unitary $U_t[p,q]$. Because of the $H_{clock}$ terms, the checks only require looking at a few qubits of the time registers\footnote{In effect, $\ket{u_t[p]}$ is the minimal description of $\ket{t}_{\sT_p}$ given that the state is a ground-state of $H_{clock}$.}. 

\item[(4) $H_{causal}$: ] The term $H_{causal}$ is used to \emph{enforce causality} meaning that the superposition is only over \emph{valid} time configurations (see Definition \ref{def:partial-config}). At a high level, a time configuration $\timeconfig = (t_1, \ldots, t_n)$ is valid if and only if for all pairs of qubits $(p,q)$ sharing a gate in layer $L_t$, both clocks $t_p$ and $t_q$ are $\leq t$ or $> t$. This is, however, complicated by the circularity of time imposed in this particular construction as ``all clocks are both ahead and behind any particular $t$''. In reality, we require the more complicated definition: for all pairs of qubits $(p,q)$ sharing gates in layers $L_{t_a}$ and $L_{t_b}$ for $t_a < t_b$, either $t_p, t_q$ are both $\in [t_a, t_b)$ or are both $\notin [t_a, t_b)$. 

Let $C_p$ be the set of qubits $q$ which interact with qubit $p$.

\begin{equation}
H_{causal} = \sum_{p = 1}^n \sum_{q \in C_p} H_{causal}[p,q]
\label{eq:causality_term}
\end{equation}
where $H_{causal}[p,q]$ is defined as follows. Let $t^{(1)} < t^{(2)} < \ldots < t^{(f)}$ be the times at which $p$ and $q$ share a gate. Then,
\begin{equation}
H_{causal}[p,q] = \sum_{j = 1}^f \sum_{t_p = t^{(j)}}^{t^{(j+1)} - 1} A_{t_p,t_p}[p] \otimes B_{t^{(j)}, t^{(j+1)}}[q]
\end{equation}
where $B_{t, t'}[q]$ is a projector ensuring that qubit $t_q$ is between $t$ and $t'$ (respecting circularity)\footnote{By this we mean that if $t < t'$, the projector is onto the set $\{t, \ldots, t'-1\}$. If $t' < t$, then the projector is onto the set $\{t, \ldots, X\} \cup \{0, \ldots, t' - 1\}$.}. Therefore, we verify that qubit $q$ is valid with respect to qubit $p$. The definition of $B_{t,t'}[q]$ is case dependent.

\begin{description}
\item[Case 1] If $0 \leq t, t' \leq X$. In this case, the flag qubit must be $\ket{0}_{\sF_q}$. Furthermore, $\sC_{q,t}$ must be $\ket{1}$ and $\sC_{q,t'}$ must be $\ket{0}$. Therefore,
\begin{equation}
B_{t,t'}[q] = \Pi_{\sF_q}^{(0)} \Pi_{\sC_{q,t}}^{(1)} \Pi_{\sC_{q,t'}}^{(0)}.
\end{equation}

\item[Case 2] If $X + 1 \leq t, t' \leq 2X + 1$. This is the similar except the flag is flipped. Hence,
\begin{equation}
B_{t,t'}[q] = \Pi_{\sF_q}^{(1)} \Pi_{\sC_{q,2X+2-t'}}^{(1)} \Pi_{\sC_{q,2X+2-t}}^{(0)}.
\end{equation}

\item[Case 3] If $0 \leq t \leq X$ and $X + 1 \leq t' \leq 2X + 1$. In this case, the flag qubit may be different. However, we can write the projector as the sum of the two projectors for the different flags.

\begin{equation}
B_{t,t'}[q] = \Pi_{\sF_q}^{(0)} \Pi_{\sC_{q,t}}^{(1)} + \Pi_{\sF_q}^{(1)} \Pi_{\sC_{q,2X + 2 - t'}}^{(1)} .
\end{equation}

\item[Case 4] If $0 \leq t' \leq X$ and $X + 1 \leq t \leq 2X + 1$. This is similar except again the flag is flipped. Hence,

\begin{equation}
B_{t,t'}[q] = \Pi_{\sF_q}^{(0)} \Pi_{\sC_{q,t'}}^{(0)} + \Pi_{\sF_q}^{(1)} \Pi_{\sC_{q,2X + 2 - t}}^{(0)} .
\end{equation}

\end{description}

\end{description}

\subsection{Bitonic sorting networks}
\label{subsection:bitonicblock}

In this section, we describe a class of circuits called \emph{bitonic sorting networks}. These are parallel circuits, devised by Batcher~\cite{Batcher:1968:SNA:1468075.1468121}, that are used to efficiently sort data arrays. Specifically, these are circuits acting on $n$ elements, with depth $\mathcal{O}(\log^2 n)$. In each layer of the circuit, pairs of elements are compared and swapped. Equivalently, for every permutation $\pi$ on $n$ elements, there is a bitonic sorting network consisting of SWAP and identity gates that implements $\pi$.

Bitonic sorting networks will be a crucial component of our code construction, as we use them to ``uniformize'' the random Brown-Fawzi encoding circuits before applying the spacetime circuit Hamiltonian construction. The uniformity of the resulting circuits will be the key ingredient that allows us to analyze the spectral gap of the Hamiltonian.

\begin{notation}
We will assume that the number of qubits $n$, is a power of $2$, with $n = 2^\ell$ for some integer $\ell$.
\end{notation}

For this paper, we will be interested in the architecture (i.e. the wiring and gate structure) of the bitonic sorting circuit. A bitonic sorting architecture consists of smaller sub-architectures, called \emph{bitonic blocks}.

\begin{definition}An \emph{architecture} is a directed acyclic graph where each vertex $v$ has $\mathrm{deg}_{in}(v) = \mathrm{deg}_{out}(v) \in \{1,2\}$ except for specific vertices $s$ and $t$ which have $\mathrm{deg}_{out}(s) = \mathrm{deg}_{in}(t) = n$ and $\mathrm{deg}_{out}(t) = \mathrm{deg}_{in}(s) = 0$. A circuit $C$ (acting on $n$ qubits) over an architecture is instantiated by specifying a gate for each vertex $v \notin \{s,t\}$ in the graph that acts on the qubits labelled by the edges adjacent to the vertex. The vertices $s$ and $t$ represent the state prior to and after the application of the circuit. That is, we can think of an architecture as an outline of a quantum circuit and one needs to fill in the blanks (specify each gate) to instantiate a circuit.
\label{def:architecture}
\end{definition}

\begin{definition}[Bitonic block \cite{Batcher:1968:SNA:1468075.1468121}]
\label{def:bitonicblock}
For a positive integer $\ell$, the bitonic block of \emph{rank} $\ell$, $\cB_\ell$, is a circuit architecture acting on $2^\ell$ qubits. $\cB_\ell$ is recurisvely defined with the architecture $\cB_1$ being an architecture consisting of a single layer, $\cL_1$, with a gate between qubits 1 and 2 (see part (a) of Figure~\ref{fig:b1b2b3}).

For $\ell > 1$, the bitonic block $\cB_\ell$ is a $\ell$-depth architecture with the first layer, $\cL_1$ being $2^{\ell-1}$ gates connecting qubit $i$ to $i + 2^{\ell-1}$ for $i = 1, 2, \ldots, 2^{\ell-1}$. The following $\ell-1$ layers, $\cL_\ell, \cL_{\ell-1}, \ldots, \cL_2$ are defined recursively as $\cB_{\ell-1}^{\otimes 2}$ where one of the two blocks acts on the qubits $\{1, 2, \ldots, 2^{\ell-1}\}$ and the other on the qubits $\{2^{\ell-1} + 1, 2^{\ell-1} + 2, \ldots, 2^\ell\}$.
\end{definition}

\noindent See Figure~\ref{fig:b1b2b3} for illustrations of blocks $\cB_2$, and $\cB_3$.

\begin{figure}[ht]
\centering
\includegraphics[width = 0.7\linewidth]{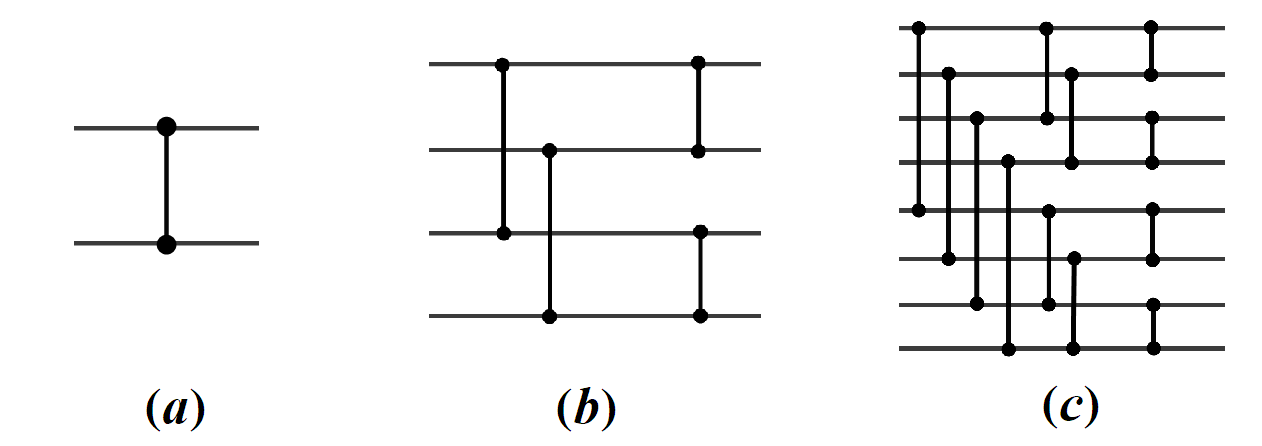}
\caption{(a) Bitonic block $\mathcal{B}_1$.  (b) Bitonic block $\mathcal{B}_2$.  (c) Bitonic block $\mathcal{B}_3$.}
\label{fig:b1b2b3}
\end{figure}

\begin{theorem}[\cite{Batcher:1968:SNA:1468075.1468121}]
Let $B_\ell$ be an instantiation of a bitonic block architecture $\cB_\ell$ with generalized comparator gates for some well-ordering -- i.e. given two input wires, it either swaps them or performs the identity such that the larger element is on the lower wire. Then, given two monotonically decreasing sequences of length $2^{\ell - 1}$ as inputs, the output of the circuit $B$ is the merged monotonically decreasing sequence. 
\end{theorem}

\begin{corollary}[\cite{Batcher:1968:SNA:1468075.1468121}]
The following $\frac{\ell(\ell + 1)}{2}$ depth circuit is a sorting circuit:
\begin{equation}
B_\ell B_{\ell - 1}^{\otimes 2} B_{\ell - 2}^{\otimes 4} \ldots B_1^{\otimes 2^{\ell - 1}}.
\end{equation}
\end{corollary}

We will use the notation $\cB_\ell^{\times r}$ for the product (i.e. concatenation) of $r$ bitonic blocks $\cB_\ell$. To simplify the analysis, we can insert additional layers of identity gates so that the circuit architecture is $\cB_\ell^{\times \ell}$; this at most doubles the size of the circuit. Therefore, we can make the following statement:

\begin{lemma}
For any permutation $\pi \in S_n$, there exists a circuit $B_\pi$ of the architecture $\cB_\ell^{\times \ell}$ applying $\pi$ on the input wires.
\label{lem:permforspaciallylocal}
\end{lemma}

\begin{proof}
Note which comparator gates of the bitonic sorting circuit would be SWAP gates if sorting according to the permutation $\pi$. Pad with identity gates as previously stated till the circuit conforms to the architecture.
\end{proof}

\subsection{Uniformizing circuits for spacetime Hamiltonians}
\label{subsec:uniformizing}

We now present a general method for encoding depth $D$ circuits $C$ into a spacetime circuit Hamiltonian, in a way that allows us to give a good lower bound on the spectral gap. Let $C$ denote a circuit of depth $D$ consisting of layers $L_1,\ldots,L_{D}$, where each $L_t$ is a set of $n/2$ two-qubit gates.

We preprocess the circuit $C$ in multiple steps to obtain a slightly larger-depth circuit $C'$. We ``uniformize'' the circuit using bitonic sorting networks described in the previous section. The circuit $C$ will not, in general, correspond to nearest-neighbor interactions in small dimension. We add bitonic sorting networks in between each layer $L_t$ of $C$ to ensure that all the Clifford gates act on adjacent qubits. Because of the regular structure of the sorting networks, the resulting circuit will consist of nearest-neighbor interactions on a hypercube of dimension $\ell = \log n$. 

More formally, we do the following: label the qubits using $\{1,\ldots,n\}$. In a layer $L_t$ of $C$ for $1 \leq t \leq D$, a qubit $q$ is generally not paired with a neighboring qubit $q - 1$ or $q + 1$. Instead, there is some permutation $\pi_t$ on $n$ qubits that maps the pairs $\{ (1,2), (2,3), \ldots, (n-1,n) \}$ to the pairs $L_t = \{ (p,q) \}$. Let $\pi_t(L_t)$ denote the layer where all the qubits are permuted by $\pi_t$, and all the gates in $L_t$ now act on consecutive qubits. 

By Lemma~\ref{lem:permforspaciallylocal}, there exists a circuit $B_{\pi_t}$ with the architecture  $\cB_\ell^{\times \ell}$ for $\ell = \log n$ that implements the permutation $\pi_t$. Replace each layer $L_t$ in $C$ by the following subcircuit $K_t$: first apply $B_{\pi_t {\pi_t}^{-1}}$ and then apply the layer $\pi_t(L_t)$. Since the last layer of $B_{\pi_t {\pi_t}^{-1}}$ and $\pi_t(L_t)$ have the same architecture, we can merge the gates into a single layer. Here we assume $\pi_0 = \Id$.%

The final $C'$ is the composition of the subcircuits $K_1,K_2,\ldots,K_{D_1}$, yielding a depth $\mathcal{O}(D \log^2 n)$ circuit. Note that by induction, circuit $C'$ is exactly equivalent to the original circuit $C$. Notice that each subcircuit $K_t$ can be implemented as nearest-neighbor gates on a hypercube of dimension $\log n$, and thus the same holds for $C'$ as well.

Let $D'$ denote the depth of circuit $C'$. We consider spacetime Hamiltonians $H_{circuit}[C']$ of the circuit $C'$, as described in Section~\ref{subsec:spacetimeconstruction}. We first note that it has the following properties: it is a $9$-local Hamiltonian, the terms act locally on a $\mathcal{O}(D' + \log n)$-dimensional lattice, and each qubit participates in at most $\mathcal{O}(D')$ terms. 

%
\section{Construction of the code Hamiltonian}
\label{sec:construction}

Here we describe our code construction in detail. Let $\eps > 0$ be the desired target approximation error. Let $n,k,d$ be integers satisfying Theorem~\ref{thm:bf} where $k,d = \Omega(n)$. Let $C_0$ denote a Clifford circuit of depth $D_0 = \mathcal{O}(\log^3 n)$ that is an encoding circuit of an $[[n,k,d]]$ code $\cC_{BF}$, as promised by Theorem~\ref{thm:bf}. Let $L_1,\ldots,L_{D_0}$ be the $D_0$ layers of $C_0$, where each $L_t$ is a set of at $n/2$ two-qubit Clifford gates. 

The first preprocessing step is to replace all the Clifford gates by gates from the set
\begin{equation}
  I = \begin{pmatrix} 1 & 0 \\ 0 & 1 \end{pmatrix}, \qquad H = \frac{1}{\sqrt{2}} \begin{pmatrix} 1 & 1 \\ 1 & -1 \end{pmatrix}, \qquad S = \begin{pmatrix} 1 & 0 \\ 0 & i \end{pmatrix}, \qquad CNOT = \begin{pmatrix} 1 & 0 & 0 & 0 \\ 0 & 1 & 0 & 0 \\ 0 & 0 & 0 & 1 \\ 0 & 0 & 1 & 0\end{pmatrix}.
\end{equation}
This is possible because the gate set $\{ I, H, S, CNOT \}$ generates the Clifford group; thus every two-qubit Clifford gate can be written as a $\mathcal{O}(1)$-length product of $I$, $H$, $S$, and $CNOT$ gates. The depth of this circuit is $D_1 = \mathcal{O}(D_0)$. Let $C_1$ denote this circuit.

Next, we \emph{pad} the circuit to have depth $3 D_1/\eps$ where the last $1 - (\eps/3)$ fraction of the layers are simply applications of the identity gate on consecutive pairs of qubits. Call this padded circuit $C_1'$; its depth is $D_1' = 3 D_1/\eps$.

Now, let $C_2$ be the circuit obtained by preprocessing $C_1'$ as described in Section~\ref{subsec:uniformizing}. 
This has depth $D = \mathcal{O}(\log^5 n)$. Let $H_{circuit}[C_2]$ denote the corresponding spacetime circuit Hamiltonian, acting on $N = \mathcal{O}(nD)$ qubits. For what follows, we will abbreviate $H_{circuit}[C_2]$ as $H$. 

Let $\cC$ denote the ground space of $H$. This will be our code. We now show that $\cC$ is an approximate QLDPC code, and we establish its parameters.

\begin{theorem}
  For all $\eps > 0$, the subspace $\cC$ is a $[[N,k,d,\eps,\ell,s]]$ approximate QLDPC code, for $k = \Omega(N/\log^5 N)$, $d = \Omega(N/\log^5 N)$, $\ell = 9$, and $s = \polylog(N)$.
\end{theorem}
\begin{proof}
First we have to show that $\cC$ is the image of an encoding map, $\Enc$. We present methods for efficiently generating a codeword of the code $\cC$ in Section~\ref{subsec:encoding_ckt}. 

Next, we present a recovery map for the code (i.e. a map that approximately corrects errors and decodes). An important point is that the Brown-Fawzi stabilizer code underlying our construction was probabilistically chosen and there is no known efficient correction algorithm for their code. However, since the stabilizer code encoded by the circuit $C$ satisfies the Knill-Laflamme error correction conditions \cite{quant-ph/9604034}, there exists an ideal recovery map, $\Rec_{BF}$, that can correct any error on $(d-1)/2$ qubits or less. In other words, for all errors $\cE$ acting on at most $(d-1)/2$ qubits, the following is equivalent to the identity channel on $k$ qubits:
\begin{equation}
  \Rec_{BF} \circ \cE \circ \Enc_{BF}
\end{equation}
where $\Enc_{BF}$ is the encoding map for the Brown-Fawzi code. This is all that we will need. %

Our recovery map $\Rec$ for our code works as follows: given an input state on registers $\sS_1 \cdots \sS_n$ and $\sT_1 \cdots \sT_n$ (i.e. the data and time registers), it
\begin{enumerate}
  \item Traces out the registers $\sT_1 \cdots \sT_n$.
  \item Applies the Brown-Fawzi ideal recovery map $\Rec_{BF}$ to $\sS_1 \cdots \sS_n$.
\end{enumerate}

We now prove the approximate error correction condition. We rely on the following Lemma, which we prove in Appendix \ref{subsec:countingconfigsofproducts}.

\begin{lemma}
\label{lem:overlap}
  Let $\cT$ denote the set of all time configurations of a spacetime history state. There exists a subset $\cT_{comp} \subset \cT$ such that for all spacetime history states
  \begin{equation}
    \ket{\psi} = \frac{1}{\sqrt{|\cT|}} \sum_{\timeconfig \in \cT} \ket{\timeconfig} \otimes \ket{\psi_\timeconfig}
  \end{equation}
  there exists a codeword $\ket{\Gamma} \in \cC_{BF}$ such that if $\timeconfig \in \cT_{comp}$, then $\ket{\psi_\timeconfig} = \ket{\Gamma}$. Furthermore, we have that{}
  \begin{equation}
    \frac{|\cT_{comp}|}{|\cT|} \geq 1 - \eps.
  \end{equation}
\end{lemma}

Recall that $\cC_{BF}$ is the $[[n,k,d]]$ stabilizer code\footnote{The $BF$ subscript stands for ``Brown-Fawzi''.} whose encoding map is the circuit $C_0$ described above. 

Let $\ket{\phi} \in (\C^2)^{\otimes k} \otimes \mathcal{R}$ be a $k$-qubit message $\rho$ that has been purified (i.e., $\rho = \Tr_{\mathcal{R}}(\ketbra{\phi}{\phi})$). Let a Schmidt decomposition of $\ket{\phi}$ be $\sum_i \sqrt{p_i} \ket{\xi_i} \ket{i}$, where the $\{ \ket{\xi_i} \}$ correspond to the Hilbert space $(\C^2)^{\otimes k}$ and the $\{ \ket{i} \}$ are orthonormal vectors in $\mathcal{R}$. Let $\ketbra{\Psi}{\Psi} = \Enc(\ketbra{\phi}{\phi})$, so that $\ket{\Psi} = \sum_i \sqrt{p_i} \ket{\Psi_i} \ket{i}$ where 
$\ket{\Psi_i} = \frac{1}{\sqrt{T}} \sum_{\timeconfig} \ket{\timeconfig} \otimes \ket{\psi_{i,\timeconfig}}$ is the spacetime history state for circuit $C$ on input state $\ket{\xi_i}\otimes \ket{0^{(n-k)}}$. 
By Lemma~\ref{lem:overlap} we can write
\begin{equation}
  \ket{\Psi_i} = \frac{1}{\sqrt{|\cT|}} \sum_{\timeconfig \notin \cT_{comp}} \ket{\timeconfig} \otimes \ket{\psi_{i,\timeconfig}} + \frac{1}{\sqrt{|\cT|}} \sum_{\timeconfig \in \cT_{comp}} \ket{\timeconfig}\otimes \ket{\Gamma_i}.
\end{equation}
Define the following (subnormalized) states:
\begin{equation}
  \ket{\lambda} = \frac{1}{\sqrt{|\cT|}} \sum_{\timeconfig \in \cT_{comp}} \ket{\timeconfig}, \qquad \ket{\wt{\Psi_i}} = \ket{\lambda} \otimes \ket{\Gamma_i}.
\end{equation}

Note that $\ket{\wt{\Psi_i}}$ has norm equal to $\| \ket{\lambda} \|^2 \geq 1 - \eps$ because of Lemma~\ref{lem:overlap}. Furthermore, $| \ip{\wt{\Psi}_i}{\Psi_i}|^2 = \| \ket{\lambda} \|^2$. If we define $\ket{\wt{\Psi}} = \sum_i \sqrt{p_i} \ket{\wt{\Psi_i}} \ket{i}$, then we have that
\begin{equation}
  F(\ketbra{\Psi}{\Psi}, \ketbra{\wt{\Psi}}{\wt{\Psi}}) \geq (1 - \eps)^2 \geq 1 - 2\eps.
\end{equation}
Let $\cE$ be a completely positive, trace preserving map acting on at most $(d - 1)/2$ qubits. Since $\cC_{BF}$ is a code that can correct up to $(d-1)/2$ errors, and $\ket{\wt{\Psi}}$ is a (sub-normalized) superposition of codewords of $\cC_{BF}$ (along with a state $\ket{\lambda}$ that gets traced out by $\Rec$), we have that $\Rec \circ \cE(\ketbra{\wt{\Psi}}{\wt{\Psi}}) = \ketbra{\phi}{\phi} \cdot \ip{\lambda}{\lambda}$. Since the fidelity metric is non-decreasing under quantum operations, we have that
\begin{align}
F \Paren{ \Rec \circ \cE \circ \Enc(\ketbra{\phi}{\phi}), \ketbra{\phi}{\phi} } 
 &\geq F \Paren{ \Rec \circ \cE(\ketbra{\Psi}{\Psi}) , \Rec \circ \cE(\ketbra{\wt{\Psi}}{\wt{\Psi}}) } \\
 &\geq F \Paren{ \ketbra{\Psi}{\Psi} , \ketbra{\wt{\Psi}}{\wt{\Psi}} }\\
&\geq 1 - 2\eps.
\end{align}

As discussed in Section~\ref{subsec:uniformizing}, the geometry underlying the Hamiltonian $H$ is a lattice with dimension $\mathcal{O}(D \polylog n)$; each $9$-local term acts in a spatially-local manner on this lattice, and each qubit participates in $\polylog n$ terms. 
This establishes the Theorem.
\end{proof}

\subsection{Encoding circuit}
\label{subsec:encoding_ckt}

We demonstrate that there is an efficient circuit generating a ground-state of the Hamiltonian.

\begin{theorem}
There exists an encoding circuit of polynomial size in $n$ which on input $\ket{\psi}$ generates the state $\Enc(\psi)$. In particular, the polynomial size circuit generating the state is $\log(n) + 2$ spatially local.
\label{thm:encodingckt}
\end{theorem}

Proving the generability of the ground-state is done in two parts. We first show that once one can generate a particular superposition over the time registers, one can generate the ground-state. Next, we provide an efficient algorithm for generating the particular superposition. This is encapsulated formally in Lemmas \ref{lem:buildtau} and \ref{lem:buildpsi}.

The superposition over the time registers (the union of all clock and flag registers) of interest is the uniform superpositions over \emph{valid partially applied configurations} -- or \emph{valid configurations}, for brevity - of an architecture\footnote{A formal definition of an architecture is given as Definition \ref{def:architecture}.}. Imagine progressively applying a circuit from an architecture, gate by gate. Non-commutativity of gates in different layers demands that the gates of the circuit cannot be applied in any order, but must be applied in a way that respects causality.

\begin{definition}[Partial configuration of an architecture]
\label{def:partial-config}
A partial configuration of an architecture on $n$ qubits of depth $D$, is a vector of integers describing how many layers of gates have been applied per qubit: $\tau \in Z_{D+1}^n$. A partial configuration is \emph{valid} if it respects the causal dependence of the gates in the circuit.

Formally, consider a gate $g$ at depth $d_g$ acting on qubits $i$ and $j$ as \emph{applied} if $d_g \leq \tau_i$ and $d_g \leq \tau_j$. Then an architecture is valid if for every marked gate $g$, any gate $g'$ such that $g' \rightarrow g$ in the DAG represented the architecture (see Definition \ref{def:architecture}) is also marked. 

Notationally, we refer to a qubit $i$ being at time $t = \tau_i$.
\end{definition}

Specifically, we are interested in generating the uniform superposition over valid configurations of the architecture behind the circuit $C_2$ from Section. We note that the architecture is similar to the product of bitonic block architectures (see Definition \ref{def:linearproduct}), $\cB_\ell^{\times m}$ for $m = D_0',\ell = O(\log^4 n / \eps)$, and $\ell = \log n$. However, on closer inspection, since the code Hamiltonian includes terms that check the consistency between clocks at the final time state and the initial time state, this does not exactly correspond to the spacetime Hamiltonian construction from a linear product of bitonic blocks. Rather, it corresponds to the spacetime Hamiltonian construction from a ``circular'' product of bitonic blocks. The set of valid configurations for this Hamiltonian includes configurations which ``wrap around'' the final time state and back to the initial time state. Formally,

\begin{definition}[Valid configurations of a circular architecture]
Let $\cA$ be an architecture on $n$ qubits of depth $D$. Let $\cA^{\times \infty}$ be the infinite circuit defined by taking infinite consecutive copies of $\cA$:
\begin{equation}
\cA^{\times \infty} = \prod_{j = -\infty}^\infty \cA.
\end{equation}
Let $\cV_\infty \subset \mathbb{Z}^n$ be the set of valid configurations for $\cA^{\times \infty}$. Define $\cV \in \mathbb{Z}_D^n$ by $\cV = \cV_\infty / D \mathbb{Z}^n$, i.e. identity identical time configurations in the infinite copies. The set of valid configurations for the circular architecture is $\cV$.
\end{definition}

We call this a \emph{circular architecture}, $\cB_\ell^{\leftrightarrow m}$ and describe it with more formality in the Appendix (Definition \ref{def:circularwrapping}).

\begin{lemma}
Let $\ket \Xi$ be the uniform superposition over valid configurations of the architecture $\cB_\ell^{\leftrightarrow m}$. Formally, let $\cV \subseteq Z_{m\ell}^n$ be the set of valid configurations. Then,
\begin{equation}
\ket{\Xi} = \frac{1}{\sqrt{\abs{\cV}}} \sum_{(t_1, \ldots, t_n) \in \cV} \ket{t_1}_{\sT_1} \ket{t_2}_{\sT_2} \ldots \ket{t_n}_{\sT_n}.
\label{eq:superposition}
\end{equation}
The state $\ket{\Xi}$ can be generated efficiently.
\label{lem:buildtau}
\end{lemma}

\newcommand{\enum}{f}

\begin{proof}
By Theorem \ref{thm:countcircular}, the number of valid configurations, $a_\ell^{\leftrightarrow m} \defeq \abs{\cV}$ has a recursive definition and is at most doubly exponential in $\ell$. Therefore, it can be calculated in time $\poly(n)$. There exists an enumeration bijection $\enum : [a_\ell^{\leftrightarrow m}] \rightarrow \mathcal{V}$ which is (classically) efficient such that $\enum^{-1}$ is also (classically) efficient. We extend this to reversible operations 
\begin{equation}
\enum(\ket{j}\ket{c}) = \ket{j}\ket{c \oplus \enum(j)} \quad \text{ and } \quad \enum^{-1}(\ket{j}\ket{c}) = \ket{j \oplus \enum^{-1}(c)} \ket{c}.
\end{equation}
As a consequence, we can create $\ket{\Xi}$ by starting\footnote{The state $\frac{1}{\sqrt{b_k}} \sum_{j = 1}^{a_\ell^{\leftrightarrow m}} \ket{j}$ can be generated efficiently by the following. Let $W$ be the largest power of 2 greater than $a_\ell^{\leftrightarrow m}$. Using Hadamard gates, we can generate the superposition $\frac{1}{\sqrt{W}} \sum_{j = 1}^{W} \ket{j}\ket{0}$. Then, we apply the reversible operation $\ket{j}\ket{z} \mapsto \ket{j}\ket{z \oplus 1_{j \leq b_k}}$ and measure the ancilla in the standard basis. If the measurement is 1, we achieve the desired state. The measurement 0  will occur with probability $\leq 1/2$.} with 
\begin{equation}
\frac{1}{\sqrt{a_\ell^{\leftrightarrow m}}} \sum_{j = 1}^{a_\ell^{\leftrightarrow m}} \ket{j} \ket{0}
\end{equation}
and applying $\enum$ followed by $\enum^{-1}$. A construction of $\enum$ and $\enum^{-1}$ is given as Theorem \ref{thm:efficientindex} in the Appendix.

\end{proof}

\begin{lemma}
Given the state $\ket{\Xi}$ (\ref{eq:superposition}) and an initial state $\ket{\psi} \in (\mathbb{C}^2)^{\otimes k}$, one can efficiently generate the state
\begin{equation}
\ket{\Psi} = \frac{1}{\sqrt{\cV}} \sum_{(t_1, \ldots, t_n) \in \cV} \ket{t_1}_{\sT_1} \ket{t_2}_{\sT_2} \ldots \ket{t_n}_{\sT_n} \otimes U_{t_1,\ldots,t_n} \ket{\psi}_{\sS_1 \sS_2 \ldots \sS_k} \ket{0^{n-k}}_{\sS_{k+1} \sS_{k+2} \ldots \sS_n}
\label{eq:gstate}
\end{equation}
where $U_{t_1, \ldots, t_n}$ is the unitary acting on $(\mathbb{C}^2)^{\otimes n}$ defined by the action of the valid configuration $(t_1, \ldots, t_n)$. The efficient generating circuit is spatially local in $2 + \log_2(n)$ dimensions.
\label{lem:buildpsi}
\end{lemma}

\begin{proof}
We describe three methods for efficiently constructing the state $\ket{\Xi}$ (\ref{eq:superposition}). The first describes a quantum circuit, the second is approximate and relies on phase estimation, and the third is based on adiabatic computation. We describe the first in detail and provide sketches for the other two.

Notationally, we will let $\sT$ be the union of registers $\{\sT_i\}$ and similarly define the registers $\sC, \sF$ and $\sS$.

\paragraph{Method 1 (quantum circuit)}
Let $C_2$ be the circuit from which the spacetime circuit Hamiltonian is built. We modify the circuit $C_2$ into a new circuit $C_2'$ which acts in one additional dimension such that
\begin{equation}
C_2' \left( \ket{\Xi}_{\widetilde{\sT}} \otimes \ket{\psi}_{\sS_1 \sS_2 \ldots \sS_k} \ket{0^{n-k}}_{\sS_{k+1} \sS_{k+2} \ldots \sS_n} \otimes \ket{0}_{\sT} \right) = \ket{0}_{\widetilde{\sT}} \otimes \ket{\Psi}_{\sS\sT}.
\end{equation}
Here the register $\widetilde{\sT}$ is a copy of the register $\sT$. The additional dimension of $C_2'$ over $C_2$ is the additional interaction with the register $\widetilde{\sT}$. Each sub-register $\widetilde{\sT}_i$ will interact with register $\sT_i$ as well as any register $\widetilde{\sT}_j$ for register $\sT_j$ that $\sT_i$ interacts with.

For every gate $g$ in $C_2$ at depth $t$ acting on qubit registers $\sS_i$ and $\sS_j$, we replace $g$ with a constant depth circuit acting on $\sS_i, \sS_j, \sT_i, \sT_j, \widetilde{\sT}_i, \widetilde{\sT}_j$. The constant depth circuit applies the following map: On input $\ket{\psi}_{\sS_i \sS_j} \ket{t_i}_{\sT_i} \ket{t_j}_{\sT_j} \ket{T_i}_{\widetilde{\sT}_j} \ket{T_j}_{\widetilde{\sT}_j}$, if $t_i < T_i$ and $t_j < T_j$, then the gate $g$ is applied to $\ket{\psi}_{\sS_i \sS_j}$ and the registers $\sT_i$ and $\sT_j$ are incremented. Otherwise, the identity map is applied. Equivalently,
\begin{equation}
\ket{\psi}_{\sS_i \sS_j} \ket{t_i}_{\sT_i} \ket{t_j}_{\sT_j} \ket{T_i}_{\widetilde{\sT}_j} \ket{T_j}_{\widetilde{\sT}_j} \mapsto 
\begin{cases}
g \ket{\psi}_{\sS_i \sS_j} \ket{t_i + 1}_{\sT_i} \ket{t_j + 1}_{\sT_j} \ket{T_i}_{\widetilde{\sT}_j} \ket{T_j}_{\widetilde{\sT}_j} & \text{if } t_i < T_i \land t_j < T_j \\
\ket{\psi}_{\sS_i \sS_j} \ket{t_i}_{\sT_i} \ket{t_j}_{\sT_j} \ket{T_i}_{\widetilde{\sT}_j} \ket{T_j}_{\widetilde{\sT}_j} & \text{otherwise}.
\end{cases}
\end{equation}
Notice that making this adjustment to every gate $g$ in $C_2$ will on input $\ket{T_1}_{\widetilde{\sT}_1}\ket{T_n}_{\widetilde{\sT}_n}\ldots \ket{T_n}_{\widetilde{\sT}_n}$, generate the partial computation of $C_2$ up to $(T_1, \ldots, T_n)$. Furthermore, in the end, the registers $\sT$ and $\widetilde{\sT}$ will both contain $(T_1, \ldots, T_n)$. Then, we can apply the map $\ket{a}_{\widetilde{\sT}} \ket{b}_\sT \mapsto \ket{a \oplus b}_{\widetilde{\sT}} \ket{b}_{\sT}$ to erase the $\widetilde{\sT}$ register.

By linearity, when ran on the input $\ket{\Xi}_{\widetilde{\sT}}$, this will yield the state $\ket{\Psi}$.

There is one complication to consider. We must consider valid configurations which cross time $t = 0$ (i.e. some clocks are near the end while others are just starting). To fix this, we first preprocess register $\widetilde{\sC}$ such that any register with $\ket{T_i}_{\widetilde{\sC}_i}$ for $T_i < D/2$ is replaced with $T_i + D$. This ensures that all clock registers are in the range\footnote{This ensure that all valid configurations are ``consecutive'' because the width of a configuration (Lemma \ref{lem:width}) is much smaller than $D$.} $[D/2, 3D/2]$. We now perform the same adjustment to each gate $g$ except we do it for gates of the circuit $C' C'$ (circuit repeated twice). We follow it with postprocessing to return all clock registers to between $0$ and $D$.

\paragraph{Method 2 (phase estimation)}  Apply the original random Clifford circuit $C_0$ of depth $D_0$ to the computational qubits, to the form the a state with no entanglement between the clocks and the data qubits,
\begin{equation}
\ket{\Psi'} = \ket{\Xi} \otimes C_0 \ket{\psi}_{\sS_1 \sS_2 \ldots \sS_k} \ket{0^{n-k}}_{\sS_{k+1} \sS_{k+2} \ldots \sS_n}.
\end{equation}
Since the spacetime history state $|\Psi\rangle$ is padded to length $T = \poly(D_0)$ with identity gates, the overlap between these states is
\begin{equation}
| \langle {\Psi}|{\Psi'}\rangle|^2 \geq 1 -  \left(\frac{D_0}{T} \right)^2= 1 - \mathcal{O}\left(\frac{1}{\polylog(n)}\right)
\end{equation}
Since the spectral gap of the code Hamiltonian scales as $1/\poly(n)$, we can apply a phase estimation circuit $U_\textrm{PE}$ to $\ket{\Psi'}$ that estimates the first $\mathcal{O}(\log n)$ digits of the energy of this state with respect to the code Hamiltonian.  By the overlap calculation above this phase estimation yields an eigenvalue of 0 with probability $1 - o(1)$ and projects $|\Psi'\rangle$ into the ground space of the code Hamiltonian.  Using the fact that distinct code words are orthogonal, this produces a state of the form
\begin{equation}
\sqrt{1 -\epsilon^2} |\Psi\rangle + \epsilon | \Psi''\rangle
\end{equation}
where $\ket{\Psi''}$ is contained in the code space and $\epsilon \leq D_0 / T \leq 1/\polylog(n)$.  

\paragraph{Method 3 (adiabatic computation)}  As described in Section 3.4 of \cite{breuckmann2014space}, a standard way to turn a circuit Hamiltonian $H(U_1,...,U_L)$ into a procedure for adiabatically preparing the ground state is to use a continuous family of circuit Hamiltonians $H(s) = H(U_1(s), ... , U_L(s))$ to define the adiabatic path.   Here $s \in [0,1]$ is a parameter such that $U_i(0) = I$ and $U_i(1) = U_i$ for each $i$.  Since we use 2-qubit gates and $SU(4)$ is simply connected we may define $U_i(s) = U_i^s$ for each $i$.  Every Hamiltonian in this continuous family has the same spectral gap, which we have shown is $1/\poly(n)$ in Section \ref{sec:analysis}, and so any rigorous version of the adiabatic theorem suffices to turn the initial ground state of $H(I,...,I)$ into the ground state of $H(U_1,...,U_L)$ in polynomial time.  Alternatively, instead of the adiabatic theorem, one can discretize the adiabatic path into polynomially many steps and use phase estimation to move between consecutive steps, which suffices to prepare the ground state of $H(U_1,...,U_L)$ with exponentially small error.
\end{proof}

\section{Spectral gap analysis}
\label{sec:analysis}

Our analysis of the spectral gap of the spacetime circuit Hamiltonian begins with several standard steps that are applied to Feynman-Kitaev Hamiltonians~\cite{aharonov2008adiabatic} and their spacetime variants~\cite{breuckmann2014space}.  First one defines a global unitary rotation,
\begin{equation}
W = \sum_{\boldsymbol{\tau} \in \mathcal{T}} U(\boldsymbol{\tau} \leftarrow 0) |\boldsymbol{\tau}\rangle \langle \boldsymbol{\tau}|
\end{equation}
which when applied to full Hamiltonian yields,
\begin{align}
W^\dagger H W&= H_\textrm{init}  + \mathcal{L}_{\textrm{prop}} \otimes \mathds{1} + H_\textrm{causal} \label{eq:rotatedH}\\
\mathcal{L}_{\textrm{prop}}  &= \sum_{\boldsymbol{\tau},\boldsymbol{\tau'} \in \mathcal{T}} \left( |\boldsymbol{\tau}\rangle \langle \boldsymbol{\tau} | + |\boldsymbol{\tau'}\rangle \langle \boldsymbol{\tau'}| - |\boldsymbol{\tau}\rangle \langle \boldsymbol{\tau'} | - |\boldsymbol{\tau'}\rangle \langle \boldsymbol{\tau} | \right)
\end{align}
where $\mathcal{L}_{\textrm{prop}}$ is the combinatorial Laplacian of a graph with vertices corresponding to valid time configurations and edges connecting time configurations $\boldsymbol{\tau},\boldsymbol{\tau'}$ that differ by the application of a 2-local gate.  Since $[H_{\textrm{causal}}, H_{\textrm{prop}}] = 0$, any state with energy less than 1 will be in the ground space of $H_\textrm{causal}$. 

Applying the argument from Section 3.1.4 of \cite{aharonov2008adiabatic}, the Hamiltonian \eqref{eq:rotatedH} in the rotated frame only acts on the computational qubits through $H_\textrm{in}$; so the Hamiltonian can be written in block diagonal form with each block $B_i$ corresponding to a different input string $i = 0,...,2^n -1$,
\begin{equation}
W^\dagger H W = \left( \begin{array}{ccccc}
B_0 &  &  &  &  \\
 & B_1 &  &  &  \\
 &  & \ddots & &  \\
 &  &  & B_{2^n -1}&  \\
\end{array} \right)
\end{equation}
The ground space of $W^\dagger H W$ is contained in the block $B_0$; therefore, the spectral gap of $H$ will either be the spectral gap of $\Delta_{\textrm{prop}}$ within $B_0$, or it will be the minimum among the ground state energies in the other blocks $B_i$.  To lower bound the ground state energies in the blocks $B_i \neq 0$ we can apply the geometrical lemma. 

\paragraph{Kitaev's Geometrical lemma}  Let $H_1, H_2 \succeq 0$ be positive semi-definite operators with 0 as an eigenvalue, and let the least nonzero eigenvalue of $H_1$ and $H_2$ be lower bounded by $\Lambda$, then
\begin{equation}
H_1 + H_2 \succeq \Lambda \sin^2\left(\frac{\theta}{2}\right) \quad , \quad \cos^2 \theta = \max_{\substack{\ket{\xi}\in \textrm{Ker}(H_1) \\ |\eta \rangle \in \textrm{Ker}(H_2)}}
\end{equation}

The lemma is applied in each of the blocks $B_i$ with $i \neq 0$, taking $H_1 = H_{\textrm{in}}$ and $H_2  = \mathcal{L}_{\textrm{prop}} \otimes \mathds{1}$.  Since the spectral gap of $H_{\textrm{in}}$ is 1 we take $\Lambda = \Delta_{\textrm{prop}}$, where $\Delta_{\textrm{prop}}$ is the spectral gap of $\mathcal{L}_{\textrm{prop}}$.  The sine of the angle between the kernels of $H_{\textrm{in}}$ and $H_{\textrm{prop}}$ is lower bounded by the overlap of the ground state of $H_{\textrm{prop}}$ with any one of the local terms in $H_{\textrm{in}}$, which is $1/T$.  Therefore the spectral gap of the full Hamiltonian satisfies
\begin{equation}
\Delta_H = \frac{\Delta}{T} = \wt{\Omega}(\Delta_{\textrm{prop}}).
\end{equation}
It remains to lower bound the spectral gap $\Delta$ of the graph Laplacian $\mathcal{L}_{\textrm{prop}}$.  This graph Laplacian is a stoquastic frustration-free Hamiltonian with a uniform ground state in the time configuration basis, and so it can be mapped to a Markov chain transition matrix by shift and rescaling,
\begin{equation}
P = I - \mathcal{L}_{\textrm{prop}}/ \|\mathcal{L}_\textrm{prop}\| \label{eq:Pdef}
\end{equation}
The transformation of a spacetime Hamiltonian $H_\textrm{prop}$ into a Markov chain in \cite{breuckmann2014space} first maps the 1 + 1 dimensional spacetime Hamiltonian to the ferromagnetic Heisenberg chain and then relates the Heisenberg model to a Markov chain transition matrix by rescaling its operator norm.  This multistep mapping reveals additional insights about the physics of that model, and is the basis for the gap analysis in \cite{breuckmann2014space}, but the mapping from stoquastic Hamiltonians to Markov chains is entirely general as described in \cite{crosson2017quantum}. 

The operator norm satisfies $\|L_\textrm{prop}\| = \|H_{\textrm{prop}}\| = \Omega(n)$ and so in terms of the spectral gap $\Delta_P$ of the Markov chain \eqref{eq:Pdef} we have
\begin{equation}
\Delta_H = \Omega(n \Delta_P). \label{eq:gapHP}
\end{equation}
\subsection{Preliminaries on Markov chains}
Throughout this section let $(\Omega,\pi,P)$ be an irreducible, ergodic, reversible Markov chain on the state space $\Omega$, with stationary distribution $\pi$ and transition matrix $P$ (see~\cite{levin2017markov} for background on these terms).   %

\paragraph{Block decomposition method~\cite{madras2002}}  The state space is decomposed into subsets $\Omega = \cup_{i=1}^R \Omega_i$ (``blocks''), which in general will have nonempty pairwise intersection, $\Omega_i \cap \Omega_j \neq \emptyset$.  Let $\Theta \defeq \max_{x \in \Omega} |\{i : x \in \Omega_i\}|$ be the maximum numbers of sets that can contain any single element $x \in \Omega$.  For any $S \subseteq \Omega$, define $\pi(S) \defeq \sum_{x \in S} \pi(x)$.  Define the aggregate (``block'') Markov chain $\overline{P}$ on the state space $\{1,...,R\}$,
\begin{equation}
\overline{P}_{ij} \defeq \frac{\pi\left(\Omega_i \cap \Omega_j\right)}{\Theta \pi\left(\Omega_i\right)} \quad , \quad i,j \in \{1,...,R\}. \label{eq:aggP}
\end{equation}
One can easily check that these transition probabilities are reversible with respect to the distribution $\bar{\pi}_i \defeq \pi(\Omega_i)$.  Next define a restricted (``within-block'') chain $P_i$ for each subset $\Omega_i$ as follows: if $x \in \Omega_i$ and $ y \neq x$ then
\begin{equation}
P_{i}(x,y) \defeq  \begin{cases} 
      P(x,y) & y \in \Omega_i \\
      0 & y \notin \Omega_i 
   \end{cases}, \label{eq:prestricted}
\end{equation}
and $P_{i}(x,x) \defeq 1 - \sum_{y \neq x} P_{i}(x,y)$.  The spectral gap of $\Delta_P$ satisfies the lower bound
\begin{equation}
\Delta_P \geq \frac{1}{2} \Delta_{\overline{P}} \min_{i = 1,...,R} \Delta_{P_i}
\end{equation}
\paragraph{Cheeger's inequality}  
For any nonempty subset $S\subseteq \Omega$ define the conductance $\Phi(S)$ by
\begin{equation}
\Phi(S) \defeq \frac{1}{\pi(S)}\sum_{x \in S , y \in S^c} \pi_x P_{xy}
\end{equation}
and define $\Phi_P \defeq \min_{S : 0 < \pi(S) \leq 1/2}\Phi(S)$.  Cheeger's inequality states that
\begin{equation}
\frac{\Phi_P^2}{2} \leq \Delta_P.
\end{equation}
\subsection{Decomposition of the circuit Propagation Markov chain}
The subsets in our decomposition are defined by
\begin{equation}
\Omega_r \defeq \left \{ (t_1,...,t_n) \in \Omega :  r \leq t_i \leq r + \ell \; \textrm{for all} \; i = 1,...,n\right\}  \quad , \quad r = 0, ... , D - \ell
\end{equation}
(recall $\ell \defeq \log(n)$).  Every valid time configuration is contained in at least one $\Omega_r$, so $\Omega = \cup_{r = 0}^{D - r} \Omega_r$ as required.  The maximum number of blocks that can contain any particular time configuration is $\Theta = k$; this maximum is attained by any configuration $(t_1,...,t_n)$ for which $t_1 = ... = t_n = r$, with such configurations being contained in $\Omega_{r - \ell} ,... , \Omega_{r}$. 

Next we compute the aggregate transition probabilities $\overline{P}_{ij}$.  In particular, we will need the transition probability between consecutive blocks $P_{r,r+1}$.  Since the distribution over time configurations is uniform, we have
$$
\frac{\pi(\Omega_r \cap \Omega_{r'}}{\pi(\Omega_r)} = \frac{|\Omega_r \cap \Omega_{r'}|}{|\Omega_r|}
$$
for all $r,r' = 0,\ldots, D - \ell$.  To determine $|\Omega_r|$, we use the recursion relation $a_\ell = 2 a^2_{\ell -1} - a^4_{\ell -2}$ for the number of partially completed circuit configurations $a_\ell$ of a bitonic block of rank $\ell$ which is defined in Definition \ref{def:bitonicblock} and the recursive relation is proved in Theorem \ref{thm:countingBitonicBlock}).  This recursion relation has been studied previously~\cite{lagarias2002counting} and has the asymptotic solution $a_\ell = \phi^{-1} \omega^{2^\ell}$, where $\phi = (1 + \sqrt{5})/2$ is the golden ratio and $\omega = 1.8445...$ does not have a known closed form. Next in Appendix \ref{sec:permutationSplitting} we show that $|\Omega_r| = |\Omega_{r'}|$ for every pair of blocks $r,r'$, which follows from the fact that the valid configurations of any circuit architecture are invariant under permutation of the qubit labels, together with an explicit set of permutations we define that relates the architecture in each block.  Therefore we have $|\Omega_r| = a_{\ell} = \phi^{-1} \omega^{2^\ell}$ for all $r = 0,...,D - \ell$.

To evaluate \eqref{eq:aggP} we also need to count the number of configurations contained in the intersection of two such consecutive blocks, see Figure \ref{fig:overlaps}.  The key insight is that removing either the first or last layer of any block will split it into two independent bitonic blocks on half the number of qubits, which implies that $|\Omega_r \cap \Omega_{r+1}| = a^2_{\ell-1}$ and so

\begin{figure}
\begin{center}
\includegraphics[width = 0.9\linewidth]{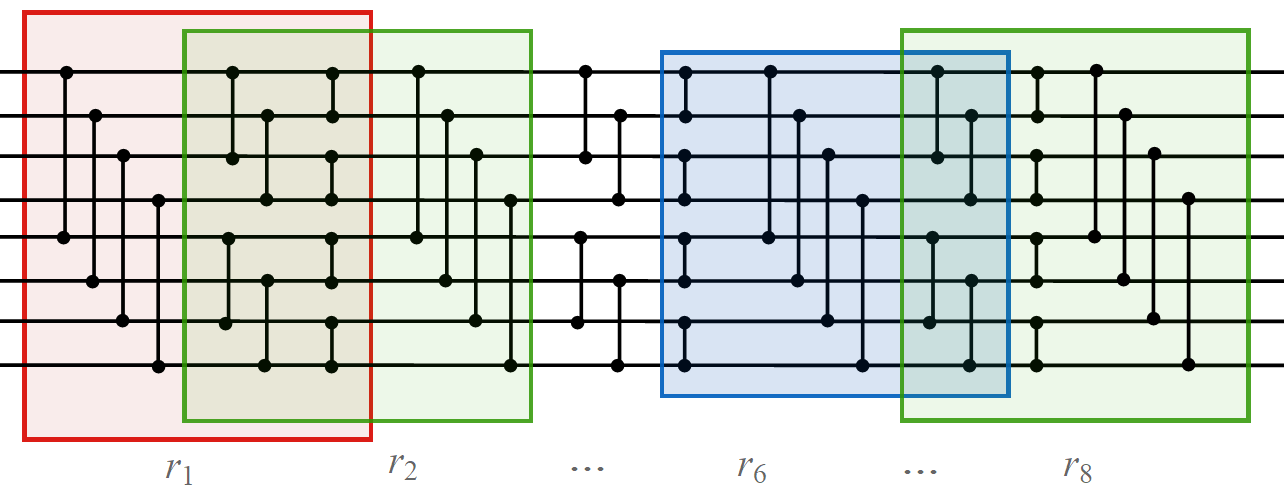}
\end{center}
\caption{The region in the intersection of the red and green bitonic blocks $\mathcal{B_3}$ contains the time configurations that belong to $\Omega_{r_1} \cap \Omega_{r_2}$.  The key insight is that the valid time configurations in the intersection of these blocks can be counted by observing that the intersection corresponds to two independent copies of $\mathcal{B_2}$.  Similarly, the configurations contained in the intersection $\Omega_6 \cap \Omega_8$ correspond to 4 independent copies of $\mathcal{B}_1$.}
\label{fig:overlaps}
\end{figure}

\begin{equation}
\frac{\pi(\Omega_r \cap \Omega_{r+1})}{\pi(\Omega_r)} = \frac{a_{\ell - 1}^2}{a_\ell} = \phi^{-2} \quad , \quad r = 0, ... , T - \ell.
\end{equation}
and similarly,
\begin{equation}
\frac{\pi(\Omega_r \cap \Omega_{r+j})}{\pi(\Omega_r)} = \frac{a_{\ell - j}^{2^j}}{a_\ell} =\phi^{-2^j}\quad , \quad r = 0, ... , T - \ell.
\end{equation}
and so the aggregate transition probabilities decay doubly exponentially with distance,
\begin{equation}
\overline{P}(r, r+ j) = \ell^{-1} \phi^{-2^j}\quad , \quad r = 0, ... , D - \ell.
\end{equation}
Next we lower bound the minimum conductance $\Phi_{\overline{P}}$.  Let $S$ be a nonempty subset of blocks, $S \subset \{0,...,D - \ell\}$.  Define $\pi(S) \defeq \sum_{r \in S} \pi(\Omega_r)$, and assume $\pi(S) \leq 1/2$.  There must be some $r'' < D - \ell$ such that $r'' \notin S$, and so
\begin{equation}
\frac{1}{\pi(S)}\sum_{r \in S , r' \in S^c} \pi(\Omega_r) \overline{P}_{r,r'} \geq  2 \ell^{-1} \pi(\Omega_{r''-1} \cap \Omega_{r''}) = 2 \ell^{-1}\frac{a^2_{\ell -1}}{|\mathcal{C}|} \geq  2 \ell^{-1}(m \ell + 1)^{-1} \phi^{-2}.
\end{equation}
Therefore, Cheeger's inequality yields
\begin{equation}
\Delta_{\overline{P}} \geq 2 \ell^{-2} ( m \ell + 1)^{-2} \phi^{-4} = \frac{1}{\polylog(n)}.
\end{equation}
It remains to lower bound the spectral gaps $\Delta_{P_i}$ corresponding to the restricted (``within-block'') chains defined in \eqref{eq:prestricted}.  By our careful choice of the block decomposition, we demonstrate in Section \ref{sec:bitonicIsomorphism} a one-to-one correspondence between the time configurations in any block $\Omega_i$ and the equal area dyadic tilings of a unit square, and crucially this correspondence also exactly maps the edge-flip Markov chain moves considered in \cite{cannon_et_al:LIPIcs:2017:7583} to the updates which describe the application of a local gate to a valid time configuration.  Since the relaxation time of the edge-flip Markov chain is $\mathcal{O}(n^{4.09})$ we have $\Delta_{P_i} = \Omega(n^{-4.09})$  and so
\begin{equation}
\Delta_P = \wt{\Omega}(n^{-4.09}),
\end{equation}
and by \eqref{eq:gapHP} this implies
\begin{equation}
\Delta_H = \wt{\Omega}(n^{-3.09}).
\end{equation}

\section{Local detection of Pauli errors}
\label{sec:localDetect}
In this section we describe the local detection of errors on spacetime codewords with probability $1 - 2^{-\polylog(N)}$ with $\polylog(N)$-depth circuits. The class of errors that we handle is the set of tensor products of Pauli operators on the physical qubits (which includes data and time qubits). Interestingly, we can detect Pauli errors even if the weight of the error (the number of qubits affected) exceeds the distance of the spacetime code!
Here we only describe a single round of error detection while assuming the ability to perform measurements implemented by low-depth circuits perfectly.  %

\begin{definition}[Pauli group]
  The \emph{Pauli group on $N$ qubits}, denoted by $\cP_N$, is the group generated by the $N$-fold tensor product of the Pauli matrices
  \begin{equation}
I=\begin{pmatrix}1&0\\0&1\end{pmatrix}, \qquad \sigma_X=\begin{pmatrix}0&1\\1&0\end{pmatrix}, \qquad \sigma_Y=\begin{pmatrix}0&-i\\i&0\end{pmatrix}, \qquad \sigma_Z=\begin{pmatrix}1&0\\0&-1\end{pmatrix}  
  \end{equation}
  along with multiplication by $\{\pm 1, \pm i\}$.
\end{definition}

\begin{definition}[Pauli channels]
A quantum operator $\cE$ acting on $N$ qubits is a \emph{Pauli channel} if it has a Kraus decomposition 
\begin{equation}
  \cE(\rho) = \sum_{P \in \cP_N} c_P \, P \rho P^\dagger
\end{equation}
where $\{ c_P \}$ is a probability distribution over $\cP_N$. 
\end{definition}

\subsection{Pauli stabilizers of the spacetime code}

There are nonidentity elements $P \neq I$ of the Pauli group $\cP_N$ that stabilize the spacetime code, i.e., for all $\ket{\psi} \in \cC$, we have $P \ket{\psi} = \ket{\psi}$. In this section, we identify three stabilizers; in the next section, we will argue that these are the \emph{only} nonidentity stabilizers, and all other nonidentity Pauli operators can be locally detected with high probability.

Let $C$ denote the circuit such that the code Hamiltonian is $H_{circuit}[C]$, as described in Section~\ref{sec:construction}. Recall that $X = \frac{D - 2}{2}$ where $D$ is the depth of $C$. 

For any $p \in \{1,\ldots,n\}$ and $j \in \{1, \ldots, z\}$ consider the set of 4 qubits $\sC_{p,j}, \sC_{q_1,j}, \sC_{q_2,j}$, and $\sC_{p',j}$ where $q_1$ is the qubit in layer $L_j$ interacting with $p$, $q_2$ is the qubit in layer $L_{2X + 1 - j}$ interacting with $p$, and $p'$ is the qubit in layer $L_{2X+1 - j}$ interacting with $q_1$. Because the layers $L_j$ and $L_{2X+1-j}$ are different layers of the bitonic architecture, we know that these layers together form a product of bitonic blocks of rank 2, $\cB_2$ (see Corollary \ref{cor:subbitonic}). Then, it is easy to also see that $p'$ is the qubit in layer $L_j$ interacting with $q_2$. Define $rect(p,j)$ as the elements $\{p,q_1, q_2, p'\}$. It is not difficult to see that $rect(\cdot,j)$ yields the same set on inputs $p, q_1, q_2, p'$. Let the stabilizer $S_{clock}[rect(p,j)]$ be
\begin{equation}
\label{eq:s_clock}
S_{clock}[rect(p,j)] \defeq \sigma_Z(\sC_{p,j}) \otimes \sigma_Z(\sC_{q_1,j}) \otimes \sigma_Z(\sC_{q_2,j}) \otimes \sigma_Z(\sC_{p',j})
\end{equation}
where $\sigma_Z(\sC_{p,j})$ denotes the $\sigma_Z$ operator acting on the clock qubit $\sC_{p,j})$.  Furthermore, let

\begin{equation}
\label{eq:s_flag}
  S_{flag} \defeq \bigotimes_{p = 1}^n \sigma_Z(\sF_p)
\end{equation}
where $\sigma_Z(\sF_p)$ denotes the $\sigma_Z$ operator acting on the flag qubit corresponding to data qubit $p$. This is the product of $\sigma_Z$'s acting on all the flag qubits.

\begin{claim}
\label{clm:pauli_stabilizers}
  $S_{clock}[rect(p,j)]$ for any qubit $p$ and $1 \leq j \leq X$ and $S_{flag}$ are Pauli stabilizers of the spacetime code.
\end{claim}
\begin{proof}
Let $\timeconfig = (t_1,\ldots,t_n) \in \cT$ be a valid time configuration of the spacetime history state.

Recall that every time configuration can be seen as the result of incrementing the clocks by applying gates. Therefore there is a sequence of time configurations $(\timeconfig_0, \timeconfig_1, \ldots, \timeconfig_f = \timeconfig)$ such that $\timeconfig_0$ has all clock and flag registers set to 0, and each $\timeconfig_{i+1}$ differs from $\timeconfig_i$ by the application of a gate. To each time configuration $\timeconfig_i$ we can associate a $b_{\timeconfig_i} \in \{\pm 1\}$ such that 
\begin{equation}
S_{clock}[rect(p,j)] \ket{\timeconfig_i} = b_{\timeconfig_i} \ket{\timeconfig_i}.
\end{equation}
Clearly $b_{\timeconfig_0} = 1$. We argue that $b_{\timeconfig_{i+1}} = b_{\timeconfig_i}$. Consider the gate differentiating these two configurations. Applying it must change the time registers by flipping the values of $\sC_{r,j'}$ and $\sC_{s,j'}$ (and perhaps the corresponding flag registers). This either flips the sign of $b_{\timeconfig_{i}}$ twice (if this gate is one of $rect(p,j)$) or not at all (if it is not). Therefore, $b_{\timeconfig_{i+1}} = b_{\timeconfig_i}$. This proves that $b_{\timeconfig} = 1$ and that $S_{clock}[rect(p,j)]$ is a stabilizer as
\begin{equation}
  S_{clock}[rect(p,j)] \ket{\psi} = \frac{1}{\sqrt{|\cT|}} \sum_{\timeconfig \in \cT} S_{clock}[rect(p,j)] \ket{\timeconfig} \otimes \ket{\psi_\timeconfig} = \ket{\psi}.
\end{equation}

A similar argument can be made showing that $S_{flag}$ is also a stabilizer by arguing that either pair of flag qubits must be flipped or none are flip when transitioning from a valid time configuration to the next. Therefore, $S_{flag}$ is also a stabilizer.

\end{proof}

Let $\cS$ be the closure of the following set under product,
\begin{equation}
\left\{I, S_{flag} \right\} \cup \bigcup_{p = 1}^n \bigcup_{j = 1}^X S_{clock}[rect(p,j)].
\end{equation}
Every element of $\cS$ is a stabilizer.

\subsection{Locally detecting errors}
In this section, we argue that there is a set of local operators that can detect, with high probability, any Pauli error in $\cP_N \setminus \cS$.

Our argument will rely on a structural property of the Brown-Fawzi circuit $C$ that holds with high probability (when the circuit is sampled according to the random Clifford model described in Section~\ref{sec:prelim}). 
\begin{definition}
A depth $D$ circuit on $n$ qubit is \emph{nice} if for every qubit $p \in \{1,\ldots,n\}$, there exists layers $1 \leq t,t' \leq D$ such that: %
  \begin{enumerate}
    \item The two-qubit gate acting on $p$ in layer $L_t$ is $H \otimes I$, where the Hadamard gate $H$ acts on qubit $p$, and
    
    \item The two-qubit gate acting on $p$ in layer $L_{t'}$ is $S \otimes I$, where the phase gate $S = \begin{pmatrix} 1 & 0 \\ 0 & i \end{pmatrix}$ acts on qubit $p$.
  \end{enumerate}
\end{definition}

\begin{fact}[\cite{brown2013short}]
\label{prop:nice}
  The Brown-Fawzi encoding circuit $C$ sampled according to the random Clifford model described in Section~\ref{sec:prelim} is nice with probability at least $1 - 2^{-\Omega(\log^2 n)}$.
\end{fact}

\noindent The following fact can be verified via a simple computation.
\begin{fact}
\label{prop:HS_real}
Each of the following unitary operators has eigenvalues $i$ and $-i$:
\begin{enumerate}
  \item $H \sigma_X H \sigma_X$
  \item $H \sigma_Z H \sigma_Z$
  \item $S^\dagger \sigma_Y^\dagger S \sigma_Y$
\end{enumerate}
Here, $H$ is the Hadamard gate, and $S$ is the phase gate.
\end{fact}

We now proceed to prove the local error detection property of the spacetime code.

\begin{theorem}
\label{thm:local_detection}
  Suppose the Brown-Fawzi encoding circuit defining the spacetime code $\cC$ is nice. Then there exists a collection $\cD$ of $\polylog(N)$-local projectors satisfying the following properties:
  \begin{enumerate}
    \item Each projector $\Pi \in \cD$ acts on $10$ qubits of the code space, and acts on $s = \polylog(N)$ ancilla qubits initialized in the $\ket{0}$ state. 
    \item For all $n$-qubit states $\ket{\psi}$, we have that $\Pi \ket{\psi} \ket{0^s} = 0$ for all $\Pi \in \cD$ if and only if $\ket{\psi}$ is a codeword in the spacetime code $\cC$.
    \item For all Pauli channels $\cE$, for all codewords $\ket{\psi} \in \cC$, there exists a projector $\Pi \in \cD$ such that
    \begin{equation}
      \Tr \Paren{\Pi \, \Paren{\cE(\psi) \otimes \ketbra{0^s}{0^s}} } \geq (1 - \alpha)(1 - 2^{-\polylog(N)})
    \end{equation}
    where $\psi = \ketbra{\psi}{\psi}$ and $\alpha = \sum_{P \in \cS} c_P$ is the weights of the channel $\cE$ on the Pauli stabilizers in $\cS$. %
    \end{enumerate}
    Furthermore, there exists a measurement $M$, implementable by a circuit of $\polylog(N)$ depth acting on $\mathcal{O}(N \polylog(N))$ qubits, such that for all Pauli channels $\cE$ and for all codewords $\ket{\psi} \in \cC$
    \begin{equation}
      \Tr \Paren{ M \, \Paren{ \cE(\psi) \otimes \ketbra{0^{Ns}}{0^{Ns}} }} \geq (1 - \alpha) (1 - 2^{-\polylog(N)}).
    \end{equation}
\end{theorem}

\begin{proof}

We first define a set of projectors $\cD_0$ that \emph{weakly} detect errors, in the sense that for every Pauli channel $\cE$, for every spacetime codeword $\ket{\psi}$, there is a projector $\Pi \in \cD_0$ that has expectation value at least $(1 - \alpha)/\polylog(N)$ on $\cE(\psi)$. We will then boost the set $\cD_0$ into the desired set $\cD$ that detects errors with high probability, using QMA-amplification techniques.

\paragraph{A weak set of detector projections} The weak detection set $\cD_0$ will simply be the set of local terms of the spacetime circuit Hamiltonian defining the spacetime code. We first show that for each member $P$ of the Pauli group $\cP_N$ that is not a Pauli stabilizer in $\cS$, there exists a projector $\Pi \in \cD_0$ such that
\begin{equation}
  \Tr \Paren{ \Pi \, P \psi P^\dagger } \geq \frac{1}{\polylog N}.
\end{equation}

Fix a $P \in \cP_N \setminus \cS$, and fix a spacetime codeword $\ket{\psi} \in \cC$, which we can write as
  \begin{equation}
    \ket{\psi} = \frac{1}{\sqrt{|\cT|}} \sum_{\timeconfig \in \cT} \ket{\timeconfig} \otimes \ket{\psi_\timeconfig}.
  \end{equation}
We divide our analysis into several cases. Write $P = P_{flag} \otimes P_{clock} \otimes P_{data}$, where $P_{flag}$ acts on the $\{ \sF_p \}$ registers, $P_{clock}$ acts on the $\{ \sC_{p,j} \}$ registers and $P_{data}$ acts on the $\{ \sS_p \}$ registers. 

\begin{description}

\item[Case 1.] Suppose that $P_{clock}$ has a tensor factor that is either $\sigma_X$ or $\sigma_Y = -i \sigma_Z \sigma_X$. In other words, there exists a data qubit $p \in \{1,\ldots,n\}$ and an associated clock qubit $j \in \{1,\ldots,X\}$ such that the tensor factor of $P$ corresponding to the register $\sC_{p,j}$ (i.e. the part of $P$ acting on the $j$'th clock qubit of $p$'s clock) is one of $\{\sigma_X,\sigma_Y\}$. We can write $P = P_Z P_X$ where $P_Z$ consists of only $\sigma_Z$ and identity factors and $P_X$ consists only $\sigma_X$ and identity factors. By assumption, $P_X$ has at least one $\sigma_X$ acting on a clock qubit.

  Let $\rho = \Tr_{\comp{\sC_p}}(P\psi P^\dagger)$ denote\footnote{Here and throughout the paper, we use the notation $\comp{\sR}$ to refer to all registers excluding $\sR$.} the reduced density matrix of $P\ket{\psi}$ on the clock register $\sC_p$. Notice that $\rho = P_p \Tr_{\comp{\sC_p}}(\psi) P_p^\dagger$ where $P_p$ is the restriction of $P$ to the qubits of the $\sC_p$ register. We now appeal to the following Lemma, which we prove in the Appendix as Lemma \ref{lem:circularDensity}:
  
  \begin{lemma}
  \label{lem:uniform_clock}
    The marginal distribution of the clock register $\sC_p$ of any data qubit $p$ in a spacetime codeword is uniform over the $X+1 = D/2$ states
    \begin{equation}
      \ket{0^X},\ket{10^{X-1}},\ldots,\ket{1^X}.
    \end{equation}
  \end{lemma}
  
  Lemma~\ref{lem:uniform_clock} implies that $\Tr_{\comp{\sC_p}}(\psi)  = \frac{1}{X+1} \sum_{t=0}^X \ketbra{1^t 0^{X-t}}{1^t 0^{X-t}}$. Since this is a convex combination over standard basis states, we have that
  \begin{equation}
    \rho = P_X \Paren{\frac{1}{X+1}\sum_{t=0}^X \ketbra{1^t 0^{X-t}}{1^t 0^{X-t}} } P_X^\dagger.
  \end{equation}
  It is easy to see that for all $P_X \neq I$, we have that there is at least one $0 \leq t \leq X$ such that $P_X \ket{1^t 0^{X-t}}$ is not a valid clock state -- that is, there is a location $j \in \{1,\ldots,X\}$ such that 
  \begin{equation}
    \Tr \Paren{ \ketbra{01}{01}_{\sC_{p,j} \sC_{p,j+1}} \, P_X \ketbra{1^t 0^{X-t}}{1^t 0^{X-t}} P_X^\dagger} = 1.
  \end{equation}
  Notice that the projector $\Pi_j = \ketbra{01}{01}_{\sC_{p,j}\sC_{p,j+1}}$ is precisely one of the terms in the spacetime Hamiltonian (see~\eqref{eq:clock}). Thus we have that
  \begin{equation}
    \Tr(\Pi_j \, P \psi P^\dagger) = \Tr( \Pi_j \, \rho) \geq \frac{2}{D} = \Omega \Paren{\frac{1}{\log^5 N}}.
  \end{equation} 
  
\item[Case 2.] Suppose that $P_{clock}$ only has $\sigma_Z$ or identity factors and $P_{flag}$ has a tensor factor that is either $\sigma_X$ or $\sigma_Y = -i\sigma_Z \sigma_X$. In other words, there exists a data qubit $p \in \{1,\ldots,n\}$ such that the associated tensor factor of $P$ corresponding to the register $\sF_p$ is one of $\{\sigma_X,\sigma_Y\}$. We can write $P_{flag} \otimes P_{clock} = P_Z P_X$ where $P_Z$ consists of (up to multiplication by $\{1,-1,i,-i\}$) only $\sigma_Z$ and identity factors and $P_X$ consists only $\sigma_X$ and identity factors. By assumption, $P_X$ has at least one $\sigma_X$ acting on a flag qubit.

  \begin{description}
    
  \item[Case 2.1.] We first consider a subcase that $P_X$ includes as a factor the operator
  \begin{equation}
    T_{flag} = \bigotimes_p \sigma_X(\sF_p).
  \end{equation}
  In other words, $P_X$ flips every flag qubit. This maps every valid time configuration $\timeconfig = (t_1,\ldots,t_n)$ to a ``mirror'' time configuration $\comp{\timeconfig} = (\comp{t}_1,\ldots,\comp{t}_n)$ where $\comp{t}_j = D-1 - t_j$ according to the mapping described in~\eqref{eq:time_register}. Mirror time configurations are also valid time configurations (i.e. they satisfy the causality constraints of the spacetime Hamiltonian). %
  
  We argue that these mirror time configurations, combined with the state of the data qubits, cannot satisfy the propagation constraints of the spacetime Hamiltonian. To see this, suppose that the circuit $C$ had the following subcircuit $J$ appended to both the beginning and end of the circuit $C$. The subcircuit $J$ consists of two bitonic block architectures $\cB_\ell$, wherein each block all the gates are identity gates except for the last layer, which is populated with $\sigma_X \otimes \sigma_X$ gates acting on each neighboring pair of qubits. Thus the subcircuit $J$ is equivalent to the identity circuit because the two layers of $\sigma_X$ gates cancel each other out. Appending $J$ to the beginning and end of the circuit $C$ yields a circuit with a small increase in depth, and it can be checked that this does not qualitatively affect the analysis of the spacetime Hamiltonian. Thus we will assume that our circuit $C$ has this structure.
  
  The circuit $C$ acts on $n$ qubits, $n - k$ of which are ancilla qubits that are initialized in all the all zeroes state. Let $p$ denote an ancilla qubit. Let $\timeconfig = (t_1,\ldots,t_n)$ be any valid time configuration such that $t_p = \ell-1$. Since this time is before the first row of $\sigma_X$ gates in the circuit $C$, qubit $p$ in $\ket{\psi_\timeconfig}$ is in the state $\ket{0}$. The $\sigma_X$ gates get applied in the transition from time $\ell - 1$ to $\ell$, so qubit $p$ in $\ket{\psi_\timeconfig'}$ is in the state $\ket{1}$, where $\timeconfig'$ is the time configuration obtained from $\timeconfig$ by applying the two-qubit $\sigma_X \otimes \sigma_X$ gate to qubit $p$ and its neighbor. 
  
  Now consider the mirror time configurations $\comp{\timeconfig} = (\comp{t}_1,\ldots,\comp{t}_n)$, and $\comp{\timeconfig}' = (\comp{t}_1',\ldots,\comp{t}_n')$. We have that $\comp{t}_p = D - 1 - t_p = D - \ell - 1$ and $\comp{t}_p' = D - 1 - t_p' = D - \ell - 2$. The gate on qubit $p$ corresponding between times $D - \ell - 2$ and $D - \ell - 1$ is an identity gate (because we're assuming that the circuit $C$ has the subcircuit $J$ at the end). %
  
  Let $\Pi = H_t[p,q]$ for $t = D - \ell - 1$. This projector acts as the identity on the $\sS$ register. Observe that
  \begin{align}
    P \ket{\psi} &= \frac{1}{\sqrt{|\cT|}} \sum_{\timeconfig \in \cT} P_Z \ket{\comp{\timeconfig}} \otimes P_{data} \ket{\psi_{\timeconfig}} \\
    &= \frac{1}{\sqrt{|\cT|}} \sum_{\timeconfig \in \cT} \ket{\comp{\timeconfig}} \otimes b_\timeconfig P_{data} \ket{\psi_{\timeconfig}}
  \end{align}
  for some $b_\timeconfig \in \{\pm 1\}$. 
  
  In what follows, we use the notation $\timeconfig[p,q]$ to denote the pair $(t_p,t_q)$, and use $\timeconfig \to_{p,q} \timeconfig'$ to indicate that the time configuration $\timeconfig'$ is $\timeconfig$ updated by the gate $U_t[p,q]$. We now calculate the expectation 
  \begin{align}
    &\bra{\psi} P^\dagger \Pi P \ket{\psi} \\
    &= \frac{1}{|\cT|} \sum_{\substack{\timeconfig \in \cT : \\ \comp{\timeconfig}[p,q] = (t,t) \\ \comp{\timeconfig} \to_{p,q} \comp{\timeconfig}'}} 
    \Paren{ b_\timeconfig \bra{\psi_{\timeconfig}} + b_{\timeconfig'} \bra{\psi_{\timeconfig'}} } P_{data}^\dagger \Pi P_{data}  \Paren{ b_\timeconfig \ket{\psi_{\timeconfig}} + b_{\timeconfig'} \ket{\psi_{\timeconfig'}} } \\
    &= \frac{1}{|\cT|} \sum_{\substack{\timeconfig \in \cT : \\ \comp{\timeconfig}[p,q] = (t,t) \\ \comp{\timeconfig} \to_{p,q} \comp{\timeconfig}'}}
    \Paren{1 - b_\timeconfig b_\timeconfig' \mathrm{Re} \, \ip{\psi_{\timeconfig}}{\psi_{\timeconfig'}}  }.
  \end{align}
  Observe that when $\comp{\timeconfig}[p,q] = (t,t)$, we have that $\timeconfig[p,q] = (D - 1 -t,D-1-t) = (\ell,\ell)$, and $\timeconfig'[p,q] = (\ell-1,\ell-1)$. But from the reasoning above, we have that $\ket{\psi_\timeconfig}$ and $\ket{\psi_{\timeconfig'}}$ are orthogonal, because the state of qubit $p$ of the two vectors are orthogonal. Therefore 
  \begin{equation}
    \bra{\psi} P^\dagger \Pi P \ket{\psi} = \sum_{\substack{\timeconfig \in \cT : \\ \comp{\timeconfig}[p,q] = (t,t) \\ \comp{\timeconfig} \to_{p,q} \comp{\timeconfig}'}} \frac{1}{|\cT|} .
  \end{equation}
  This is equal to the probability that a uniformly random time configuration $\timeconfig$ is such that $\comp{\timeconfig}[p,q] = (t,t)$. By Lemma~\ref{lem:uniform_clock}, this is at least $1/D^2 = 1/\polylog(N)$.

    \item[Case 2.2.] The second subcase is that there is at least one flag qubit $\sF_p$ such that $P_X$ is identity on it. We follow a similar line of reasoning as in Case 1.
        
    Let $(p,q)$ be a pair of data qubits such that $P_X$ acts as the identity on $\sF_p$ but has a $\sigma_X$ acting on $\sF_q$. Let $t^* = \lceil X/2 \rceil$. Let $\timeconfig = (t_1,\ldots,t_n)$ be any time configuration where $t_p = t_q = t^*$. 
    
  Let $\ket{\timeconfig'} = P_X \ket{\timeconfig}$. Then it must be that $\timeconfig' \notin \cT$. In other words, it is not a valid time configuration. This is because if $\timeconfig' = (t_1',\ldots,t_n')$, then $t_p' = t_p = t^*$, yet $t_q' = 2X + 1 - t_q > X + 1$. In particular, $f_q(\timeconfig') = 1$ and $f_p(\timeconfig') = 0$, which violates the causality constraints on the set of time configurations. In other words, the ``membrane'' described by $\timeconfig'$ is broken between qubits $p$ and $q$. Let $\Pi$ be the component of $H_{causal}[p,q]$ verifying $t_p = t^*$ from the spacetime Hamiltonian (see~\eqref{eq:causality_term}). Then we have that $\bra{\timeconfig'} \Pi \ket{\timeconfig'} = 1$. 
  
  Let $\rho = \Tr_{\comp{\sT}}(P \psi P^\dagger)$ denote the reduced density matrix of $P\ket{\psi}$ on the time configuration register $\sT$. Notice that $\rho = P^{time} \Tr_{\comp{\sT}}(\psi) (P^{time})^\dagger$ where $P^{time}$ is the restriction of $P$ to the time configuration register. Since
  \begin{equation}
    \Tr_{\comp{\sT}}(\psi) = \frac{1}{|\cT|} \sum_{\timeconfig} \ketbra{\timeconfig}{\timeconfig}
  \end{equation}
  is a convex combination of classical states, and the $P_Z$ operator leaves classical states invariant, we have that $\rho = P_X \Tr_{\comp{\sT}}(\psi) P_X^\dagger$. 
  
  From a similar argument to that in Case 1, we obtain that the probability of sampling a time configuration $\timeconfig = (t_1,\ldots,t_n)$ such that $t_p = t_q = t^*$ is $1/D^2$. Thus
  \begin{equation}
    \Tr(\Pi P \psi P^\dagger) = \Tr(\Pi \rho) \geq \frac{1}{D^2}.
  \end{equation}

  \end{description}
  
  \item[Case 3.] Now suppose that $P_{flag} \otimes P_{clock}$ only has $\sigma_Z$ or identity factors. This means that for all $\timeconfig \in \cT$, we have $P_{flag} \otimes P_{clock} \ket{\timeconfig} = b_\timeconfig \ket{\timeconfig}$, where $b_\timeconfig \in \{\pm 1\}$. Thus we can write $P \ket{\psi}$ as
  \begin{equation}
    P \ket{\psi} = \frac{1}{\sqrt{|\cT|}} \sum_{\timeconfig \in \cT} \ket{\timeconfig} \otimes b_\timeconfig P_{data} \ket{\psi_\timeconfig}.
  \end{equation}

  \begin{description}
  
  \item[Case 3.1.] First, suppose that $P_{data} \neq I$. Let $p \in \{1,\ldots,n\}$ be a data qubit such that $P_{data}$ is some non-identity Pauli matrix $\sigma \in \{ \sigma_X ,\sigma_Y, \sigma_Z\}$ on $\sS_p$. If $\sigma = \sigma_Y$, let $t$ denote a layer and $q \neq p$ denote a qubit such that $U_t[p,q]$ in the Brown-Fawzi circuit is $S \otimes I$, with $S$ acting on $p$; otherwise, let $t$ and $q$ be such that $U_t[p,q] = H \otimes I$. Such $t,q$ exist by Proposition~\ref{prop:nice}. Without loss of generality suppose that $\sigma = \sigma_X$.

  Consider the projector $\Pi = H_t[p,q]$ in the spacetime circuit Hamiltonian, which is one of the projectors in the set $\cD_0$. Let $\timeconfig = (t_1,\ldots,t_n)$ be a time configuration such that $t_p = t_q = t$. Let $\timeconfig' = (t_1',\ldots,t_n')$ be the same as $\timeconfig$ except $t_p' = t_q' = t+1$ (i.e. it is the configuration after gate $U_t[p,q]$ is applied). Thus $\ket{\psi_{\timeconfig'}} = H \ket{\psi_\timeconfig}$. Then notice that
  \begin{align}
    &\Paren{ b_{\timeconfig} \bra{\timeconfig} \otimes \bra{\psi_\timeconfig}+ b_{\timeconfig'} \bra{\timeconfig'} \otimes \bra{\psi_{\timeconfig'}}  } P_{data}^\dagger \Pi P_{data} \Paren{ b_{\timeconfig} \ket{\timeconfig} \otimes \ket{\psi_\timeconfig}  + b_{\timeconfig'} \ket{\timeconfig'} \otimes \ket{\psi_{\timeconfig'}} } \label{eq:t_H}\\
    &= \Paren{ b_{\timeconfig} \bra{\timeconfig} \otimes \bra{\psi_\timeconfig} + b_{\timeconfig'} \bra{\timeconfig'} \otimes \bra{\psi_{\timeconfig}} H} (I \otimes \sigma^\dagger ) \Pi (I \otimes \sigma) \Paren{ b_{\timeconfig} \ket{\timeconfig} \otimes \ket{\psi_\timeconfig}  + b_{\timeconfig'} \ket{\timeconfig'} \otimes  H \ket{\psi_{\timeconfig}} }
  \end{align}
  This expectation vanishes if and only if
  \begin{align}
    \bra{\psi_\timeconfig} H \sigma^\dagger H \sigma \ket{\psi_\timeconfig} = b_{\timeconfig} b_{\timeconfig'}.
  \end{align}
  However, Proposition~\ref{prop:HS_real} implies that $\bra{\psi_\timeconfig}H \sigma^\dagger H \sigma \ket{\psi_\timeconfig}$ is either 0 or purely imaginary. This implies that~\eqref{eq:t_H} does not vanish, and furthermore, it is exactly equal to $1$.
  
  Thus we can evaluate the expectation of $\Pi = H_t[p,q]$ with respect to $P \ket{\psi}$. The expectation $\bra{\psi}P^\dagger \Pi P\ket{\psi}$ is equal to
  \begin{equation}
  \begin{aligned}
&\frac{1}{|\cT|} \sum_{\substack{\timeconfig,\timeconfig' \\ \timeconfig[p,q] = (t,t) \\ \timeconfig \to_{p,q} \timeconfig'}} \Paren{ b_{\timeconfig} \bra{\timeconfig} \otimes \bra{\psi_\timeconfig} + b_{\timeconfig'} \bra{\timeconfig'} \otimes \bra{\psi_{\timeconfig'}} }  P_{data}^\dagger \Pi P_{data} \Paren{ b_{\timeconfig} \ket{\timeconfig} \otimes \ket{\psi_\timeconfig}  + b_{\timeconfig'} \ket{\timeconfig'} \otimes \ket{\psi_{\timeconfig'}} } \\
&= \sum_{\timeconfig: \timeconfig[p,q] = (t,t)} \frac{1}{|\cT|}.
  \end{aligned}
  \end{equation}
  This is equal to the probability that a uniformly random time configuration $\timeconfig$ is such that $\timeconfig[p,q] = (t,t)$. By Lemma~\ref{lem:uniform_clock}, this is at least $1/D^2 = 1/\polylog(N)$.  The cases $\sigma = \sigma_Z$ and $\sigma = \sigma_Y$ can be treated as described above by replacing $H$ with $S$.\\
  
  \item[Case 3.2.] Next, we handle the case of $P_{data} = I$. Since $P \notin \cS$, we have that $P_{flag} \neq S_{flag}$.

  \begin{description}

  \item[Case 3.2.1.] First suppose that $P_{clock} \neq I$. 

  \begin{description}

  \item[Case 3.2.1.1.]
  Suppose there exists a pair of data qubits $(p,q)$ and an index $1 \leq j \leq X$ such that 
  \begin{enumerate}
    \item $p$ and $q$ are neighboring qubits in the circuit $C$ at time at a time $t$ such that $j = t$ or $j = 2X + 1 - t$.
    \item $P_{clock}$ has a $\sigma_Z$ factor acting on $\sC_{p,j}$ but has an identity factor acting on $\sC_{q,j}$. 
  \end{enumerate}
  With probability at least $1/D^2$ over a uniformly random time configuration $\timeconfig$ such that $\timeconfig[p,q] = (j,j)$, we have that $b_{\timeconfig} b_{\timeconfig'} = -1$, where $\timeconfig \to_{p,q} \timeconfig'$. Consider the projector $\Pi = H_j[p,q]$ in $\cD_0$. In order for $\bra{\psi}P^\dagger \Pi P \ket{\psi}$ to vanish, we would need that $\ip{\psi_{\timeconfig'}}{\psi_\timeconfig} = -1$ for all such $\timeconfig$ and $\timeconfig'$, which cannot hold. Therefore $\bra{\psi} P^\dagger \Pi P \ket{\psi} \geq 1/\polylog(N)$.

  \item[Case 3.2.1.2.]
  Now assume that for all pairs of data qubits $(p,q)$ and an index $1 \leq j \leq X$ such that
  \begin{enumerate}
    \item $p$ and $q$ are neighboring qubits in the circuit $C$ at time $t$ for $j = t$ or $j = 2X + 1 - t$, then
    \item $P_{clock}$ has a $\sigma_Z$ factor acting on $\sC_{p,j}$ if and only if $P_{clock}$ has a $\sigma_Z$ factor acting on $\sC_{q,j}$. 
  \end{enumerate}

  We argue that if there is a $\sigma_Z$ factor on $\sC_{p,j}$ then there is a $\sigma_Z$ factor on $\sC_{q,j}$ for all $q$ in $rect(p,j)$ (defined earlier). This is because this is precisely the equivalence class of qubits that are neighbors in times involving qubits in the $j$th layer. 

  Therefore, $P_{clock}$ is a product of $S_{clock}[rect(p,j)]$ for some subset of rectangles meaning $P_{clock} \in \cS$ and we can equivalently consider $P \cdot P_{clock}$ as the logical Pauli error.

  \end{description}

  \item[Case 3.2.2.] Finally, suppose that $P_{clock} = I$ but $P_{flag} \neq I$. Since $P_{flag} \neq S_{flag}$, there exists two data qubits $(p,q)$ such that $P_{flag}$ has a $\sigma_Z$ factor acting on $\sF_{p}$ but has an identity factor acting on $\sF_q$. With probability at least $1/D^2$ over a uniformly random time configuration $\timeconfig$ such that $\timeconfig[p,q] = (X,X)$, we have that $b_{\timeconfig} b_{\timeconfig'} = -1$, where $\timeconfig \to_{p,q} \timeconfig'$. This is because it is the transition from time $t = X$ to $t = X+1$ that the flag qubit on qubits $p$ and $q$ switch from $\ket{0}$ to $\ket{1}$. Consider the projector $\Pi = H_X[p,q]$ in $\cD_0$. In order for $\bra{\psi}P^\dagger \Pi P \ket{\psi}$ to vanish, we would need that $\ip{\psi_{\timeconfig'}}{\psi_\timeconfig} = -1$ for all such $\timeconfig$ and $\timeconfig'$, which cannot hold. Therefore $\bra{\psi} P^\dagger \Pi P \ket{\psi} \geq 1/\polylog(N)$.

  \end{description}

  \end{description}
\end{description}

  \paragraph{Boosting the success probability} We now boost our detection set $\cD_0$ to a stronger set of projectors $\cD$ that can detect errors with very high probability. We leverage the following result of Marriott and Watrous~\cite{marriott2005quantum}:
  
  \begin{lemma}[In-place amplification~\cite{marriott2005quantum}]
  \label{lem:mw}
  Let $\delta > 0 $. Let $A$ be a circuit on $r$ qubits along with $s$ ancilla bits. Then there exists a circuit $A'$ on $r$ qubits and $s' = s + \mathcal{O}(\delta^{-3})$ ancillas, that has size at most $\mathcal{O}(\delta^{-3})$ times the size of $A$, such that the following holds: for all $r$-qubit states $\ket{\varphi}$, if $A$ accepts $\ket{\varphi} \ket{0^s}$ with probability $0$, then $A'$ accepts $\ket{\varphi} \ket{0^{s'}}$ with probability $0$. Otherwise, if $A$ accepts $\ket{\varphi} \ket{0^s}$ with probability at least $\delta$, then $A'$ accepts $\ket{\varphi} \ket{0^{s'}}$ with probability at least $1 - 2^{-1/\delta}$.  
  \end{lemma}
  Here, we define the acceptance probability of the circuit to be the probability that the first qubit measures $\ket{1}$.
  
  For each projector $\Pi \in \cD_0$, we create a ``boosted'' projector $\Pi' \in \cD$ in the following way: let $A$ denote a circuit that performs the projective measurement $\{ \Pi, I - \Pi \}$ and records the outcome in an ancilla qubit. The circuit $A$ acts on $9 + 1$ qubits, and has size $\mathcal{O}(1)$. Let $A'$ be the amplified circuit given by Lemma~\ref{lem:mw} for $\delta = 1/\polylog(N)$. The circuit $A'$ acts on $9$ qubits and $s' = \mathcal{O}(\delta^{-3})$ ancillas. Define the projector $\Pi' = A' (\ketbra{1}{1} \otimes I) (A')^\dagger$ where $\ketbra{1}{1}$ denotes the projection onto the first qubit being in the state $\ket{1}$. 

  Then we have that, for all spacetime codewords $\ket{\psi}$, for all Pauli errors $P \in \cP_N$,
  \begin{enumerate}
    \item If $\Pi P \ket{\psi} = 0$, then $\Pi' (P \ket{\psi} \ket{0^{s'}}) = 0$, and
    \item If $\Tr( \Pi P \psi P^\dagger) \geq 1/\polylog(N)$, then $\Tr( \Pi' (P \psi P^\dagger) \otimes \ketbra{0^{s'}}{0^{s'}}) \geq 1 - 2^{-\polylog(N)}$.
  \end{enumerate} 
  
  Thus we have established that for all non-identity $P \in \cP_N$, there exists a projector $\Pi \in \cD$ such that $\Tr(\Pi \, P \psi P^\dagger) \geq 1 - 2^{-\polylog(N)}$. For general Pauli channels $\cE$, we have that
  \begin{equation}
    \Tr \Paren{ \Pi \, \cE(\psi) } = \sum_P c_P \Paren{ \Pi \, P \psi P^\dagger} \geq \sum_{P \neq I} c_P (1 - 2^{-\polylog(N)}) = (1 - c_I)(1 - 2^{-\polylog(N)}).
  \end{equation}
  This concludes the first part of the Theorem.
  
  We now establish the ``Furthermore'' part of the Theorem. Since the spacetime Hamiltonian is spatially local in $\polylog(N)$ dimensions, and each qubit participates in at most $\polylog(N)$ terms, this implies that the projectors in $\cD$ can be divided into $K = \polylog(N)$ layers $B_1,\ldots,B_K$ such that the projectors in any set $B_j$ act on disjoint sets of qubits. %
  
  The measurement $M$ will consist of measuring the layers $B_1,B_2,\ldots,B_K$ in sequence, and accepting if any of the projectors in the layers accept. Since each projector can be implemented using a size $\polylog(N)$ circuit, each layer measurement can be implemented using a depth $\polylog(N)$ circuit, so, therefore, $M$ can be implemented using a depth $\polylog(N)$ circuit. 
  
  Let $P$ denote a non-identity Pauli operator in $\cP_N$, and let $\ket{\psi}$ be a spacetime codeword. Let $B_j$ be the first layer that contains a projector $\Pi \in \cD$ projector that accepts $P \ket{\psi}$ with positive probability. Then the probability that measuring $M$ rejects on the state $P \ket{\psi}$ is at most the probability that measuring $\Pi$ rejects $P \ket{\psi}$, which is $2^{-\polylog(N)}$. We do not have to worry about the projectors in earlier layers, because by definition they reject $P \ket{\psi}$ with certainty, and leave the state unchanged.
  
  We can extend this argument to a general Pauli channel $\cE$ in the same way as before, and this completes the proof of the Theorem.
\end{proof}

\section{Alternate constructions and spatial locality}
\subsection{Good approximate QLDPC from weighted FK Hamiltonians}
\label{sec:globalFK}
In this section, we describe another closely related construction of approximate LDPC codes, which is based on using the standard Feynman-Kitaev construction with a global clock as well as recently-introduced variants that increase the overlap of the history state with the beginning and end of the computation~\cite{Bausch2018analysislimitations,caha2018clocks}.  The primary advantage of this version of the construction is the significantly simpler analysis of the spectral gap, even in the presence of nonuniform weight distributions on the time steps of the computation.  The main disadvantage for this version of the construction is that the increase in energy caused by local errors is significantly reduced in some cases (thereby making them more difficult to detect).  In the global clock construction, there are local errors with expected energy scaling like $1/T$ where $T$ is the (polynomial) size of the computation, instead of errors having energy $1/D$ in the spacetime construction where $D$ is the (polylogarithmic) depth.  This can be seen as a fulfillment of the intuition that the spacetime circuit-to-Hamiltonian is more robust than its global clock counterpart.  Finally, due to the simplification in the analysis for the global clock construction, we can achieve a provably optimal tradeoff between the approximate error $\varepsilon$ of the code and the soundness $s$, using a result that was previously established in~\cite{Bausch2018analysislimitations}.  

\paragraph{Result}For any $\varepsilon(N) > 0$ there exists an $[[N,k,d,\varepsilon,\ell,s]]$ approximate LDPC code with $k = \Omega(N/\log^5 N)$, $d = \Omega(N/\log^5 N)$, $\ell = 5$, $s = \polylog(N)$ .  The spectral gap of the code Hamiltonian is 
\begin{equation}
\Delta_H = \Omega \left(\frac{\varepsilon(N)}{N^3 \polylog(N)} \right).
\end{equation}
The encoding circuits analyzed by Brown and Fawzi with depth $D = \mathcal{O}(\log^3 n)$ have size $T = \mathcal{O}(n \log^2 n)$.  To ensure that only a polylogarithmic number of Hamiltonian terms act on each physical qubit we consider the same sequence of random Clifford gates interspersed with bitonic sorting circuits of Section \ref{subsection:bitonicblock},
\begin{equation}
  \prod_{t=1}^D \bigotimes_{(p,q) \in L_t} U_t[p,q]
\end{equation}
but now the local gates are each applied individually in sequence.  Note that in this section we do not reverse the application of the gates, and the time register is not periodic.  For each layer $L_t$ we choose an ordering $(p_1,q_1) ... (p_{n/2},q_{n/2})$ for the pairs of qubits interacting within that layer, and we re-index this sequence of $T = n D / 2$ gates as simply $U_T...U_1$,
\begin{equation}
   U_T...U_1 = U_t[q_{n/2},p_{n/2}]...U_2[p_1,q_1] U_1[p_{n/2},q_{n/2}] ... U_1[p_1,q_1]
\end{equation} 
The code space, which will be the ground space of the code Hamiltonian, is
\begin{equation}
\mathcal{C}  = \left \{  \sum_{t=0}^{T + n} \sqrt{\pi_t} \ket{t}_C \otimes U_{t-n}...U_1 \ket{\psi, 0,...,0}_S : \ket{\psi} \in \mathds{C}^{n-k} \right \}  \label{eq:globalCodeSpace}
\end{equation}
where by convention we define $U_{t - n} ... U_1 = \mathds{1}$ for $t < n$, and the distribution $\pi$ is defined by
\begin{equation}
\pi_t = \begin{cases}
\frac{\varepsilon}{T + n}\quad ,  \quad 0 \leq t < T + n\\
1- \varepsilon \quad , \quad t = T + n\\
\end{cases}.\label{eq:weights}
\end{equation} 
Instead of the standard $H_{in}$ of the form,
\begin{equation}
\ket{0}\bra{0}_C \otimes \mathds{1}_{S_{1...n-k}} \otimes \left(\sum_{r = n - k}^{n } \ket{1}\bra{1}_{S_r} \right).
\end{equation}
we use a "staggered" version of the input check,
\begin{equation}
H_{\textrm{in}} = \sum_{r = n-k}^n  \ket{r - n + k}\bra{r - n + k}_C \otimes I_{S_{1...n-k}} \otimes \ket{1}\bra{1}_{S_R}.
\end{equation}
The point of the staggered input check is to avoid having a nonconstant number of terms acting on the clock bits that represent $t = 0$.  

The propagation Hamiltonian for this nonuniform distribution over time steps is based on the method used in~\cite{Bausch2018analysislimitations}.  One first considers the Markov chain with Metropolis transition probabilities (see~\cite{levin2017markov} for a general background on Markov chains) from $t$ to $t-1$, $t+1$ that is reversible with respect to $\pi$.  For $0 < t < T - n$ we have
\begin{equation}
   P_{t,t+1} = \frac{1}{4} \min\left\{1,\frac{\pi_{t+1}}{\pi_t}\right \} \; , \;  P_{t,t-1}  = \frac{1}{4}\min\left\{1,\frac{\pi_{t-1}}{\pi_t}\right \} \; , \;  P_{t,t} = 1 -  P_{t,t+1} -  P_{t,t-1}\label{eq:MetropolisTransitionProbabilities}
\end{equation}
and also
\begin{equation}
P_{0,t} = \frac{1}{2}  \min\left\{1,\frac{\pi_{1}}{\pi_0}\right \} \quad , \quad P_{0,0} =1 - P_{0,1}
\end{equation}
and
\begin{equation}
P_{T + n, T + n -1} = \frac{1}{2}  \min\left\{1,\frac{\pi_{T + n -1 }}{\pi_{T + n}}\right \} \quad , \quad P_{T + n, T+n} =1 - P_{T+n,T + n - 1}.
\end{equation}
The transition probabilities satisfy $\pi_t P_{t,t'} = \pi_{t'} P_{t',t}$ for all $0 \leq t,t' \leq T + n$, and so $\sum_{t = 0}^{T + n} \pi_t P_{t,t'} = \pi_t'$.  Therefore the propagation Hamiltonian defined by $H_{\textrm{prop}} = \sum_{t = 0}^{T + n} H_{\textrm{prop}}(t)$ with
\begin{align}
H_{\textrm{prop}}(t) &= \frac{1}{2}\bigl( P_{t,t}\ket{t}\bra{t}_C\otimes \mathds{1}_S + P_{t-1,t-1}\ket{t}\bra{t}_C \otimes \mathds{1}_S \\ &-  \pi_t^{1/2} \pi_{t-1}^{-1/2} P_{t,t-1}\ket{t}\bra{t-1}_C\otimes U_{t - n} -  \pi_{t-1}^{1/2} \pi_{t}^{-1/2} P_{t-1,t}\ket{t-1}\bra{t}_C\otimes U_{t - n}^\dagger \bigr),
\end{align}
is such that $H_{\textrm{in}} + H_{\textrm{prop}}$ has the ground space $\mathcal{C}$ in \eqref{eq:globalCodeSpace} as claimed. 

The locality of the checks $\ell = 5$ follows from the fact that when the clock register is implemented with qubits as in \eqref{eq:clock} the local terms in $H$ are at most $5$-local, and this is unaffected by the modified coefficients in the propagation terms.  The error bound of $\varepsilon$ for the code follows from the same argument used in Section \ref{sec:construction} together with the fact that the distribution $\pi$ assigns a probability of $1 - \varepsilon$ to the final time step of the computation.  

To obtain the bound $s = \polylog(N)$ on the number of check terms acting on each physical qubits, we first consider the clock bits.  For $n + 2 \leq t \leq T + n -2$ there are 5 terms acting on clock bit $t$, 
\begin{equation}
H_{\textrm{prop}}(t-2) , H_{\textrm{prop}}(t-1) ,  H_{\textrm{prop}}(t) , H_{\textrm{prop}}(t+1) , H_{\textrm{prop}}(t+2)
\end{equation}
and for $t = n -1,T+n - 1,T + n$ the number of propagation terms is even fewer.   For $ 0 \leq t \leq n$, each clock bit participates in one term from $H_{\textrm{in}}$ and at most 5 terms from $H_{\textrm{prop}}$.  Finally, the number of nontrivial gates acting on each system qubit is at most $D = \textrm{polylog}(N)$. 

The spectral gap of $H_{\textrm{prop}}$ is the same as the spectral gap of the Markov chain $P$ described above, which can be lower bounded by Cheeger's inequality.  Since $1 - o(1)$ of the weight in the stationary distribution is concentrated on the final time step $t = T + n$, the subset $S = \{0,...,T + n -1\}$ has the minimum conductance which is
\begin{equation}
\Phi = \frac{1}{\pi(S)} \sum_{t \in S , t \notin S} \pi_t P_{t,t'} =  \frac{1}{\epsilon}\left(\frac{\epsilon}{T + n}\right) P_{T + n -1, T+n} = \frac{1}{4(T + n)}
\end{equation}
and since $T = n \polylog(n)$ we have
\begin{equation}
\Delta_{H_{\textrm{prop}}} = \Omega \left(\frac{1}{n^2 \textrm{polylog}(n)} \right).
\end{equation}
To go from $\Delta_{H_{\textrm{prop}}}$ to $\Delta_H$ we apply the same argument as in Section \ref{sec:analysis}, and use the geometrical lemma to obtain
\begin{equation}
\Delta_H = \Omega \left(\frac{\varepsilon}{n^3 \textrm{polylog}(n)} \right).
\end{equation}
Finally, we note that because of the dual importance of the spectral gap and the overlap with the endpoint of the computation, which respectively determine the soundness of the code and the infidelity of recovery, the optimality of the distribution $\pi$ in \eqref{eq:weights} follows from Theorem 8 in~\cite{Bausch2018analysislimitations}.
\begin{theorem}[\cite{Bausch2018analysislimitations}]
Let $|\psi\rangle$ be the ground state of a Hamiltonian $H$ with eigenvalues $E \defeq E_0 \leq E_1 \leq \ldots \leq E_T$.  If $H$ is tridiagonal in the basis $\{|0\rangle,\ldots,|T\rangle\}$,
\begin{equation}
H \defeq \sum_{t = 0}^T a_t |t\rangle \langle t | + \sum_{t =0}^{T-1} \left ( b_t|t +1\rangle \langle t | + b^*_t|t\rangle \langle t + 1| \right),
\end{equation}
with $|a_t|, |b_t| \leq 1$ for $t = 0,\ldots,T$ then the product $\Delta_{H} \cdot \min \{ |\psi|^2_0,  |\psi_T|^2\} $ of the spectral gap $\Delta_{H} = E_1 - E$ and the minimum endpoint overlap is $\mathcal{O}(T^{-2})$.  
\end{theorem}

\subsection{Spatial locality of the Hamiltonian}
\label{sec:spatialLocality}
In this section, we demonstrate that the code Hamiltonian is indeed $\polylog(n)$-spatially local. We also provide a sketch of how to make the construction $\mathcal{O}(\log n)$-spatially local at the cost of increasing the locality of the Hamiltonian from 9 to 15.

We give a technical definition for spatial locality that fits the previous descriptions given in other works \cite{bravyi2009no,flammia2017limits}.

\newcommand{\emb}{\mathrm{emb}}

\begin{definition}
A code defined by a local Hamiltonian $H = \sum_i H_i$ is $d$-spatially local if there exists an embedding map $\emb: Q \rightarrow \mathbb{R}^d$, where $Q$ is the set of qubits, satisfying the following conditions:
\begin{enumerate}
  \item For all $q_1 \neq q_2 \in Q$, $\norm{\emb(q_1) - \emb(q_2)}_2 \geq 1$.
  \item Let $Q_i \subseteq Q$ be the set of qubits acted on non-trivially by Hamiltonian $H_i$. There exists a constant $c > 0$ such that for all $q_1, q_2 \in Q_i$,
  \begin{equation}
    \norm{\emb(q_1) - \emb(q_2)}_2 \leq c.
  \end{equation}
\end{enumerate}
\end{definition}

We propose such an embedding for $d = \mathcal{O}(\log^5 n)$. First, consider the interaction graph of the qubits in a bitonic block architecture $\cB_\ell$ for $\ell = \log n$ (there exists an edge between two qubits if they share a gate). We note the following lemma:

\begin{lemma}
The incidence graph of the bitonic block architecture $\cB_\ell$ is equivalent to the $\ell$-dimensional hypercube.
\end{lemma}

\begin{proof}
The result is easy to see for $\ell = 1$. For $\ell > 1$, notice that layer $\cL_1$ connects matching vertices in two bitonic blocks $\cB_{\ell - 1}$ (Corollary \ref{cor:subbitonic}), which by induction yields a $\ell$-dimensional hypercube.
\end{proof}
Therefore, there is an encoding $h: [n] \rightarrow \mathbb{R}^\ell$ such that if qubits $i$ and $j$ interact then $\norm{h(1) - h(2)}_2 = 1$. Let $e_j$ denote the $j$-th standard basis vector.

\begin{theorem}
The code defined in this paper is $\mathcal{O}(\log^5 n)$-spatially local.
\end{theorem}

\begin{proof}
We provide the explicit embedding map and prove it satisfies the definition. Our embedding map can be seen as
\begin{equation}
\emb: Q \rightarrow \mathbb{R}^\ell \times \mathbb{R}^1 \times \mathbb{R}^X
\end{equation}
defined by
\begin{align}
\textbf{Data registers}  &\qquad \emb(\sS_i) = (h(i), 0, 0^X), \\
\textbf{Flag registers}  &\qquad \emb(\sF_i) = (h(i), 1, 0^X), \\
\textbf{Clock registers} &\qquad \emb(\sC_{i,j}) = (h(i), 0, e_j).
\end{align}
Note that $\ell + 1 + X = \log n + 1 + \mathcal{O}(\log^5 n) = \mathcal{O}(\log ^5 n)$.

Every coordinate of every qubit is either $0$ or $1$ and clearly the qubits are distinct. Therefore, the minimal distance between them is indeed 1. We now verify that each of the Hamiltonian terms act on qubits that are only a constant distance of $\sqrt{3}$ apart.
\begin{description}
\item[$H_{clock}$ terms.] All $H_{clock}$ terms are projections of the form $\Pi_{\sC_{i,j},\sC_{i,j+1}}^{(01)}$. Any pair of clock qubits $\sC_{i,j}$ and $\sC_{i,j'}$ are distance $\sqrt{2}$ in the $\ell_2$ norm.

\item[$H_{init}$ terms.] All $H_{init}$ terms are projections of the form $\Pi_{\sC_{i,0}}^{(1)} \otimes \Pi_{\sS_i}^{(1)}$. We note that for any $i$, any clock qubit $\sC_{i,j}$ is distance $\sqrt{2}$ in the $\ell_2$ norm from $\sS_i$.

\item[$H_{prop}$ terms.] All $H_{prop}$ terms are interactions between the clock, flag, and data qubits for qubits $p$ and $q$ that share a gate $U_t[p,q]$ in the circuit. We note that $h(p)$ and $h(q)$ have Hamming distance 1 (i.e. differ in only one coordinate). The specific set of qubits involved are $\sS_p, \sS_q, \sF_p, \sF_q$, as well as six clock qubits, three $\sC_{p,\cdot}$ and three $\sC_{q, \cdot}$ (the exact collection depends on $t$ according to (\ref{eq:cases_for_h_prop})). It is not difficult to see that any two of the embeddings of these qubits differ in at most 3 coordinates, hence a distance of $\sqrt{3}$ in the $\ell_2$ norm.

\item[$H_{causal}$ terms.] A term in $H_{causal}$ compares the clock of a qubit $p$ to an adjacent qubit $q$. It will involve the flag qubit $\sF_p$ and up to two clock qubits $\sC_{p,a}, \sC_{p,a+1}$ (again, $a$ depends on $t$; see (\ref{eq:cases_for_h_prop})). In addition, it checks the flag qubit $\sF_q$ and up to two clock qubits $\sC_{q,b}, \sC_{q,c}$ (here $b$ and $c$ depend on $t$; see the case-wise definition of $H_{causal}$). Likewise, it is not difficult to see that any two of these embeddings of these qubits differ in at most 3 coordinates, hence a distance of $\sqrt{3}$ in the $\ell_2$ norm.
\end{description}
\end{proof}

\subsection{Alternate construction}

We now present an alternate construction which improves the spatial locality of the Hamiltonian. We only provide a sketch of the construction as the majority of the analysis is similar to that presented in the main sections of the paper. In particular, we will demonstrate that a different representation of the time register can be used to make this code $\mathcal{O}(\log n)$-spatially local at the cost of worsening the code to being 50-local.

Instead of encoding the time register using a flag and a domain wall, we will encode it using the multi-dimensional clock method used in \cite{nirkhe_et_al:LIPIcs:2018:9095}. At a high level, we express time in its unique representation in base $\bar{D} = \lceil \sqrt[6]{D} \rceil + 1$ and represent each coordinate of the representation using a flag and domain wall.

More specifically, additionally let $\bar{X} = \frac{\bar{D} - 2}{2}$. Then $\bar{D} = \mathcal{O}(\log n)$ for our construction. For any number $0 \leq j \leq D$, let $a_0, \ldots, a_5$ be the unique numbers $\in \{0, \ldots, \bar{D} - 1\}$ such that
\begin{equation}
j = a_0 + a_1 \bar{D} + \ldots + a_5 \bar{D}^5.
\end{equation}
We then express $\ket{j}_{\sT_i}$ as
\begin{equation}
\ket{j}_{\sT_i} = \ket{a_0}_{\sT_i^{(0)}} \otimes \ket{a_5}_{\sT_i^{(5)}}
\end{equation}
where $\ket{a_k}_{\sT_i^(k)}$ is an encoding with a flag and domain wall of times between $\{0, \ldots, \overline{D} - 1\}$ as described in the main section of the paper. $\sT_i^{(k)}$ consists of a flag register $F_i^{(k)}$ and $\{ C_{i,j}^{(k)}\}_{j = 0}^{\bar X}$.

This adjustment will require the Hamiltonian terms to act on 6 times as many time registers as before; hence the 50-locality. The encoding to demonstrate $\mathcal{O}(\log n)$-spatial locality is

\begin{align}
\textbf{Data registers}  &\qquad \emb \left(\sS_i \right) = (h(i), 0^6, 0^{6\bar{X}}), \\
\textbf{Flag registers}  &\qquad \emb\left(\sF_i^{(k)}\right) = (h(i), e_k, 0^{6\bar{X}}), \\
\textbf{Clock registers} &\qquad \emb\left(\sC_{i,j}\right) = (h(i), e_k, \e_k \otimes e_j).
\end{align}

\section*{Acknowledgments}
We thank Winton Brown, Aram Harrow, and Umesh Vazirani for helpful discussions.
Author TB acknowleges support from NSERC through a PGSD award.
Author CN is supported by ARO Grant W911NF-12-1-0541 and NSF Grant CCF-1410022.
Part of this work was completed while authors EC, CN, and HY were visitors at the Simon's Institute 2018 summer cluster \emph{Challenges in Quantum Computation}.

\bibliography{references}
\bibliographystyle{alpha}

\appendix

\section{Partially applied configurations of a bitonic sorting circuit}
\label{appendix:bitonic}

In this Appendix, we provide the mathematical foundations required for analyzing the spectral gap of the Hamiltonian and generating the encoding circuit. We explore the properties of the bitonic block \cite{Batcher:1968:SNA:1468075.1468121} (see Definition \ref{def:bitonicblock}) and prove results about the space of \emph{valid configurations} of partial computations of a bitonic block (see Definition \ref{def:partial-config}).

\subsection{Configurations and width}

We study the structure and combinatorics of the valid configurations (see Definition \ref{def:partial-config}).
One can think of a valid partial configuration as being represented visually on the architecture as a string which partitions the architecture into two halves: gates that have and have not been applied. Depending on the configuration of the circuit in question, there is a maximum \emph{width}, a number of layers, that such a string can have.

\begin{definition}
The \emph{width} of a partial configuration $\tau$ is 
\begin{equation}
w(\tau) = \max(\tau) - \min(\tau).
\end{equation}
\end{definition}

\begin{lemma}
  For bitonic block $\cB_\ell$,
  \begin{equation}
    w(\tau)< \ell
  \end{equation}
  for any valid partial configuration $\tau$.
\label{lem:width}
\end{lemma}
\begin{proof}
  Suppose that we have a configuration of width $\ell$, then at least one gate in the final layer $\cL_\ell$ must be applied (and corresponding qubits are at time $t = \ell$), and at least one gate in the first layer $\cL_1$ must not have been applied (so that its corresponding qubits are at time $t=0$). Let $g$ be a gate in $\cL_\ell$ that has been applied. 

  Consider the light cone $\Lambda$ of the gate $g$; there are two gates of $\cL_{\ell - 1}$ in $\Lambda$. Precisely, these are the two gates of the $\cB_2$ block which connect the $\cB_1$ block containing $g$ to its neighboring $\cB_1$ block (recall Definition \ref{def:bitonicblock}). Likewise, there are $2^j$ gates of layer $\cL_{\ell - j}$ in $\Lambda$ as they are the $2^{j-1}$ gates of a $\cB_j$ block connecting the block $\cB_{j-1}$ to its neighboring $\cB_{j-1}$ block.

  Carrying this until the first layer, there are $2^{\ell-1}$ gates of $\Lambda$ in $\cL_1$. Since all of $\Lambda$ must be applied, every gate of $\cL_1$ is applied. Therefore, the assumption of a configuration of width $\ell$ is false.

\end{proof}

The proof of Lemma \ref{lem:width} illustrates an interesting and important property of bitonic blocks; the light cone of any gate in the architecture doubles in size each layer. Additionally,

\begin{corollary}
Any valid configuration of a bitonic block $\cB_\ell$ must satisfy at least one of the following:
\begin{enumerate}
\item Every gate in layer $\cL_1$ is activated.
\item Every gate in layer $\cL_\ell$ is \emph{not} activated.
\end{enumerate}
\label{cor:dictotomyofactivation}
\end{corollary}

\subsection{Permutations and the splitting property}\label{sec:permutationSplitting}

In this subsection, we demonstrate some important combinatorial properties of bitonic blocks.

\begin{fact}
  The number of valid configurations of any architecture is invariant under permutation of the qubit labels.
\end{fact}

\begin{definition}
Two architectures $\cA_1$ and $\cA_2$ both acting on $n$ qubits are called \emph{isomorphic} (denoted $\cA_1 \simeq \cA_2$) if one can be obtained from the other by some permutation of the qubit labels.
\end{definition}

Therefore, isomorphic architectures have the same number of valid configurations.

\begin{fact}
\label{fact:penultimatelayers}
For a bitonic block $\cB_\ell$, the first $\ell - 1$ layers, $\{\cL_1, \cL_2, \ldots, \cL_{\ell - 1}\}$, can be viewed as $\cB_{\ell-1}^{\otimes 2}$ where the first smaller block acts on odd indexed qubits and the second on even indexed qubits.
\end{fact}

\begin{proof}
It is easy to see for $\ell = 2$. For $\ell > 2$, by induction, we know the odd indexed qubits in layers $\{\cL_2, \ldots, \cL_{\ell - 1}\}$ form $\cB_{\ell - 2}^{\otimes 2}$. Recall that $\cL_1$ contains gates between qubit $i$ and $i + 2^{\ell - 1}$ for $i \leq 2^{\ell - 1}$. This produces a block $\cB_{\ell - 1}$ on the odd indexed qubits; a similar argument holds for even indexed qubits.
\end{proof}

\begin{lemma}
A single layer left cyclic shift of the layers of a bitonic block $\mathcal{B}_\ell$ is isomorphic to the bitonic block $\mathcal{B}_\ell$. The isomorphism is described by the permutation $\pi_\ell$:
\begin{equation}
\pi_\ell(i) \defeq \begin{cases} 2i - 1 & \textrm{if } i \leq 2^{\ell-1} \\ 2i - 2^\ell & \textrm{if } i > 2^{\ell-1}. \end{cases}
\end{equation}
A single layer right cycle shift isomorphism is described by the permutation $\pi_\ell^{-1}$:
\begin{equation}
\pi_\ell^{-1}(i) \defeq \begin{cases} \frac{i+1}{2} & \textrm{if } i \textrm{ odd} \\
\frac{i + 2^\ell}{2} & \textrm{if } i \textrm{ even}.\end{cases}
\end{equation}
\label{lem:permutation}
\end{lemma}

Figure \ref{fig:colorfulPermutation} illustrates the above permutations for $\ell=3$.

\begin{proof}
By Fact \ref{fact:penultimatelayers}, the first $\ell - 1$ layers of a bitonic block $\cB_\ell$ form $\cB_{\ell-1}^{\otimes 2}$ where one is the collection of odd indexed rows and other the collection of even indexed rows. By the definition, the last $\ell-1$ layers form $\cB_{\ell-1}^{\otimes 2}$ where one is the collection of the first $2^{k-1}$ rows and the other the collection of $2^{\ell-1}$ rows. Therefore, any permutation for a single layer left cyclic shift must permute these $\cB_{\ell-1}$ blocks onto each other. It then is easy to check that the permutation $\pi$ will also send the first layer to the last layer. The single layer right cycle shift is just the inverse permutation, $\pi^{-1}$.
\end{proof}

\begin{figure}[h]
\begin{center}
\includegraphics[scale=.4]{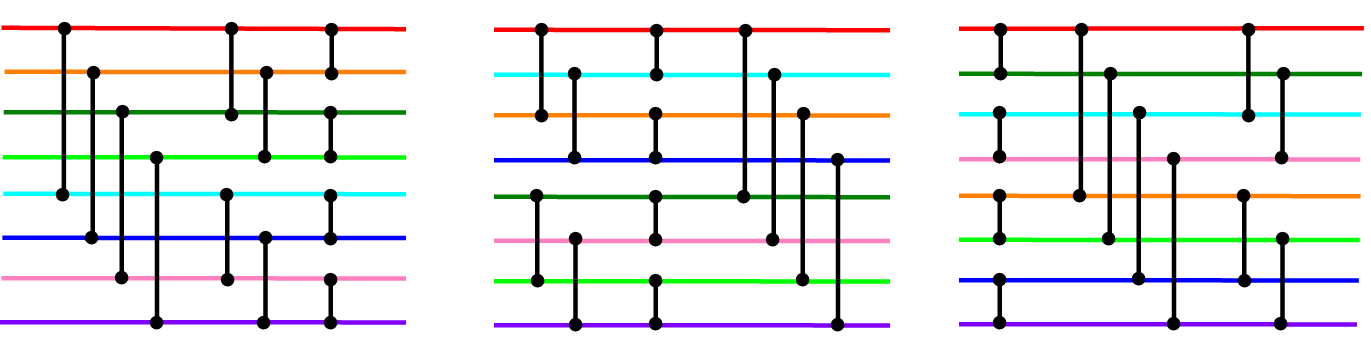}
\end{center}
\caption{A demonstration of Lemma \ref{lem:permutation}. The colored wires represent the permutation mapping each of the shifted bitonic blocks $\cB_3$ back to the original bitonic block. 
\label{fig:colorfulPermutation}}
\end{figure}

\begin{corollary}
A $j$ layer left (or right) cyclic shift of the layers of a bitonic block $\cB_\ell$ is isomorphic to a bitonic block $\cB_\ell$. The permutation describing the isomorphism is $\pi_\ell^j$ (or $\pi_\ell^{-j}$).
\end{corollary}

This yields the following important corollary.

\begin{corollary}
Let $\cA_\ell^j$ be a sub-architecture formed by taking any distinct $j$ layers (in any order). Then,
\begin{equation}
\cA_\ell^j \simeq \bigotimes_{i = 1}^{2^{k-j}} \cB_j.
\end{equation}
\label{cor:subbitonic}
\end{corollary}

\begin{proof}
Consider any excluded layer. By the previous corollary, we can assume it to be the first layer. Thus, the remaining layers decompose into the tensor product of smaller bitonic blocks. We can repeat for each excluded layer.
\end{proof}

\subsection{Counting configurations}

\subsubsection{Configurations of a bitonic block}\label{sec:countingBitonicBlock}

We now recursively count the number of valid configurations of a block $\cB_\ell$. This will be useful in the encoding circuit and the spectral gap analysis.

\begin{theorem}\label{thm:countingBitonicBlock}
  Let $a_\ell$ be the total number of valid partial configurations of $\mathcal{B}_\ell$. This number is described by the recurrence relation\footnote{This recurrence relation does not have a known solution. It is, however, known to scale as
\begin{eqnarray}
  a_\ell \sim \frac{\omega^{2^\ell}}{\phi}
\end{eqnarray}
where $\phi$ is the golden ratio, and $\omega=1.8445\ldots$, a number with no known form.}
  \begin{eqnarray}
    a_\ell \defeq 2a_{\ell-1}^2-a_{\ell-2}^4,
  \end{eqnarray}
  with initial conditions $a_1=2$, $a_2=7$.
\label{thm:numbconfigs}
\end{theorem}

\begin{proof}
The initial cases can be counted by hand. For $\ell > 2$, by Corollary \ref{cor:dictotomyofactivation}, we know that the first layer is entirely activated or the last layer is entirely not activated. Corollary \ref{cor:subbitonic}, tells us that, in either case, the remaining layers are isomorphic to $\cB_{\ell - 1}^{\otimes 2}$. Therefore, aside from double-counting between the two cases, there are $2a_{\ell - 1}^2$ valid configurations. The set of double counted configurations are all configurations that lie entirely in the middle $\ell - 2$ layers. Again we apply Corollary \ref{cor:subbitonic}, to argue that this set of layers is isomorphic to $\cB_{\ell - 2}^{\otimes 4}$, and therefore, has $a_{\ell - 2}^4$ valid configurations.
\end{proof}

\begin{corollary}
The total number of valid partial configurations of $\cB_\ell$ with some gate in layer $\cL_1$ not activated is $a_\ell - a_{\ell - 1}^2$.
\label{cor:gatenotactivated}
\end{corollary}

\begin{proof}
We need to ignore the valid partial configurations which have the entire $\cL_1$ layer activated. The last $\ell - 1$ layers are isomorphic to $\cB_{\ell - 1}^{\otimes 2}$ by Corollary \ref{cor:subbitonic}.
\end{proof}

\subsubsection{Configurations of products of bitonic blocks}
\label{subsec:countingconfigsofproducts}

Consider an architecture composed of $m$ consecutive copies of a bitonic block of rank $\ell$.
\begin{definition}[Linear product of bitonic blocks]
\begin{equation}
\cB_\ell^{\times m} \defeq \prod_{i = 1}^m \cB_\ell.
\end{equation}
\label{def:linearproduct}
\end{definition}

\begin{theorem}
The total number of configurations of $\cB_\ell^{\times m}$ is
\begin{equation}
a_\ell^{\times m} \defeq \left(\left(m-1\right)\ell+1\right)a_\ell-(m-1)\ell a_{\ell-1}^2.
\end{equation}
\label{thm:configcountprod}
\end{theorem}

\begin{proof}
Notice, that \emph{any} $\ell$ consecutive layers -- henceforth called a window -- of $\cB_\ell^{\times m}$ is isomorphic to a bitonic block $\cB_\ell$. Then by Lemma \ref{lem:width}, we know that any valid configuration is contained within a window. We can, therefore, count the number of valid configurations by considering the first window it appears in. For every window except the last, all configurations corresponding to the window must have some gate in the first layer not activated; otherwise, they would correspond to a later window. By Corollary \ref{cor:gatenotactivated}, there are $a_\ell - a_{\ell - 1}^2$ configurations for every window except the last. For the last, there are no restrictions, so there are $a_{\ell}$ configurations. It is easy to see that there are $(m-1)\ell + 1$ windows. Then,
\begin{equation}
\begin{aligned}
a_\ell^{\times m} = (m-1)\ell (a_\ell - a_{\ell - 1}^2) + a_\ell = \left(\left(m-1\right)\ell+1\right)a_\ell-(m-1)\ell a_{\ell-1}^2.
\end{aligned}
\end{equation}
\end{proof}

\begin{definition}
Let $\cB_\ell^{\leftrightarrow m}$ be the circular architecture defined by taking $m$ copies of the bitonic block and wrapping it around the cylinder.
\label{def:circularwrapping}
\end{definition}

\begin{theorem}
The total number of configurations of $\cB_\ell^{\leftrightarrow m}$ is
\begin{equation}
a_\ell^{\leftrightarrow m} \defeq (a_\ell - a_{\ell - 1}^2) m \ell.
\end{equation}
\label{thm:countcircular}
\end{theorem}

\begin{proof}
We can consider a similar argument as that of Theorem \ref{thm:configcountprod}. In this case, there are $\ell m$ windows as windows can wrap around the circular architecture. Since we identify each configuration with the first window containing it, every window must have some gate in the first layer not activated. Each layer of the architecture can be the start of a window since the architecture is circular. Therefore, there are $m \ell$ windows, completing the proof.
\end{proof}

We now provide the proof of Lemma \ref{lem:overlap} which was omitted from the main article.

\begin{proof}[Proof of Lemma \ref{lem:overlap}]
Let $\cT_{comp}$ be the set of valid configurations for whom all clocks were past $D_1 \ell^2$. Since all the gates in the circuit past time $D_1 \ell^2$ are identity gates, $\ket{\psi_\timeconfig}$ is constant. The subcircuit of identity gates has a depth of $(3/\eps - 1)D_1 \ell^2$ depth. By Lemma \ref{lem:width}, we know that any valid configuration has a width of at most $\ell$ and by the counting argument of Lemma \ref{lem:circularDensity} and Theorem \ref{thm:countcircular}, we know that we can get a lower bound on $\abs{\cT_{comp}}/\abs{\cT}$ by counting the fraction of windows purely contained in the subcircuit of identity gates. To avoid any configurations that cross outside the region of identity gates, we will ignore the first and last $D_1 \ell^2$ gates. Then the fraction of windows purely contained is at least
\begin{align}
\frac{\left(\frac{3}{\eps} - 1\right)D_1 \ell^2 - 2 D_1 \ell^2}{\frac{3 D_1 \ell^2}{\eps}} = 1 - \eps.
\end{align}
\end{proof}

\subsubsection{Configurations overlapping the initial state}

\begin{lemma}
Let $i$ be a fixed qubit. Then the number of valid configurations of $\cB_{\ell}^{\times m}$ such that the clock of qubit $i$ is at 0 is $\prod_{j = 1}^{\ell - 1} a_j$.
\end{lemma}

\begin{proof}
We only need consider the first block $\cB_\ell$ of $\cB_\ell^{\times m}$ due to Lemma \ref{lem:width}. By Corollary \ref{cor:dictotomyofactivation}, we know that no gate in the last layer of $\cB_\ell$ is activated. Therefore, we only need to consider the first $\ell - 1$ gates which are isomorphic to $\cB_{\ell - 1}^{\otimes 2}$ (Corollary \ref{cor:subbitonic}). The block $\cB_{\ell - 1}$ corresponding to the set of qubits of which $i$ is not a member has $a_{\ell - 1}$ valid configurations. The set containing $i$ can be recursively seen to have $\prod_{j = 1}^{\ell - 2} a_j$ valid configurations.
\end{proof}

\begin{lemma}\label{lem:circularDensity}
Let $i$ be a fixed qubit. Then the number of valid configurations of $\cB_\ell^{\leftrightarrow m}$ such that the clock of qubit $i$ is at 0 is $(a_\ell - a_{\ell - 1}^2)$.
\end{lemma}

\begin{proof}
By symmetry (Corollary \ref{cor:subbitonic}), this corresponds to one of the $m\ell$ windows described in Theorem \ref{thm:countcircular}.
\end{proof}

\subsection{Isomorphism with dyadic tilings}

\subsubsection{Dyadic Tilings}

The numbers $a_{\ell}$ count the number of valid partial circuit configurations of $\mathcal{B}_{\ell}$, but they also happen to enumerate a different combinatorial structure: the number of dyadic tilings of the unit square of rank $\ell$ \cite{cannon_et_al:LIPIcs:2017:7583}. To facilitate the analysis of the gap of our code Hamiltonian, we will describe an explicit isomorphism between the two sets and the Markov chains defined on them.

\begin{definition}[Dyadic tiling]
A dyadic tiling of rank $\ell$ is a tiling of the unit square by $2^{\ell}$ \emph{equal-area} dyadic rectangles, which are rectangles of the form $[a 2^{-s}, (a+1)2^{-s}]\times [b 2^{-t},(b+1) 2^{-t}]$, where $a,b,s,t$ are nonnegative integers for some positive integer $\ell$.
\end{definition}

Figure \ref{fig:dyadicex} shows some examples for $\ell = 4$:
\begin{figure}[h]
\begin{center}
\includegraphics[scale=.35]{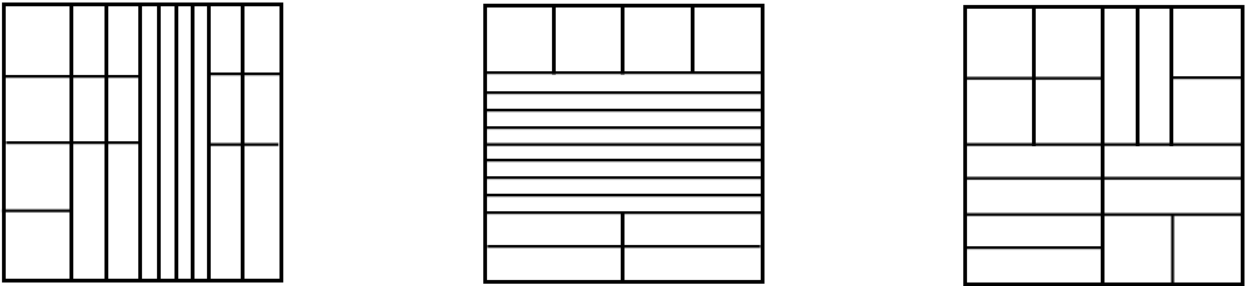}
\caption{Examples of rank 4 dyadic tilings. \label{fig:dyadicex}}
\end{center}
\end{figure}

Each tiling of rank $\ell$ can be described recursively: beginning from the unit square, draw a line that is either a horizontal or vertical bisector. This divides the square into two rectangles of equal area.%
Then, choose two (not necessarily distinct) dyadic tilings of rank $\ell-1$ and scale them to overlay with the two rectangles.

\begin{definition}
Let $\cT_{\ell}$ be the set of dyadic tilings of rank $\ell$.
\end{definition}

There is a natural Markov chain on $\cT_{\ell}$ called the \emph{edge-flip} Markov chain \cite{randomdyadictilingsoftheunitsquare}. Given a dyadic tiling, there is a distinguished set of edges in the tiling which can be removed and replaced by their perpendicular bisector to obtain another valid dyadic tiling of the same size. So, the transitions between states of the Markov chain are described by choosing one of the \emph{flippable} edges uniformly at random and flipping it, obtaining a new tiling.

We can formally define the edge-flip Markov chain as follows:

\begin{definition}[Edge-flip Markov chain \cite{randomdyadictilingsoftheunitsquare,cannon_et_al:LIPIcs:2017:7583}]

The edge-flip Markov Chain, $\mathcal{M}_{\ell}$ on state space $\cT_{\ell}$ is defined with the following transition rule. Starting from state $m_i \in \cT_{\ell}$:
\begin{itemize}
  \item Choose a rectangle in the tiling $m_i$ uniformly at random.
  \item Choose an edge $e$ of the four edges of the rectangle (left, right, top, or bottom) uniformly at random.
  \item If $e$ can be flipped to produce a new dyadic tiling in $\cT_{\ell}$, then flip the edge and let $m_{i+1}$ be the resulting tiling.
  \item If $e$ cannot be flipped to produce a new dyadic tiling in $\cT_{\ell}$, then choose a new edge at random and return to the previous step.
\end{itemize}

\end{definition}

There exists an isomorphism between the valid configurations of $\mathcal{B}_{\ell}$ and the set of dyadic tilings of rank $\ell$.%
That is, adding or removing a single gate from a partial configuration of $\mathcal{B}_{\ell}$ corresponds directly to a unique valid edge flip of the corresponding dyadic tiling (and vice-versa).

We will describe this isomorphism in a way that will build a visual intuition for it. First, we identify the empty configuration of $\cB_{\ell}$ on $2^{\ell}$ qubits with the tiling consisting entirely of horizontal cuts, $t_0$ s(the tiling consisting of $2^{\ell}$ horizontal rectangles stacked on top of each other). 

The tiling $t_0$ can be described by $2^\ell - 1$ horizontal cuts. We now describe these horizontal cuts as the disjoint union of smaller components we call $c$-segments. The following recursive procedure describes the segments:
\begin{itemize}
  \item Begin with an empty square and initialize a counter $c=\ell$.
  \item Make a horizontal cut through the square and subdivide this cut into $2^{c-1}$ equal length segments, each a $c$-segment.
  \item If $c > 1$, decrement $c\to c-1$ and repeat the previous step for \emph{each} of the two empty rectangles (using the same $c$ for each of the two) produced by the cut made in the last previous step. Otherwise, stop.
\end{itemize}

In the resulting representation of $t_0$, there are always $2^{\ell-1}$ $c$-segments for $c\in \{1,...,\ell\}$. The set of $2^{\ell-1}$ $c$-segments are distributed evenly across $2^c$ horizontal cuts of $t_0$ in groups of $2^{\ell-c-1}$. We call the resulting representation of $t_0$ a \emph{dressed} tiling. We give an example of $t_0$ for $\ell=3$ in Figure \ref{fig:vdi} below, with the number labels for the $c$-segments $c=1,2,3$ replaced with the colors blue, green, and red, respectively, for visual clarity:

\begin{figure}[h]
\begin{center}
\includegraphics[scale=.35]{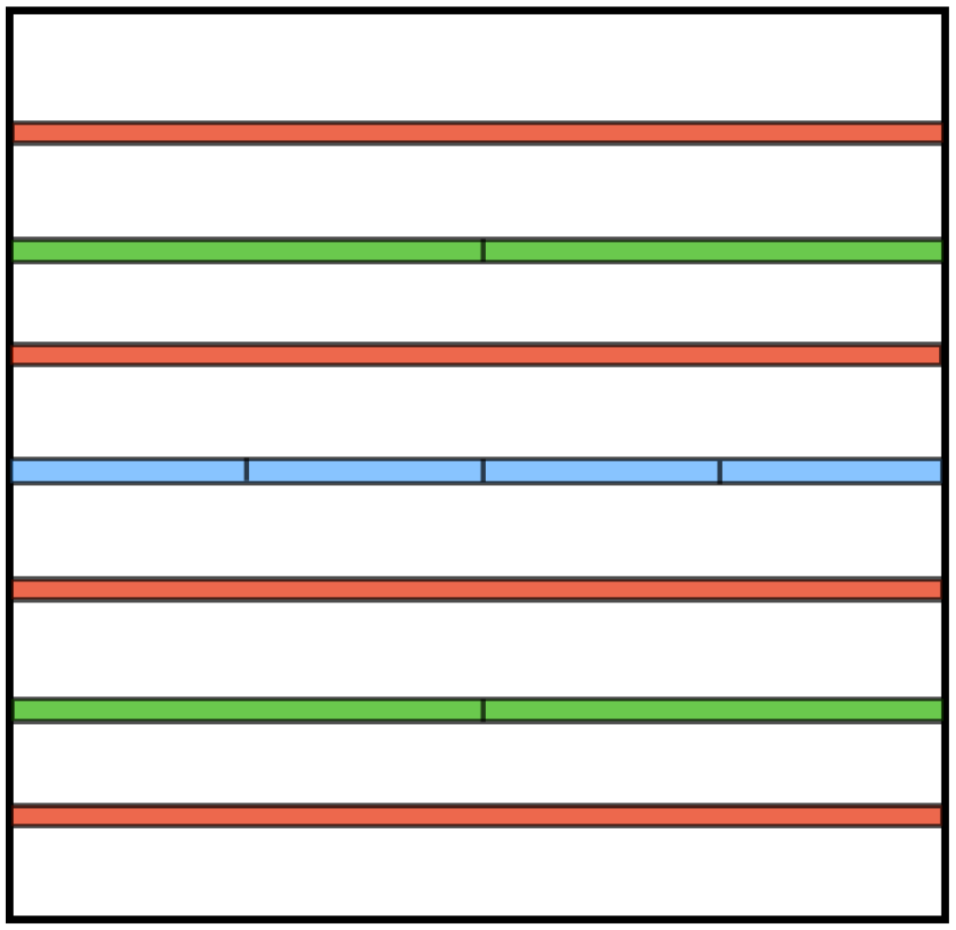}
\end{center}
\caption{The dressed tiling $t_0$ for $\ell=3$}
\label{fig:vdi}
\end{figure}

To understand the edge-flip Markov chain, we need to understand how to determine which edges can be flipped (flipping that edge will take you from one dyadic tiling to another), and which can not.

\begin{fact}[\cite{cannon_et_al:LIPIcs:2017:7583}]
  A \emph{flippable} edge of a tiling $t\in \cT_{\ell}$ is a \emph{$c$-segment} which, when considered alone, forms an entire edge of \emph{both} dyadic rectangles it borders.
\end{fact}

This fact tells us that starting from $t_0$, and flipping flippable $c$-segments, we obtain the edge-flip Markov chain.

Next, we establish a bijection between the gates of $\cB_{\ell}$ with the $c$-segments of the all horizontal tiling $t_0\in\cT_{\ell}$ in the following way: the set of $2^{\ell-1}$ $c$-segments corresponds to the set of $2^{\ell-1}$ gates in the $c$th layer of $\cB_{\ell}$. Consider the following procedure for assigning each particular gate to a particular $c$-segment:
\begin{itemize}
  \item Assign each $\ell$-segment of $t_0$, from left to right, to the gates in layer $\cL_{\ell}$, from top to bottom.
  \item For each gate in $\mathcal{L}_{\ell}$, identify the two gates in its past light-cone in $\mathcal{L}_{\ell-1}$ with the two nearest $\ell-1$-segments sitting above and below the $\ell$-segment in question.
  \item Continue this procedure recursively for the $\ell-1$-segments all the way down to the $1$-segments.
\end{itemize}

Figure \ref{fig:niso} illustrates this bijection for $\ell=3$.

\begin{figure}[h]
\begin{center}
\includegraphics[scale=.35]{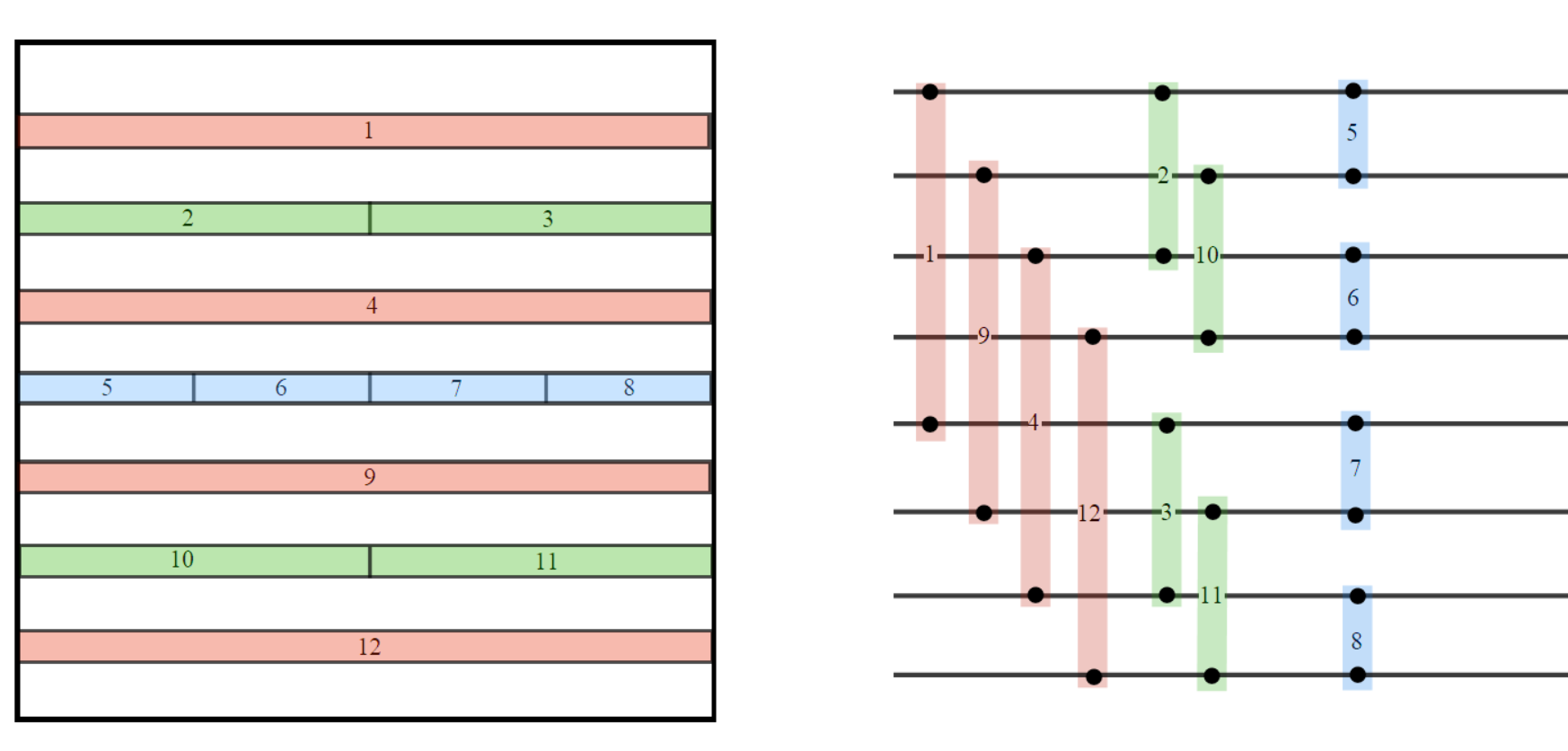}
\end{center}
\caption{This example for $\ell=3$ illustrates the correspondence between $c$-segments and their associated gates in $t_0$}
\label{fig:niso}
\end{figure}

Now, consider the following bijection between the partial configurations of $\cB_{\ell}$ and dyadic tilings $\cT_{\ell}$: For a tiling $t\in \cT_{\ell}$, identify it with the circuit in which the only gates applied are those for which the corresponding $c$-segments of the tiling are vertical. To see that the edge flip and circuit Markov chains are isomorphic, we need only see that the valid edge flips of a tiling correspond directly to the possible gate activations and deactivations of the corresponding partial configuration of $\cB_{\ell}$.

But this is clear by the gate to $c$-segment identification procedure described above: for a $c$-segment to be flippable from horizontal to vertical (vertical to horizontal), it must comprise an entire edge of both dyadic rectangles that it borders. This only happens once the two nearest $c-1$-segments ($c+1$-segments) -- the one directly above and the one directly below (directly left and directly right of) -- have been flipped to be vertical (horizontal). By the bijection procedure outlined above, those two $c-1$ ($c+1$) segments correspond to the gates that directly precede the gate corresponding the $c$-segment in its past (future) light-cone. Figure \ref{fig:flipz}, below, illustrates an example with $\ell=3$.

\begin{figure}[h]
\begin{center}
\includegraphics[scale=.35]{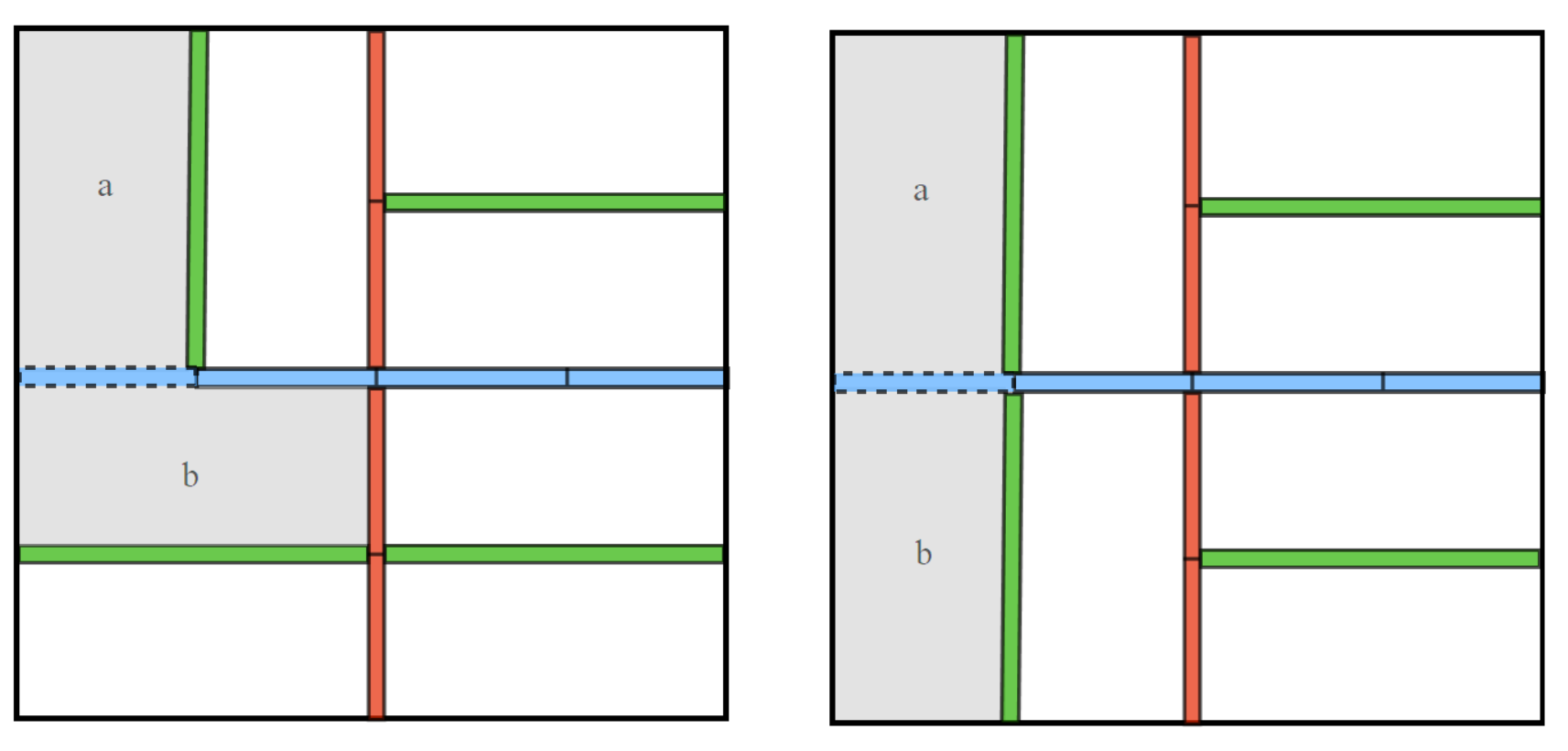}
\end{center}
\caption{Consider the left-most 3-segment in both of these dyadic tilings. In the left tiling, we see that it can not be flipped from horizontal to vertical because only one of the $2$-segments directly above and below it has been flipped to vertical. Consequently, this 3-segment is a complete edge of the dyadic rectangle a but not rectangle b. In the right tiling we see that this is remedied by flipping the remaining nearest 2-segment.}
\label{fig:flipz}
\end{figure}

\subsubsection{$HV$-Trees}\label{sec:bitonicIsomorphism}

In order to give a more complete understanding of the isomorphism between valid partial configurations of $\cB_\ell$ and $\cT_\ell$, we introduce an alternate representation of dyadic tilings called $HV$-trees \cite{randomdyadictilingsoftheunitsquare}.

Janson, Randall, and Spencer showed that the recursive description of a tiling gives an easy isomorphism to a graph called a $HV$-tree \cite{randomdyadictilingsoftheunitsquare}. Consider a complete binary tree of depth $\ell$ where each vertex is labeled either $H$ or $V$. There is a clear mapping from such trees to dyadic tilings: starting with the unit square and the root of the tree, draw a horizontal ($H$-cut) bisector or vertical ($V$-cut) bisector depending on the label of the root. Then recursively draw $H$- or $V$-cuts by the labels of the children on the two generated rectangles.

This mapping isn't a bijection, however. It is easy to see that
the following two $HV$-trees produce the same dyadic tiling. 
\begin{center}
\includegraphics[scale =.5]{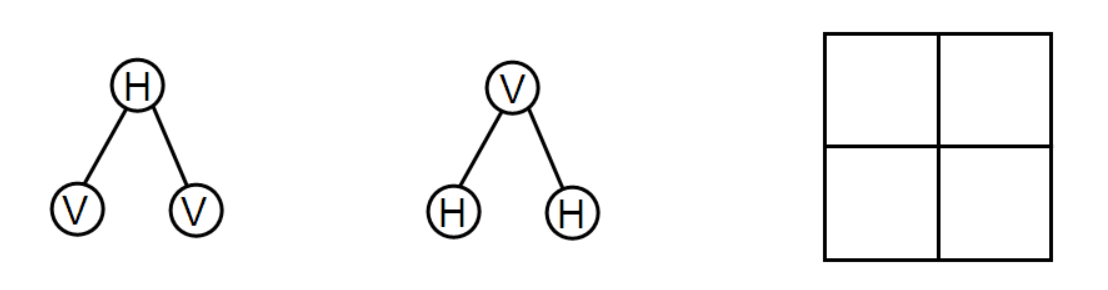}
\end{center} \cite{randomdyadictilingsoftheunitsquare} noticed that collisions only occur when there is a cross in the dyadic tiling, either a $H$-cut followed by 2 $V$-cut children or a $V$-cut followed by 2 $H$-cut children. By disallowing any $H$-vertex to have both children be $V$-vertexes, we obtain an isomorphism.

\begin{definition}[$HV$-trees]
An $HV$-tree of depth $\ell$ is a complete binary tree of depth $\ell$ with each vertex labeled either $H$ or $V$ and the restriction that no $H$-vertex has both children labeled $V$.
\end{definition}

\begin{theorem}[\cite{randomdyadictilingsoftheunitsquare}]
There is an isomorphism between $\cT_{\ell}$, the set of dyadic tilings of rank $\ell$, and the set of $HV$-trees of depth $\ell$.
\label{thm:dyadichv}
\end{theorem}

\begin{proof}
We previously described the mapping from $HV$-trees to dyadic tilings. For the other direction, look at the unit square. If there is a $V$-cut, label the root $V$ and proceed recursively. Otherwise, there must exist a $H$-cut (Theorem 1.1 of \cite{randomdyadictilingsoftheunitsquare}) label the root $H$ and proceed recursively. Note that by choosing a $V$-cut if both cuts exist ensures that the generated tree satisfies the $HV$-condition.
\end{proof}

We now prove the isomorphism between valid configurations of the bitonic block $\cB_{\ell}$ and $\cT_{\ell}$; it suffices to show the isomorphism between bitonic blocks and $HV$-trees. Recall Corollary \ref{cor:dictotomyofactivation}: Any valid configuration of $\cB_{\ell}$ must have every gate in $\cL_1$ activated or every gate in layer $\cL_{\ell}$ is not activated. Call valid configurations satisfying the first property $v$-configurations and call configurations satisfying the second property $h$-configurations.

Notice that given a $v$-configuration, we can recursively specify it by describing the configuration of the last $\ell - 1$ layers of $\cB_{\ell}$, which are isomorphic to $\cB_{\ell - 1}^{\otimes 2}$ (Corollary \ref{cor:subbitonic}). Similarly, given a $h$-configuration, we can recursively specify it by describing the configuration of the first $\ell - 1$ layers of $\cB_\ell$. A configuration that is both a $v$-configuration and $h$-configuration, can be recursively specified by the configuration of the middle $\ell - 2$ layers which are isomorphic to $\cB_{\ell - 2}^{\otimes 4}$.

\begin{theorem}
There is an isomorphism between the set of valid configurations of $\cB_{\ell}$ and the set of $HV$-trees of depth $\ell$. With Theorem \ref{thm:dyadichv}, this proves an isomorphism between the set of valid configurations of $\cB_{\ell}$ and $\cT_{\ell}$.
\label{thm:isomorphism}
\end{theorem}

\begin{proof}
Given a $HV$-tree, if the root is a $V$-vertex, we activate all the gates in $\cL_1$. We proceed recursively using the two children of the root to describe the configuration on the last $\ell - 1$ layers with each child describing the configuration on one of the blocks $\cB_{\ell - 1}$. Likewise, if the root is a $H$-vertex, we set all the gates in $\cL_{\ell}$ as not activated and proceed recursively on the first $\ell - 1$ layers.

Given a configuration, we know it must be a $v$-configuration or $h$-configuration. If it is a $v$-configuration, we set the root as a $V$-vertex and build the tree recursively with the blocks $\cB_{\ell - 1}$ in the last $\ell - 1$ layers describing the children sub-trees. Likewise, if it is a $h$-configuration, we set the root as a $H$-vertex and build the tree recursively with the blocks $\cB_{\ell - 1}$ in the first $\ell - 1$ layers describing the children sub-trees.

Notice, that by checking if a configuration is a $v$-configuration before checking if it is a $h$-configuration, we ensure the $HV$-tree property.
\end{proof}

\subsection{Indexing configurations}

Inspired by the uniform sampling algorithm for dyadic tilings of Janson, Randall, and Spencer \cite{randomdyadictilingsoftheunitsquare}, we adapt, using the isomorphisms in Theorem \ref{thm:isomorphism}, them to generate an indexing algorithm for valid configurations of a bitonic block.

\begin{theorem}
There exists an isomorphism between $[a_\ell] = \{1, 2, \ldots, a_\ell\}$ and the set of valid configurations of bitonic block $\cB_\ell$. Furthermore, both maps are efficiently calculable.
\end{theorem}

\begin{proof}
Theorem \ref{thm:numbconfigs}, tells us that these sets have the same magnitude. Corollary \ref{cor:gatenotactivated}, tells us that the number of $v$-configurations is $a_{\ell - 1}^2$ and the the number of $h$-configurations which are not $v$-configurations is $a_\ell - a_{\ell - 1}^2$.

Divide the set $[a_\ell]$ into $S_V =[a_{\ell-1}^2]$ and $S_H = a_{\ell-1}^2 + [a_\ell - a_{\ell - 1}^2]$. Given an index $i \in [a_\ell]$, we use these disjoint sets to decide whether the configuration is a $v$-configuration or $h$-configuration. If $i \in S_V$, then we set $\cL_1$ as activated and we express $i$ uniquely as $i_L a_{\ell - 1} + i_R$ and then recursively, using $i_L$ and $i_R$ as indices, decide the configurations on the two bitonic blocks $\cB_{\ell - 1}$ generating the last $\ell - 1$ layers. The case of $i \in S_R$ is a bit more subtle. Since we are choosing a $h$-configuration, we know its children cannot both be $v$-configurations. Therefore, we divide $S_R$ into 3 parts, corresponding to the bitonic blocks on the first $\ell - 1$ layers being both $h$-configurations, or one being a $h$-configuration and the other a $v$-configuration. It is not difficult to check that there are $(a_{\ell-1} - a_{\ell-2}^2)^2$ configurations with both children being $h$-configurations and $a_{\ell-2}^2(a_{\ell-1} - a_{\ell-2})^2$ for the other two cases. Now, we can proceed recursively.

For the other direction, given a valid configuration, we decide if it is a $v$-configuration or $h$-configuration. If a $v$-configuration, we recursively decide the index within $S_V$ and output it. Otherwise, we recursively decide the index within $S_R$, add $a_{\ell - 1}^2$ and output it.

We note that in both directions, the smaller bitonic blocks will involve permuted indexes. However, since Lemma \ref{lem:permutation} is efficient, this is not an issue.
\end{proof}

A similar proof holds for $\cB_{\ell}^{\times m}$ and $\cB_{\ell}^{\leftrightarrow m}$.

\begin{theorem}
There exists an isomorphism between $[a_\ell^{\times m}]$ and the set of valid configurations of architecture $\cB_{\ell}^{\times m}$. Likewise, there exists an isomorphism between $[a_\ell^{\leftrightarrow m}]$ and the set of valid configurations of circular architecture $\cB_{\ell}^{\leftrightarrow m}$. Furthermore, both maps are efficiently calculable.
\label{thm:efficientindex}
\end{theorem}

\begin{proof}
The proof is nearly identical except the initial partition of $[a_\ell^{\times m}]$ is based on the window that the configuration lies in. We note configurations of all windows except the last must be $h$-configurations. For the case of circular architecture, all windows are $h$-configurations.
\end{proof}

\end{document}